\documentclass[11pt,a4paper]{article}
\usepackage[a4paper]{geometry}
\usepackage{amssymb,latexsym,amsmath,amsfonts,amsthm}
\usepackage{amsmath}
\usepackage{amsthm}
\usepackage{amssymb}
\usepackage{amsfonts}
\usepackage{graphicx}
\usepackage{epstopdf}
\usepackage{hyperref}
\usepackage{color}
\def\Im {\mathop{\rm Im}\nolimits}
\def\arg {\mathop{\rm arg}\nolimits}
\def\Re {\mathop{\rm Re}\nolimits}

\def\tr {{\rm tr}}

\newtheorem{thm}{Theorem}
\newtheorem{cor}{Corollary}
\newtheorem{lem}{Lemma}
\newtheorem{prop}{Proposition}

\numberwithin{equation}{section}

\newtheorem{rem}{Remark}
\numberwithin{equation}{section}

\newcounter{comment}
\numberwithin{equation}{section}

\title{{ Gap probability of the circular unitary ensemble  with a Fisher-Hartwig singularity  and the coupled Painlev\'{e} V system }}

\author{Shuai-Xia Xu$^{\,1}$  and  Yu-Qiu Zhao$^{\,2}$}

\date{
\hbox{\small \emph{$^1$ Institut Franco-Chinois de l'Energie Nucl\'{e}aire, Sun Yat-sen University, GuangZhou 510275,  China }}
 \hbox{\small
\emph{$^2$  Department of Mathematics, Sun Yat-sen University, GuangZhou 510275,
China}}}

\begin{document}

\maketitle

\noindent \hrule width 6.27in\vskip .3cm

\noindent {\bf{Abstract } }
We consider the circular unitary ensemble with a Fisher-Hartwig singularity of both jump type and root type at $z=1$. A rescaling of the ensemble at the Fisher-Hartwig singularity leads to the  confluent hypergeometric kernel. By studying the asymptotics of the Toeplitz determinants, we show that the probability of there being  no eigenvalues in a symmetric arc  about the singularity on the unit circle for a random matrix in the ensemble can be explicitly evaluated via an integral of the Hamiltonian of the coupled Painlev\'{e} V system in dimension four.  This leads to a Painlev\'e-type representation of the   confluent hypergeometric-kernel  determinant. Moreover, the large gap asymptotics, including the constant terms, are  derived by evaluating the total integral of the Hamiltonian. In particular, we reproduce the large gap asymptotics of the confluent hypergeometric-kernel determinant obtained by Deift, Krasovsky and Vasilevska,  and the sine-kernel determinant as a special case, including the constant term conjectured earlier by Dyson.


\vspace{.2in} \noindent {\it {Keywords:Random matrices; Tracy-Widom distribution; Painlev\'{e} equations; Riemann-Hilbert problem; constant problem; Deift-Zhou method.}}
\noindent \hrule width 6.27in\vskip 1.3cm

\tableofcontents

\section{Introduction and statement of results} \indent\setcounter{section} {1}

One of the most celebrated results in random matrix theory is the universality of local statistics of eigenvalues for large random matrices.  Well-known universality classes  arise in the unitary invariant random matrices in the bulk, and at the soft edges and hard edges   \cite{abd, F,m}.  Accordingly, for large unitary random matrices, the eigenvalues near a regular point in the bulk, near the soft edge and the hard edge, form determinantal point processes  associated with the sine kernel, Airy kernel and Bessel kernel, respectively; see  \cite{m} and \cite[Chapter\;6]{abd}.  These three determinantal point processes are now well understood. A significant feature   is   that the gap probabilities in these determinantal point processes can be expressed in terms of the Painlev\'e transcendents  and their associated Hamiltonians; see \cite{FW-2001,FW-2002, JMMS,TW, TW2}.  The large gap asymptotics, including the non-trivial constant terms, are also worked out; see \cite{DIK08, DIKZ,diz, E06, E10, kr}.

Another basic universality class in the unitary  random matrix ensembles is the confluent hypergeometric-kernel  determinantal point process, having been used to  describe  the statistics of eigenvalues  near a Fisher-Hartwig singularity in the spectrum; see  \cite[Chapter\;6]{abd}.   For example, we consider the following  Gaussian unitary ensemble  with a Fisher-Hartwig singularity at the origin, defined by the
 joint probability density function
 \begin{equation}\label{GUE-with FH}p_n(\lambda_1,...,\lambda_n)=\frac{1}{Z_n}
\prod_{j=1}^n e^{- \lambda_j^2}|\lambda_j|^{2\alpha}\chi(\lambda_j)\prod_{j<k}(\lambda_j-\lambda_k)^2,\end{equation}
where
$Z_n$ is the normalization constant, the parameter  $\alpha>-\frac{1}{2}$, and the function
  \begin{equation*}
  \chi(x)=\left\{ \begin{array}{ll}
                e^{i\pi\beta},&  x<0, \\
               e^{-i\pi\beta},&  x>0,
        \end{array}\right.  \end{equation*}
with $i\beta\in \mathbb{R}$.  The simultaneous occurrence of a  jump type singularity as $\chi(x)$, and  a root type singularity $|x|^{2\alpha}$,  indicates a Fisher-Hartwig singularity at the origin; cf. \cite{DIK1}.
 For $\alpha\in\mathbb{N} $ and $\beta=0$, the ensemble  is a degenerate GUE of size $ n + \alpha$ with $\alpha$   eigenvalues  located at the origin, which was considered by Chen and Feigin in  \cite{cf}. If $\alpha=0$ and $i\beta\in\mathbb{R}$, the ensembles can be realized as a  conditional GUE.  More precisely,  one obtains a thinned process of eigenvalues by deleting  each of the eigenvalues of the original  GUE  independently with probability $p=e^{-2i\beta}$ for $i\beta>0$ (while $p=e^{2i\beta}$ if $i\beta<0$); see \cite{bp04,bp06}. One may interpret the remaining and removed eigenvalues
  as observed and unobserved particles, respectively.  Knowing the information that the particles observed are all negative when $i\beta>0$ (or positive if $i\beta<0$),  the eigenvalue distribution of the conditional GUE  is given  by \eqref{GUE-with FH}; see \cite{bci,cc}. The asymptotics of the partition function and the recurrence coefficients of the  orthogonal polynomials  associated with the Gaussian weight with a Fisher-Hartwig singularity at the origin have been obtained by Its and Krasovsky in \cite{ik}. The asymptotic  analysis of the orthogonal polynomials therein  implies that the scaling limit of the correlation kernel near the origin can be expressed in terms of the confluent hypergeometric kernel. For $\beta=0$, the kernel is reduced to the Bessel kernel of the first kind,   obtained earlier by Kuijlaars and Vanlessen in \cite{kv}.

  The confluent hypergeometric kernel also arises in the Dyson circular unitary ensemble with a Fisher-Hartwig singularity
  \begin{equation}\label{def:CUE-with FH}p_n(\theta_1,\cdots,\theta_n)=\frac{1}{Z_n'}
\prod_{j=1}^n |e^{i\theta_j}-1|^{2\alpha}e^{i\beta(\theta_j -\pi)}\prod_{j<k}  |e^{i\theta_j}-e^{i\theta_k}  |^2, \quad \theta_1,...,\theta_n\in(0, 2\pi), \end{equation} with real parameter $\alpha>-1/2$ and  pure imaginary parameter $\beta$.
Then, it is direct to see that the  weight function $|e^{i\theta}-1|^{2\alpha}e^{i\beta(\theta -\pi)}$ is positive for $\theta\in(0,2\pi)$ and possesses  a Fisher-Hartwig singularity of both root type and jump type at $\theta=0$. With the normalization constant $Z_n'$,    \eqref{def:CUE-with FH} is made a probability measure. We denote by $C_{t}$ an arc of the  unit circle, oriented counterclockwise, that is,
 \begin{equation}\label{def: arc}
 C_{t}=\{e^{i\theta}: t \leqslant \theta \leqslant 2\pi-t \} \quad\mbox{for}\quad 0< t<\pi,
 \end{equation}
 and define the following positive weight function $w(z;t)$ on $C_{t}$
  \begin{equation}\label{def: weight-circle}
w(z;t)=|z-1|^{2\alpha}z^{\beta}e^{-\pi i\beta},  \quad z=e^{i\theta}\in C_{t}, \end{equation}
with real parameter $\alpha>-1/2$ and  pure imaginary parameter $\beta$.
Let $D_n(t)$ denote the Toeplitz determinant with the  weight function \eqref{def: weight-circle}, namely
\begin{align}\label{def: D-n}
D_n(t)&=\det\left(\frac{1}{2\pi}\int_{0}^{2\pi}w(e^{i\theta};t)e^{-(j-k)i\theta}d\theta\right)_{j,k=0}^{n-1}\nonumber\\
&=\frac{1}{(2\pi)^nn!}\int_{0}^{2\pi}\cdots\int_{0}^{2\pi}\prod_{j=1}^n w(e^{i\theta_j};t)\prod_{j<k} | e^{i\theta_j}-e^{i\theta_k} |^2 \prod_{j=1}^nd\theta_j, \end{align}
where the second expression is known as  Heine's multiple integral representation \cite{Szego}.
Then, the probability that all eigenvalues of the random matrix in the ensemble \eqref{def:CUE-with FH} are contained in the arc $C_{t}$   can be expressed  as
\begin{equation}\label{def: gap pro-n}
\mbox{Pro}(\{\theta_j\in(t,2\pi-t): 1\leqslant j\leqslant n\})=\frac{D_n(t)}{D_n(0)}. \end{equation}
The circular ensemble  \eqref{def:CUE-with FH} was introduced by Deift, Krasovsky and Vasilevska
 in \cite{dkv} to derive the large gap asymptotics for the Fredholm determinant of the  confluent hypergeometric kernel. It was shown that the scaling limit of the correlation kernel near $z=1$ is the confluent hypergeometric kernel
 \begin{equation}\label{def: CFH-kernel}
 K^{(\alpha,\beta)}(u,v)=\frac{1}{2\pi i}\frac{\Gamma(1+\alpha+\beta)\Gamma(1+\alpha-\beta)}{\Gamma(1+2\alpha)^2}
\frac{A(u)B(v)-A(v)B(u)}{u-v}, \end{equation}
where
 \begin{equation}\label{def: A,B}
 A(x)=\chi(x)^{1/2}|2x|^{\alpha}e^{-ix}\psi(1+\alpha+\beta,1+2\alpha,2ix), \quad B(x) =\overline{A(x)}, \end{equation}
 the confluent hypergeometric function \cite{as}
 \begin{equation}\label{def: phi-CHF}
\psi(a,b,z)=1+\sum_{k=1}^{\infty}\frac{a(a+1)\cdots (a+k-1)}{b(b+1)\cdots(b+k-1)}\frac{z^k}{k!}, \end{equation}
and $\chi(x)=\left\{ \begin{array}{ll}
                e^{i\pi\beta},&\quad x<0, \\
               e^{-i\pi\beta}, &\quad x>0,
        \end{array}\right.$ as defined in \eqref{GUE-with FH}. If $\alpha=\beta=0$, the kernel is reduced to the classical sine kernel,
        \begin{equation}\label{def: sin-kernel}
 K^{(0,0)}(u,v)=K_{\sin}(u,v)=
\frac{\sin (u-v)}{\pi(u-v)} .\end{equation}
Thus, for $n$ large, the gap probability \eqref{def: gap pro-n} can be expressed in terms of the Fredholm determinant
\begin{equation}\label{def: determinant of CHF }
\lim_{n\to \infty}\frac{D_n(\frac{2t}{n})}{D_n(0)}=\det(I-K^{(\alpha,\beta)}_t), \end{equation}
where $K^{(\alpha,\beta)}_t$ is the integral operator on $L^2(-t,t)$, $t\in(0,\pi)$,  with the confluent hypergeometric kernel \eqref{def: CFH-kernel}; see \cite[Lemma 6]{dkv}.   In \cite{dkv}, the large-$n$ asymptotics of the logarithmic derivative of  the Toeplitz determinant  $D_n(t_n)$ were derived, uniformly
for the parameter $2t_0/n  < t _n<\pi$
with   $t_0$   sufficiently large. The
boundary condition for $D_n(t)$ as $t\to\pi$ was obtain by comparing it with the Hankel determinant associated with the Legendre polynomials, of which the large-$n$ asymptotics was derived before in \cite{W}.
 Therefore, asymptotic approximation  of the Toeplitz determinant,  uniformly for  $2t_0/n  < t _n<\pi$ with $t_0$ big enough, was obtained by
 integration. Then, recalling the relation \eqref{def: determinant of CHF }, the large gap asymptotics of the determinant of confluent hypergeometric kernel as $t\to \infty$ was readily worked out. The  constant term, namely the Widom-Dyson constant, was expressed in terms of the Barnes $G$-functions therein.

 Recently, the studies of the asymptotics of Toeplitz determinant  with Fisher-Hartwig singularities are of
 great interest. In \cite{DIK1}, the asymptotics of Toeplitz determinant  for a general weight function on the circle with several Fisher-Hartwig singularities were derived and   the Barnes $G$-functions were involved.  A nice review of the history of the asymptotics of Toeplitz determinants, with   applications in the Ising model,   can be found in \cite{DIK2}. In \cite{cik},  Claeys, Its and Krasovsky studied the emergence of a Fisher-Hartwig singularity for the Toeplitz determinant. They obtained transition type asymptotics of the  Toeplitz determinant between the  case with no singularity and the case with one  Fisher-Hartwig singularity. Later, the transition asymptotics as two Fisher-Hartwig  singularities merge into  one Fisher-Hartwig  singularity was considered by  Claeys and Krasovsky in \cite{ck}.  Both of the transition asymptotics are described by the Painlev\'e V transcendents with different
 boundary behaviors. Applications in two-dimensional Ising models and random matrix theory were also considered  therein.

From the relation \eqref{def: determinant of CHF },  the large-$n$ asymptotics of the  Toeplitz determinant
$D_n(\frac{2t}{n})$ defined in \eqref{def: D-n} for bounded $t$ will give  the determinant of the confluent hypergeometric kernel on $(-t,t)$. Thus one obtains in the unitary ensembles \eqref{def:CUE-with FH} the  probability that the symmetric interval around the  Fisher-Hartwig singularity is free of eigenvalues.  For $\beta=0$, the  kernel was reduced to the Bessel kernel of the first kind and the determinant can be expressed  in terms of the Hamiltonian associated with the Painlev\'e V equation; see \cite{WF}.  It was observed by Borodin and Deift  \cite{bd} that the determinant of the confluent hypergeometric kernel, restricted to one side interval $(0,t)$, should be related to a solution to the Painlev\'e V equation.
Since the confluent hypergeometric kernel is not symmetric, thus the determinant on the interval $(-t,t)$  can not be obtained from that on the one side interval $(0,t)$.  Moreover, it is interesting to note that as $\frac{2t}{n}\to 0$, in \eqref{def: weight-circle}, the two  Fisher-Hartwig singularities of jump type at $\theta=\frac{2t}{n}$ and $\theta= 2\pi-\frac{2t}{n}$ are merging together  to the root type singularity $|e^{i\theta}-1|^{2\alpha}$. The result of such a process is a new Fisher-Hartwig singularity of both jump-type and the root-type singularity in the weight function in \eqref{def:CUE-with FH}.  The transition is different from the ones considered in \cite{cik,ck} as mentioned above.  It is of interest to study the new transition asymptotics of the  Toeplitz determinant from this point of view.

More precisely, this paper is devoted to the study of the determinant of the confluent hypergeometric kernel on $(-t,t)$ both for finite $t$ and the asymptotics as $t\to\infty$. Firstly, we  derive the  large-$n$ asymptotics of the Toeplitz determinant $D_n(\frac{2t}{n})$ defined in \eqref{def: D-n} for bounded $t$.  Using the relation \eqref{def: determinant of CHF } with the Toeplitz determinant, we obtain an integral representation of the determinant in terms of the Hamiltonian of a coupled Painlev\'{e} V system. We further study the existence of solutions to the system and the asymptotics of the solutions.  Secondly,  we evaluate the total integral of the  Hamiltonian, which then gives us the  large gap asymptotics of the determinant, including the constant terms.

 The Hamiltonian  integral is  closely related to the Jimbo-Miwa-Ueno tau-function, since the logarithmic derivative of the tau-function is
 the Hamiltonian; see \cite{JMU}.  The tau-functions of Painlev\'e equations play important roles in random matrix theory, two-dimensional quantum gravity,  Ising models  and many other areas of  mathematical physics; see \cite{fikn,F,ilp,TW}.  Usually, it is a hard problem to determine the
 constant term in the asymptotics of the tau-functions.    Recently,  Its, Lisovyy, Prokhorov  and their coauthors have developed a new method to evaluate the asymptotics of the tau-functions, including the constant terms; see  \cite{ilp, ilt,ip}. A crucial observation is that
 the Hamiltonian is equal to the classical action differential up to a total differential.  A nice property of the action differential is that the derivative of it, with respect  to the parameters of the equation or the Stokes multipliers,  is again a total differential.
 Thus, the evaluation of the Hamiltonian integral is eventually reduced to the calculation of an action integral for special parameters and Stokes multipliers.
 For certain parameters and Stokes multipliers, the Hamiltonian may admit
 solutions in terms of classical special functions, and its integral can be evaluated.  Quite recently, the asymptotics of the  tau-functions for classical Painlev\'e equations are successfully  obtained by applying this method; see \cite{bip,ilt13,ilp, ilt,ip,LR}. More general Painlev\'e-type  equations are also considered. For example, in \cite{XD}, Xu and Dai obtained the asymptotics of the Hamiltonian integral of the coupled Painlev\'e II system with an application in the large gap asymptotics of the Painlev\'e II and Painlev\'e XXXIV determinants.  In the present paper, however, we find we are in a different position, to use  the approach to  evaluate   the tau-function of non-classical Painlev\'e equations, namely the  coupled Painlev\'e V system in dimension four in the list of  \cite{kns-1}.


\subsection{The coupled Painlev\'{e} V system}
To state our main results, we need to introduce the    coupled Painlev\'e V system in dimension four, which is a special 4-dimensional Painlev\'e-type
equation; see \cite{k-3,kns-1}. It is convenient to express the  coupled Painlev\'e V system
in the    following Hamiltonian form
\begin{equation}\label{int:H-equation}
  \frac{d v_k}{ds}=\frac{\partial H}{\partial u_k} ,\quad  \frac{d u_k}{ds}=-\frac{\partial H}{\partial v_k}, \quad k=1,2,\end{equation}
 where the Hamiltonian $H=H(u_1,v_1,u_2,v_2,s;\alpha,\beta)$ is defined by
\begin{equation}\label{int:H}
sH=\frac{s}{2}H_{V}(u_1,v_1,s/2; \alpha,\beta) -\frac{s}{2}H_{V}( u_2,v_2,-s/2; \alpha,\beta)  +u_1u_2(v_1+v_2)(v_1-1)(v_2-1) ,     \end{equation}
and $ H_V$ is the classical Hamiltonian for the Painlev\'e V equation, such that
\begin{equation}\label{int:H-v}
sH_{V}(u,v,s; \alpha, \beta)=u^2v(v-1)^2-suv-\alpha u (v^2-1)-\beta u (v-1)^2;\end{equation}
see \cite{FW-2002,JM}.
  In terms of $u_k$ and $v_k$, we write \eqref{int:H-equation} as  the following system of equations
 \begin{equation}\label{int:cpv}
 \left\{ \begin{array}{l}
                 s\frac{d u_1}{ds}=\frac s{2}u_1-u_1^2(v_1-1)(3v_1-1)-u_1u_2(v_2-1)(2v_1+v_2-1)+2(\alpha+\beta)u_1v_1-2\beta u_1,\\
                    s\frac{d u_2}{ds}=-\frac{s}{2}u_2-u_2^2(v_2-1)(3v_2-1)-u_1u_2(v_1-1)(v_1+2v_2-1)+2(\alpha+\beta)u_2v_2-2\beta u_2 ,\\
                   s\frac{d v_1}{ds}=-\frac{  s}{2}v_1+2u_1v_1(v_1-1)^2+u_2(v_1+v_2)(v_1-1)(v_2-1)-\alpha(v_1^2-1)-\beta (v_1-1)^2,\\
                      s\frac{d v_2}{ds}=\frac{s }{2}v_2+2u_2v_2(v_2-1)^2+u_1(v_1+v_2)(v_1-1)(v_2-1)-\alpha(v_2^2-1)-\beta (v_2-1)^2.
                                           \end{array}
\right. \end{equation}
If the parameter $\beta=0$, we have the symmetry relations
\begin{equation}\label{H-symm}
 u_2=-u_1v_1^2\quad \mbox{and} \quad v_2=1/v_1,
\end{equation}
as shown later in Section \ref{subsec:special function-sol}.
Accordingly, the  system of equations  \eqref{int:cpv} is reduced to
  \begin{equation}\label{eq:mpv}
 \left\{ \begin{array}{ll}
                 s\frac{d u_1}{ds}=\frac{1}{2}su_1-2u_1^2v_1(v_1^2-1)+2\alpha u_1v_1 ,\\
                s\frac{d v_1}{ds}=-\frac{1}{2}sv_1+u_1(v_1^2-1)^2-\alpha(v_1^2-1).
        \end{array}
\right. \end{equation}
From the second equation in  \eqref{eq:mpv}, we see that $u_1$ can be expressed in terms of $v_1$ and $v_1'$. Using this expression to delete $u_1$ from the system, we find the second order nonlinear equation for $v_1$
 \begin{equation}\label{eq:mpv-v-1}
 \frac{d^2v_1}{ds^2}-\frac{2v_1}{v_1^2-1}\left(\frac{dv_1}{ds}\right)^2+\frac{1}{s}\frac{dv_1}{ds}+\frac{v_1}{2s}+\frac{v_1(v_1^2+1)}{4(v_1^2-1)}+\alpha\frac{v_1^2+1}{2s}=0.\end{equation}
Applying  the Mobius transformation $v=\frac{v_1+1}{v_1-1}$ to \eqref{eq:mpv-v-1}, we obtain the Painlev\'{e} III equation
 \begin{equation}\label{eq:piii}
\frac{d^2v}{ds^2}-\frac{1}{v}\left(\frac{dv}{ds}\right)^2+\frac{1}{s}\frac{dv}{ds}-\frac{1}{2s}
\left(\left(\alpha+\frac{1}{2}\right )v^2+\left(\alpha-\frac{1}{2}\right )\right)-\frac{1}{16}v^3+\frac{1}{16v}=0.
  \end{equation}

{\begin{rem}\label{Rem:H-special par} If the parameters $\alpha=\beta=0$,  we obtain the Painlev\'e V  equation from \eqref{eq:mpv-v-1}
 after the transformation $q=v_1^2$, that is,
 \begin{equation}\label{eq:pv}
 \frac{d^2q}{ds^2}-\left (\frac{1}{q-1}+\frac{1}{2q}\right )\left(\frac{d q}{ds}\right)^2+\frac{1}{s}\frac{d q}{ds}+
\frac{q}{s}+\frac{q(q+1)}{2(q-1)}=0.
\end{equation}
Moreover, using the symmetry \eqref{H-symm}, the Hamiltonian \eqref{int:H} is reduced to the classical Hamiltonian for the Painlev\'e V equation after elementary transformations,
\begin{equation}\label{H-special parameter}
sH(u_1,v_1,u_2,v_2,s;0,0)=u_1^2(v_1^2-1)^2-su_1v_1=p^2q(q-1)^2-spq,
\end{equation}
where $q=v_1^2$ and $p=u_1/v_1$.  Denote
\begin{equation}\label{eq: sigma-PV}
\sigma(s)=sH(u_1,v_1,u_2,v_2,s;0,0),
\end{equation}
we get  from the Hamilton equations $q'=\frac{\partial H}{\partial p}$ and $p'=-\frac{\partial H}{\partial q}$ the Jimbo-Miwa-Okamoto $\sigma$-form of the Painlev\'e V equation
\begin{equation}\label{equation-sigma}
\left(s\sigma''\right)^2-\left(\sigma-s\sigma'+4\sigma'^2\right)\left(\sigma-s\sigma'\right)=0;
\end{equation}
see \cite{FW-2002, JM}.
\end{rem}}

Quite recently, the coupled Painlev\'e systems have found interesting applications in random matrix theory. For example, a different coupled Painlev\'e V system was used in the studies  of the Bessel process   and  the distribution of the ratio between the second smallest and the smallest eigenvalue in the Laguerre unitary ensemble \cite{acz,cd}. The coupled Painlev\'e II system was applied  in the studies of the Airy point process and the determinants of the Painlev\'e II and Painlev\'e  XXXIV kernels  \cite{cld,XD}.  Moreover, the coupled Painlev\'e III system was involved in the Painlev\'e-type representation of the determinants of a Painlev\'e III kernel, which arose  in the studies of a singularly perturbed Laguerre unitary ensemble \cite{dxz, xdz}.

We have the following existence and asymptotic results for the coupled Painlev\'e V system \eqref{int:cpv}. The solutions play a central role in our
integral representation of the gap probability.
  \begin{thm}\label{thm-asy}  For pure imaginary parameter $\beta$ and $\alpha>-1/2$,  there exist solutions to the coupled Painlev\'e V system \eqref{int:cpv}
  with the following asymptotic behavior
 \begin{equation}\label{thm: u-1-asy-small s}
u_1(s)=-\frac{\Gamma(1+\alpha+\beta)\Gamma(1+\alpha-\beta)}{2\pi i\Gamma(1+2\alpha)^2}e^{-\pi i\beta}\left(\frac {|s|}2\right )^{2\alpha}\left[1+O\left(s^{2\alpha+1}\right )+O(s)\right], \quad is\to0^+,
\end{equation}
\begin{equation}\label{thm: v-1-asy-small s}
v_1(s)=1+O\left (s^{2\alpha+1}\right )+O(s),  \quad is\to 0^+,
\end{equation}
\begin{equation}\label{thm: u-2-asy-small s}
u_2(s)=\frac{\Gamma(1+\alpha+\beta)\Gamma(1+\alpha-\beta)}{2\pi i\Gamma(1+2\alpha)^2}e^{\pi i\beta}\left (\frac{|s|}2\right )^{2\alpha}\left[1+O\left (s^{2\alpha+1}\right )+O(s)\right], \quad is\to 0^+,
\end{equation}
\begin{equation}\label{thm: v-2-asy-small s}
v_2(s)=1+O\left (s^{2\alpha+1}\right )+O(s),  \quad is\to 0^+,
\end{equation}
\begin{equation}\label{thm: u-1-asy-large s}
u_1(s)=\frac{i}{8}s\left(1+O\left(  1/s\right)\right ), \quad is\to+\infty,
\end{equation}
\begin{equation}\label{thm: v-1-asy-large s}
v_1(s)=i+O(1/s), \quad is\to+\infty,
\end{equation}
\begin{equation}\label{thm: u-2-asy-large s}
u_2(s)=\frac{i}{8}s(1+O(1/s)), \quad is\to+\infty,
\end{equation}
\begin{equation}\label{thm: v-2-asy-large s}
v_2(s)=-i+O(1/s), \quad is\to+\infty.
\end{equation}
Moreover, the Hamiltonian $H(s;\alpha,\beta)=H(u_1,v_1,u_2,v_2,s;\alpha,\beta)$, associated with these $u_k(s)$  and $v_k(s)$,   is  pole-free
  for $ s\in -i(0, +\infty)$  and has the asymptotic behavior
 \begin{equation}\label{thm: H-asy-small s}
 H(s;\alpha,\beta)=\frac{\Gamma(1+\alpha+\beta)\Gamma(1+\alpha-\beta)\cos(\pi \beta)}{i\pi 2^{2\alpha+1}(2\alpha+1)\Gamma(1+2\alpha)^2}|s|^{2\alpha}+O\left (s^{2\alpha+1}\right ), \quad is\to 0^+,
 \end{equation}
 and
 \begin{equation}\label{thm: H-asy-large s}
H(s;\alpha,\beta)=\frac{s}{16}+\frac{i}{2}\alpha-\left(\alpha^2-\beta^2+\frac{1}{4}\right)\frac{1}{s}+O\left(\frac 1{s^2}\right), \quad is\to+\infty.
\end{equation}
\end{thm}

{\begin{rem}\label{Rem:special solution} For the parameters $\alpha=1/2 $ and $\beta=0$,   the coupled Painlev\'e V system \eqref{int:cpv} admits solutions in terms of classical special functions; see Section \ref{subsec:special function-sol} for details. In particular, we have
 \begin{equation}\label{rem: H-Bessel solution}
H(u_1,v_1,u_2,v_2,s;1/2,0)=\frac{s}{16}+\frac{i}{4}\frac{I_0'(is/4)}{I_0(is/4)},
\end{equation}
where $I_0(s)$ is the modified Bessel function of order zero.  It is noted that for this special solution, the integral of the Hamiltonian can be easily evaluated.
\end{rem}}
\subsection{Gap probability and Painlev\'e-type formulas}
 The  gap probability for the eigenvalues of
  the circular unitary ensemble  with a Fisher-Hartwig singularity can be expressed in terms of the Toeplitz determinant associated with the weight function  \eqref{def: weight-circle}.
Via deriving  the asymptotics of the Toeplitz determinant,  we give an explicit expression  for the gap probability as an integral of the Hamiltonian to  the coupled Painlev\'{e} V system in dimension four.
\begin{thm}\label{thm-gap-probability}  For pure imaginary parameter $\beta$ and $\alpha>-1/2$, as $n\to \infty$ and $t\to 0^{+}$ in a way such that $nt$ is bounded, we have the asymptotic approximation of the Toeplitz determinant associated with the weight function \eqref{def: weight-circle}
\begin{equation}\label{eq:D-n-approx}
\frac{D_n(t)}{D_n(0)}=\exp\left(\int_0^{-2int}H(\tau;\alpha,\beta)d\tau+O\left (\frac 1 n\right )\right),\end{equation}
where $H(s;\alpha,\beta)$ is the Hamiltonian for the coupled Painlev\'e V system with the properties specified in Theorem \ref{thm-asy} and the error term is uniform for $nt$  bounded.
\end{thm}

It then follows from \eqref{def: gap pro-n} and  Theorem \ref{thm-gap-probability}  that we have the Painlev\'e-type representation of the gap probability for  the circular ensemble with a  Fisher-Hartwig singularity \eqref{def:CUE-with FH}  as $n\to\infty$
\begin{equation*}
\mbox{Pro}(\{\theta_j\in(t,2\pi-t): 1\leqslant j\leqslant n\})=\exp\left(\int_0^{-2int}H(\tau;\alpha,\beta)d\tau+O\left (\frac 1 n\right )\right). \end{equation*}
The gap probability can also be expressed
in terms of the Fredholm determinant of the confluent hypergeometric kernel, as given in
\eqref{def: determinant of CHF }, such that
\begin{equation*}
\lim_{n\to\infty}\frac{D_n(\frac{2s}{n})}{D_n(0)}= \det\left(I-K_s^{(\alpha,\beta)}\right ),
\end{equation*}
 where $K_s^{(\alpha,\beta)}$ is the operator with the confluent hypergeometric kernel acting on $L^2(-s,s)$; see \cite{dkv}.
As a direct application of the above theorem, we obtain an explicit representation of the   determinant of the confluent hypergeometric kernel.
 \begin{thm}\label{thm: gap pro}   For pure imaginary parameter $\beta$ and $\alpha>-1/2$, let $K_s^{(\alpha,\beta)}$ be
  the operator with the confluent hypergeometric kernel $K^{(\alpha,\beta)}(u,v)$ in \eqref{def: CFH-kernel} acting on $L^2(-s,s)$, $s>0$, we have
\begin{equation}\label{eq:gap pro}
\det\left(I-K_s^{(\alpha,\beta)}\right )=\exp\left(\int_0^{-4is}H(\tau;\alpha,\beta)d\tau\right),\end{equation}
where $H(s;\alpha,\beta)$ is the Hamiltonian for the coupled Painlev\'e V system with the properties specified in Theorem \ref{thm-asy}.
\end{thm}

For the special parameters $\alpha=\beta=0$,  the confluent hypergeometric kernel is reduced to the classical sine kernel  \eqref{def: sin-kernel}.  Applying Theorem \ref{thm-asy}, Theorem \ref{thm: gap pro}
and  Remark \ref{Rem:H-special par}, we recover the following  Painlev\'e V representation of the determinant of the sine kernel, which was obtained
earlier by Jimbo, Miwa, M\^ori and Sato \cite{JMMS}.
\begin{cor}\label{cor:gap-pro}   Let $K_s^{(0,0)}$ be
  the operator with sine kernel  \eqref{def: sin-kernel} acting on $L^2(-s,s)$, $s>0$, we have
\begin{equation}\label{eq:gap pro-sine-sigma}
\det\left(I-K_s^{(0,0)}\right )=\exp\left(\int_0^{s}\frac{\sigma_V(\tau)}{\tau}d\tau\right),\end{equation}
where $\sigma_V(s)$ is the solution of the equation
\begin{equation}\label{equation-sigma-V}
\left(s\sigma_V''\right)^2+4\left(4\sigma_V-4s\sigma_V'-{\sigma_V'}^2\right)\left(\sigma_V-s\sigma_V'\right)=0,
\end{equation}
characterized by the asymptotic behavior
\begin{equation}\label{asy: sigma-V}
\sigma_V(s)\sim -\frac{2}{\pi}s, \quad s\to 0^+ \quad  \mbox{and} \quad \sigma_V(s)\sim -s^2-\frac{1}{4}, \quad s\to +\infty.
\end{equation}

\end{cor}

{\begin{rem}\label{Rem:evaluation}
In view of Remark \ref{Rem:special solution} and Theorem \ref{thm: gap pro},  for the special parameters $\alpha=1/2$ and $\beta=0$, the determinant $K^{(\alpha,\beta)}_s$ can be evaluated explicitly to give
\begin{equation}\label{initial-D}
\det\left(I-K^{(1/2,\;0)}_s\right )=e^{- {s^2}/{2}}I_0(s),
\end{equation}
where $I_0(s)$ is the modified Bessel function of order zero, with $I_0(0)=1$.  It is worth noting that for these special parameters, the confluent hypergeometric kernel is equivalent to the Bessel kernel  (see \cite{dkv}):
 \begin{equation*}
 K^{({1}/{2},\;0)}(u,v)=\frac{\sqrt{|uv|}}{2}\;\frac{J_0(u)J_0'(v)-J_0(v)J_0'(u)}
{u-v},
 \end{equation*}
where $J_0(x)$ is the Bessel function of order zero.  This explicit form
\eqref{initial-D}
of the determinant   furnishes   a   starting point in our asymptotic analysis of the integral of the Hamiltonian in  later sections.
\end{rem}}

\subsection{Total integral of the Hamiltonian and large gap aysmptotics}

As is noted in Remark \ref{Rem:evaluation} that for certain special parameters, the integral of the Hamiltonian associated can be  evaluated explicitly.
   For general parameters, the derivation is divided into two steps. First we evaluate the derivative of the Hamiltonian integral
   by deriving and applying certain differential identities, with respect  to the parameters,  for the Hamiltonian.  Then, after taking integration   we obtain the following  total integral of the Hamiltonian.

 \begin{thm}\label{thm: total integral}  For pure imaginary parameter $\beta$ and $\alpha>-1/2$, we have the total integral of the Hamiltonian after regularization  at infinity
 \begin{align}\label{eq: total integral-H-thm}
&\int_0^{c}H(\tau;\alpha,\beta)d\tau+\int_c^{-i\infty}\left [H(\tau;\alpha,\beta)
-\frac{\tau}{16}-\frac{i}{2}\alpha+\frac{\alpha^2-\beta^2+\frac{1}{4}}{\tau}\right ]d\tau \nonumber\\
&=\frac{c^2}{32}+\frac{i\alpha c}{2}-\left(\alpha^2-\beta^2+\frac{1}{4}\right)\ln \frac {|c|}4
+\ln\left(\frac{\sqrt{\pi}\; G( {1}/{2})^2\;G(1+2\alpha)}{2^{2\alpha^2}G(1+\alpha+\beta)\;G(1+\alpha-\beta)}\right),
\end{align}
where  $ic>0$ and $G$ is   the Barnes $G$-function.
\end{thm}

As an application of  the above total integral of the Hamiltonian, we derive the  large gap asymptotics of the   confluent hypergeometric-kernel determiant, including the explicitly given constant term. From the large-$s$ asymptotic formula  \eqref{thm: H-asy-large s}, we see that the second integral in \eqref{eq: total integral-H-thm} is of order $O(1/c)$
as $ic\to+\infty$. Therefore,  let $c=-4is$ and  substitute \eqref {eq: total integral-H-thm} into \eqref{eq:gap pro}, we obtain the following theorem as $s\to+\infty$.

 \begin{thm}\label{thm: large gap asy} (Deift, Krasovsky and  Vasilevska  \cite{dkv})  For pure imaginary parameter $\beta$ and $\alpha>-1/2$, let $K_s^{(\alpha,\beta)}$ be
  the operator with the confluent hypergeometric kernel $K^{(\alpha,\beta)}(u,v)$ in \eqref{def: CFH-kernel} acting on $L^2(-s,s)$, we have the asymptotics as $s\to +\infty$
\begin{equation}\label{eq:large gap pro}
\det(I-K_s^{(\alpha,\beta)})=
\frac{\sqrt{\pi} G(\frac{1}{2})^2G(1+2\alpha)}{2^{2\alpha^2}G(1+\alpha+\beta)G(1+\alpha-\beta)}s^{-(\alpha^2-\beta^2+\frac{1}{4})}e^{-\frac{s^2}{2}+2\alpha s}\left [1+O\left(\frac 1 s\right )\right ],
\end{equation}where $G$ is  the Barnes $G$-function.
\end{thm}
Therefore, we have reproduced the large gap asymptotics of the determinant with a confluent hypergeometric kernel, obtained previously by Deift, Krasovsky and Vasilevska \cite{dkv}. Particularly, we recover the large gap asymptotics of the sine-kernel determinant by taking  $\alpha=\beta=0$ in \eqref{eq:large gap pro} and using the relation between  the Barnes $G$-function and the  Riemann zeta-function that $\ln G(1/2)=\frac{3}{2}\zeta'(-1)-\frac{1}{4}\ln \pi+\frac{1}{24}\ln 2$.

 \begin{cor}\label{thm:large gap -sine}  Let $K_s^{(0,0)}$ be
  the operator with sine kernel \eqref{def: sin-kernel} acting on $L^2(-s,s)$, we have the large gap asymptotics, as $s\to+\infty$,
\begin{equation}\label{eq:gap pro-sine}
\ln \det\left(I-K_s^{(0,0)}\right )=-\frac{s^2}{2}-\frac{\ln s}{4}+\frac{\ln2}{12}+3\zeta'(-1)+O\left(\frac 1 s\right ),\end{equation}
where $\zeta'(s)$ is the derivative of the Riemann zeta-function.
\end{cor}

Using the series expansion of $\sigma_V(s)$ as $s\to +\infty$, of which the first few terms can be found in \eqref{asy: sigma-V}, Dyson was able to  obtain the asymptotics of the sine-kernel determinant of the form \eqref{eq:gap pro-sine} in \cite{Dyson}. While the constant factor can not be derived from the series expansion of $\sigma_V(s)$, it was conjectured  to be given by the  Riemann zeta-function $c_0=\frac{\ln2}{12}+3\zeta'(-1)$ by Dyson in the same paper \cite{Dyson}.  This conjecture has been proved rigorously by Krasovsky \cite{kr} and Ehrhardt \cite{E06} using different methods.  In the present paper, our result provides another proof of the large gap asymptotics \eqref{eq:gap pro-sine}
 including the  constant conjectured by Dyson.

 The rest of the paper is arranged as follows. In Section \ref{Analysis of the model RH problem},  we formulate the model Riemann-Hilbert problem (RH problem or RHP, for short) for $\Psi(\zeta; s)$. We derive a Lax pair corresponding to the model RH problem, of which the compatibility condition is expressed as a coupled Painlev\'e V system.   Several differential identities of the Hamiltonian  are also established, which are crucial in the evaluation of the total integral of the Hamiltonian.  We then justify the solvability of $\Psi(\zeta;s)$ for the parameters $\alpha>-1/2$, $\beta\in i\mathbb{R}$ and $s\in -i(0, \infty)$ by proving a vanishing lemma. The solvability  implies the pole-free of the Hamiltonian  associated with the coupled Painlev\'e V system for $ is>0$. The special function solutions to the coupled Painlev\'e system  are also derived for certain specific parameters in this section.
 In Section \ref{sec:small-s}, we derive asymptotic approximations for the coupled Painlev\'e V system as $is\to 0^+$
by using Deift-Zhou nonlinear steepest descent method for the RH problems. While the large-$s$ analysis will be carried out in Section \ref{sec:large-s}. These two sections together provide a proof of Theorem \ref{thm-asy}.
In Section \ref{sec:d-identity}, we formulate  a RH problem for $Y$, corresponding to the orthogonal polynomials associated with the weight function \eqref{def: weight-circle} on an arc, we also derive a differential identity connecting the Toeplitz determinant \eqref{def: D-n} with $Y$.
In Section \ref{sec:Y}, we carry out, in full details,  the Deift-Zhou nonlinear steepest descent analysis  of the RH problem for $Y$.
The asymptotics of the orthogonal polynomials and the differential identity  then lead to a proof of Theorem \ref{thm-gap-probability} in this section.
Section \ref{sec:large-gap}  will be devoted to the evaluation of the total integral of the Hamiltonian and the large gap asymptotic analysis of the model. Theorem \ref{thm: total integral} is also proved  in this last section.

\section{Model RH problem and the coupled Painlev\'e V system}\label{Analysis of the model RH problem}
In this section, we introduce a  RH problem for $\Psi(\zeta;s)$ for later use. We derive a Lax pair for the solution to the RH problem for $\Psi(\zeta;s)$, which turns out to be  a Garnier system  in the list of  \cite{k-3,kns-1}, with three regular singularities and one irregular singularity of order two. The compatibility condition of the Lax pair is described by the  coupled Painlev\'e V system in dimension four.
We then prove the solvability of this model RH problem and  the pole-free of the Hamiltonian associated with the
 coupled Painlev\'e V system for $is>0$. The Hamiltonian plays a   central role in the derivation of our main results on the asymptotics for the Toeplitz determinant.
In this section,  we  will also show that for special parameters, the   coupled Painlev\'e system  admits special function solutions.

\subsection{A Model RH problem }\label{sec:model-RHP}

We formulate a model RH problem, which will pave the road to the steepest descent analysis in later sections. The model  RH problem for $2\times 2$ matrix function $\Psi(\zeta)=\Psi(\zeta; s)$ is the following:
\begin{description}
  \item(a)  $\Psi(\zeta; s)$ is analytic in
  $\zeta\in \mathbb{C}\setminus \{\cup^7_{j=1}\Sigma_j\}$, where the oriented $\zeta$-contours
   \begin{equation*}
   \Sigma_1=1+e^{\frac{ \pi i}{4}}\mathbb{R}^{+},~\Sigma_2=-1+e^{\frac{3\pi i}{4}}\mathbb{R}^{+}, ~ \Sigma_3=-1+e^{-\frac{3\pi i}{4}}\mathbb{R}^{+},~
  \Sigma_4=e^{-\frac{ \pi i}{2}}\mathbb{R}^{+},~ \Sigma_5=1+e^{-\frac{\pi i }{4}}\mathbb{R}^{+},
   \end{equation*}
   $\Sigma_6=(0,1)$ and $\Sigma_7=(-1,0)$, as depicted  in Figure \ref{Model RH contour}; see also an illustration of the regions $\Omega_j$ for $j=1,2,\cdots,5$, each having $\Sigma_j$ and $\Sigma_{j+1}$ as portions of its boundary.

  \begin{figure}[th]
 \begin{center}
   \includegraphics[width=8 cm]{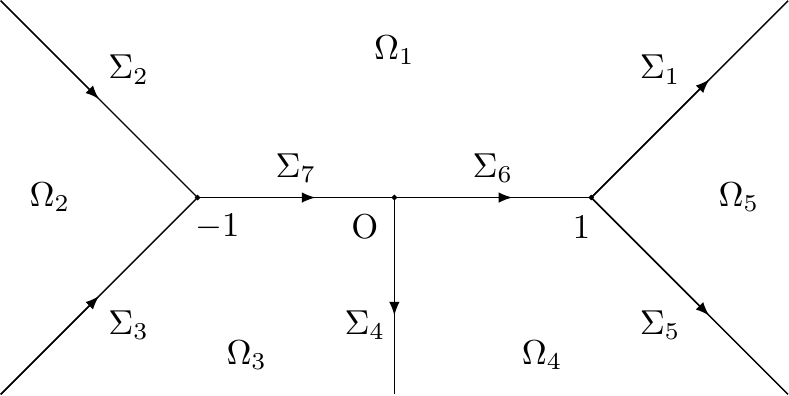} \end{center}
  \caption{\small{Contours and regions for the model RH problem}}
 \label{Model RH contour}
\end{figure}

  \item(b)  $\Psi(\zeta;s)$ satisfies the jump conditions
 \begin{equation}\label{eq:Psi-jump}
 \Psi_+(\zeta;s)=\Psi_-(\zeta;s)
 \left\{
 \begin{array}{ll}
    \begin{pmatrix}
                                 1 & 0 \\
                               e^{-\pi i(\alpha-\beta)}& 1
                                 \end{pmatrix}, &  \zeta \in \Sigma_1, \\[.5cm]
    \begin{pmatrix}
                                1 &0\\
                              e^{\pi i(\alpha-\beta)}&1
                                \end{pmatrix},  &  \zeta \in \Sigma_2, \\[.5cm]
    \begin{pmatrix}
                                 1 &-e^{-\pi i(\alpha-\beta)} \\
                                 0 &1
                                 \end{pmatrix}, &   \zeta \in \Sigma_3, \\[.1cm]

     e^{2\pi i\beta\sigma_3} , &   \zeta \in \Sigma_4, \\[.1cm]
                               \begin{pmatrix}
                                 1 &-e^{\pi i(\alpha-\beta)} \\
                                 0 &1
                                 \end{pmatrix}, &   \zeta \in \Sigma_5, \\[.5cm]
                                 \begin{pmatrix}
                                0 &-e^{\pi i(\alpha-\beta)} \\
                                e^{-\pi i(\alpha-\beta)} &0
                                 \end{pmatrix}, &   \zeta \in \Sigma_6, \\[.5cm]
                                 \begin{pmatrix}
                                0 &-e^{-\pi i(\alpha-\beta)} \\
                                e^{\pi i(\alpha-\beta)} &0
                                 \end{pmatrix}, &   \zeta \in\Sigma_7.
 \end{array}  \right .  \end{equation}
 For later use, we denote the jump  matrix on $\Sigma_k$ by $J_k$ for $k=1,2,...,7.$

\item(c) As $\zeta\to \infty$,  we have
  \begin{equation}\label{eq:Psi-infty}
 \Psi(\zeta;s)=\left( I  +\frac{\Psi_1(s)}{\zeta}+\frac{\Psi_2(s)}{\zeta^2}+
O\left(\frac{1}{\zeta^3}\right) \right) \zeta^{-\beta\sigma_3} e^{\frac{1}{4} s\zeta \sigma_3},
   \end{equation}
  where  the branch cut of the function  $\zeta^{\beta} $ is taken along $(0, -i\infty)$ such  that $\arg \zeta \in (-\pi/2, 3\pi/2)$,  and $\sigma_3$ is the Pauli matrix. This and the other Pauli matrices are
 \begin{equation}\label{Pauli-matrix}
 \sigma_1=\begin{pmatrix}
                    0 &1 \\
                   1 & 0
          \end{pmatrix},   \quad \sigma_2=
                 \begin{pmatrix}
                    0 & -i \\
                    i & 0 \end{pmatrix} \quad\mbox{and}\quad \sigma_3 =\begin{pmatrix}
                     1 & 0 \\
                    0 & -1
                    \end{pmatrix}.
                    \end{equation}

 \item(d) As $\zeta\to 0$, we have
  \begin{equation}\label{eq: Psi-origin}
  \Psi(\zeta;s)=\Psi^{(0)}(\zeta; s) \zeta^{\alpha \sigma_3} E^{(0)}_{j},  \quad \zeta\in \Omega_j, ~j=1,3,4, \end{equation}
  where $\Psi^{(0)}(\zeta; s) =\Psi^{(0)}_0(s)\left(I+\Psi^{(0)}_1(s)\zeta+O\left(\zeta^2\right )\right)$ is analytic at $\zeta=0$, and $\zeta^{\alpha } $ takes the principal branch.
  The connection matrices are given below
$$E^{(0)}_1=I, \quad E^{(0)}_4=J_6^{-1}, \quad E^{(0)}_3=J_6^{-1}J_4^{-1},$$
where   $J_k$ is the jump matrix for $\Psi$ on $\Sigma_k$ for $k=1,2,...,7$; cf. \eqref{eq:Psi-jump}.

 \item(e) As $\zeta\to 1$, we have
  \begin{equation}\label{eq: Psi-1}
  \Psi(\zeta;s)=\Psi^{(1)}(\zeta; s)    \begin{pmatrix}
                    1 & c_1\ln(\zeta-1)\\
                 0 & 1
                  \end{pmatrix}E^{(1)}_{j},  \quad \zeta \in \Omega_j, ~j=1,4,5, \end{equation}
  where  $\Psi^{(1)}(\zeta; s) =\Psi^{(1)}_0(s)\left(I+\Psi^{(1)}_1(s)(\zeta-1)+O\left ((\zeta-1)^2\right )\right)$ is analytic at  $\zeta=1$,
  the logarithmic  function takes the principal branch,
    the constant $c_1=-\frac{1}{2\pi i}e^{\pi i(\alpha-\beta)}$ and the connection matrices are given as
 \begin{equation*}
 E^{(1)}_1=I, \quad E^{(1)}_5=J_1^{-1}, \quad E^{(1)}_4=J_1^{-1}J_5^{-1},
 \end{equation*}
    $J_k$ being the jump matrix for $\Psi$ on $\Sigma_k$ for $k=1,2,...,7$; see \eqref{eq:Psi-jump}.

 \item(f) As $\zeta\to -1$, we have
  \begin{equation}\label{eq: Psi-1 minus}
  \Psi(\zeta;s)=\Psi^{(2)}(\zeta; s)    \begin{pmatrix}
                    1 & c_2\ln(\zeta+1)\\
                 0 & 1
                  \end{pmatrix}E^{(2)}_{j},  \quad \zeta\in \Omega_j, ~j=1,2,3,\end{equation}
  where $\Psi^{(2)}(\zeta; s) =\Psi^{(2)}_0(s)\left(I+\Psi^{(2)}_1(s)(\zeta+1)+O\left((\zeta+1)^2\right)\right)$ is analytic at $\zeta=-1$.
  The branch for the logarithm   is taken so that $\arg(\zeta+1)\in(0,2\pi)$.
    The constant $c_2=\frac{1}{2\pi i}e^{-\pi i(\alpha-\beta)}$ and the connection matrices are given as follow
$$E^{(2)}_1=I, \quad E^{(2)}_2=J_2^{-1}, \quad E^{(2)}_3=J_2^{-1}J_3^{-1},$$
where   $J_k$ is the jump matrix for $\Psi$ on $\Sigma_k$ for $k=1,2,...,7$; see \eqref{eq:Psi-jump}.
\end{description}

         In the above model RH problem, the connection matrices describing  the sector-wise behaviors of $\Psi$ near the node points $\zeta=0$ and $\pm1$,
are specified  in a way such that they are in compliance with  the jump conditions \eqref{eq:Psi-jump} in neighborhoods of  the origin and $\zeta=\pm1$.

\subsection{Lax pair and the coupled Painlev\'e V  system}
In this section, we derive a Lax pair from the RH problem for $\Psi$. The Lax pair is associated with a special Garnier system,  of which the compatibility condition is expressed as the coupled Painlev\'e  V system in dimension four.
\begin{prop} We have the following  Lax pair
\begin{equation}\label{def: Lax pair}
 \frac{d}{d\zeta}\Psi(\zeta;s)=L(\zeta;s) \Psi(\zeta;s),  \quad   \frac{d}{ds}\Psi(\zeta;s)=U(\zeta;s) \Psi(\zeta;s),
  \end{equation}
  where
  \begin{equation}\label{eq: L}
 L(\zeta;s)=\frac{s}{4}\sigma_3+\frac{A_0(s)}{\zeta}+\frac{A_1(s)}{\zeta-1}+ \frac{A_2(s)}{\zeta+1},  \end{equation}
  \begin{equation}\label{eq: U}
 U(\zeta;s)=\frac{1}{4}\zeta\sigma_3+B(s), \end{equation}
with the coefficients  given below
   \begin{equation}\label{eq:A-0}
 A_0(s)= \begin{pmatrix}
                    u_1(s)v_1(s)+u_2(s)v_2(s)-\beta &  -\left( u_1(s)v_1(s)+u_2(s)v_2(s)-\alpha-\beta\right)y(s)\\
                 \left(u_1(s)v_1(s)+u_2(s)v_2(s)+\alpha-\beta\right)/y(s) &  -u_1(s)v_1(s)-u_2(s)v_2(s)+\beta
                  \end{pmatrix}, \end{equation}
  \begin{equation}\label{eq:A-k}
 A_k(s)= \begin{pmatrix}
                   - u_k(s)v_k(s) &  u_k(s)y(s)\\
                -u_k(s)v_k^2(s)/y(s) &  u_k(s)v_k(s)                  \end{pmatrix},  \quad k=1,2, \end{equation}
and
  \begin{equation}\label{eq:B}
B(s)=\frac{1}{s} \begin{pmatrix}
                  0 &  b_1(s)y(s)\\
           b_2(s)/y(s) & 0                  \end{pmatrix},   \end{equation}
            where
 \begin{equation}\label{eq:b-k}
 \left\{ \begin{array}{ll}
               b_1(s)=-u_1(s)(v_1(s)-1)-u_2(s)(v_2(s)-1)+\alpha+\beta, \\
                     b_2(s)=-u_1(s)v_1(s)(v_1(s)-1)-u_2(s)v_2(s)(v_2(s)-1)+\alpha-\beta.
                                                                \end{array}
\right.
\end{equation}
The compatibility condition of the Lax pair is expressed as the coupled Painlev\'e V system
 \begin{equation}\label{eq:cpv}
 \left\{ \begin{array}{l}
                 s\frac{d u_1}{ds}=\frac{s}{2} u_1-u_1^2(v_1-1)(3v_1-1)-u_1u_2(v_2-1)(2v_1+v_2-1)+2(\alpha+\beta)u_1v_1-2\beta u_1,\\
                   s\frac{d u_2}{ds}=-\frac{s}{2} u_2-u_2^2(v_2-1)(3v_2-1)-u_1u_2(v_1-1)(v_1+2v_2-1)+2(\alpha+\beta)u_2v_2-2\beta u_2, \\
                   s\frac{d v_1}{ds}=-\frac{s}{2} v_1+2u_1v_1(v_1-1)^2+u_2(v_1+v_2)(v_1-1)(v_2-1)-\alpha(v_1^2-1)-\beta (v_1-1)^2,\\
                      s\frac{d v_2}{ds}=\frac{s}{2} v_2+2u_2v_2(v_2-1)^2+u_1(v_1+v_2)(v_1-1)(v_2-1)-\alpha(v_2^2-1)-\beta (v_2-1)^2.
                                           \end{array}
\right. \end{equation}
Let
\begin{equation}\label{eq:H-psi}
  sH=sH(u_1,v_1,u_2,v_2,s;\alpha,\beta)=-\frac{s}{2}\left(\Psi_1\right)_{11}-(\alpha^2-\beta^2),\end{equation}
  where $(\Psi_1)_{11}$ is the $11$-entry of the  coefficient $\Psi_1$ in the large-$\zeta$ asymptotic approximation \eqref{eq:Psi-infty} of $\Psi(\zeta;s)$, then we have
  \begin{equation}\label{eq:H}
sH=\frac{s}{2}H_{V}( u_1,v_1,s/2; \alpha,\beta) -\frac{s}{2}H_{V}(u_2,v_2,-s/2; \alpha,\beta)  +u_1u_2(v_1+v_2)(v_1-1)(v_2-1),
\end{equation}
as defined in \eqref{int:H}, where $H_V( u,v,s; \alpha, \beta)$ is the Hamiltonian for the Painlev\'e V equation
\begin{equation}\label{eq:H-v}
  sH_{V}(u,v,s; \alpha, \beta)=u^2v(v-1)^2-suv-\alpha u (v^2-1)-\beta u (v-1)^2.\end{equation}
The system \eqref{eq:cpv} is equivalent to the following Hamiltonian formulation
\begin{equation}\label{eq:H-equation}
  \frac{d v_k}{ds}=\frac{\partial H}{\partial u_k} ,\quad  \frac{d u_k}{ds}=-\frac{\partial H}{\partial v_k}, \quad k=1,2.\end{equation}
  \end{prop}

 \begin{proof}
 Note that all the jump matrices in \eqref{eq:Phi-jump} of the RH problem for $\Psi(\zeta;s)$ are independent of the variables $\zeta$ and $s$.
 Then $\frac{d}{d\zeta}\Psi(\zeta;s)$, $\frac{d}{ds}\Psi(\zeta;s)$ and $\Psi(\zeta;s)$ satisfy the same jump conditions.
 Thus, the matrix-valued functions $L(\zeta;s)= \frac{d}{d\zeta}\Psi(\zeta;s)\Psi^{-1}(\zeta;s)$ and $U(\zeta;s)= \frac{d}{ds}\Psi(\zeta;s)\Psi^{-1}(\zeta;s)$
 are meromorphic for $\zeta$ in the complex plane with only possible isolated singularities  at $\zeta=0, \pm 1$.
 Then, it follows from the local behaviors of $\Psi(\zeta; s)$ as $\zeta\to\infty$, $\zeta\to 0$, and $\zeta\to \pm 1$; cf. \eqref{eq:Psi-infty}, \eqref{eq: Psi-origin}, \eqref{eq: Psi-1} and \eqref{eq: Psi-1 minus},  that $L$ and $U$ are rational functions in $\zeta$ and take  the form as given in \eqref{eq: L} and \eqref{eq: U}, respectively.

 Using the fact that $\det \Psi=1$, we have $\tr L=\tr U=0$ and thus all the coefficients $A_k$, $k=0,1,2$ and $B$ are trace-zero.
 Substituting the behavior of $\Psi $ at infinity \eqref{eq:Psi-infty} into the first equation of the Lax pair \eqref{def: Lax pair}, we find after comparing the
 coefficient of $\frac{1}{\zeta}$ that
\begin{equation}\label{eq:A-k-Psi}A_0+A_1+A_2=\begin{pmatrix}
                - \beta&  -\frac{s}{2}(\Psi_1)_{12}\\
              \frac{s}{2}(\Psi_1)_{21} & \beta                  \end{pmatrix},
              \end{equation}
              where $\Psi_1$ is the coefficient of
             $ {1}/{\zeta}$ in the large-$\zeta$  asymptotic expansion of $\Psi(\zeta;s)$ in \eqref{eq:Psi-infty}.
 Moreover, combining the master equation in \eqref{def: Lax pair} with the local behavior \eqref{eq: Psi-origin} of  $\Psi(\zeta;s)$ at $\zeta=0$, and \eqref{eq: Psi-1}-\eqref{eq: Psi-1 minus} at $\zeta=\pm 1$, we have
 \begin{equation*}
  \det A_0=-\alpha^2, \quad \det A_k=0, \quad k=1,2.
 \end{equation*}
 Thus,  the coefficients $A_k$, $k=0,1,2$, can be taken in the form appeared in \eqref{eq:A-0} and \eqref{eq:A-k}.

  Similarly, substituting the behavior of $\Psi $ at infinity \eqref{eq:Psi-infty} into the second equation of the Lax pair \eqref{def: Lax pair} and comparing    the leading-order terms
  in $\zeta$, we find
\begin{equation}\label{eq:B-Psi}B=\begin{pmatrix}
                0&  -\frac{1}{2}(\Psi_1)_{12}\\
              \frac{1}{2}(\Psi_1)_{21} & 0     \end{pmatrix},
              \end{equation}
              where $ \Psi_1 $ is  the coefficient  of the large-$\zeta$ asymptotic approximation of $\Psi(\zeta;s)$ in \eqref{eq:Psi-infty}.
 This,  together with \eqref{eq:A-k-Psi}, implies that $B$ is given by  \eqref{eq:B}.

 The compatibility condition of the Lax pair \eqref{def: Lax pair} is readily written   as
 \begin{equation}\label{eq:nonlinear equations-matrix}
 \left\{ \begin{array}{ll}
                  \frac{d}{ds} A_0=[B, A_0], \\
                     \frac{d}{ds} A_1=[B+\frac{1}{4}\sigma_3, A_1], \\
                       \frac{d}{ds} A_2=[B-\frac{1}{4}\sigma_3, A_2],
                    \end{array}
\right.
 \end{equation}where the commutator $[A,B]=AB-BA$.
 Denote  $a(s)$ by
 \begin{equation}\label{eq:a}
                   a(s)=(A_0)_{11}(s)=u_1(s)v_1(s)+u_2(s)v_2(s)-\beta. \end{equation}
Recalling the definition of $b_k$ in \eqref{eq:b-k} and substituting \eqref{eq:A-0}, \eqref{eq:A-k} and \eqref{eq:B} into \eqref{eq:nonlinear equations-matrix} yields
  \begin{equation}\label{eq:nonlinear equations}
 \left\{ \begin{array}{ll}
                  s\frac{d}{ds} ((a-\alpha)y)=2ab_1y,\\
                    s \frac{d}{ds}((a+\alpha)/y) =2ab_2/y,\\
                      s \frac{d}{ds} (u_kv_k)=b_1u_kv_k^2+b_2u_k, \quad k=1,2,\\
                       s  \frac{d}{ds} (u_ky)=\frac{(-1)^{k+1}s}{2}u_k y+2b_1u_k v_k y, \quad k=1,2.
                                           \end{array}
\right. \end{equation}
Then, the first two equations imply that the gauge parameter $y(s)$ of the Lax pair satisfies the equation
 \begin{equation}\label{eq:y}
 sy'(s)/y(s)=b_1(s)-b_2(s)=u_1(s)(v_1(s)-1)^2+u_2(s)(v_2(s)-1)^2+2\beta.
 \end{equation}
 This equation, together with the last two equations of \eqref{eq:nonlinear equations}, gives us the system   \eqref{eq:cpv}.

Substituting the large-$\zeta$ expansion   \eqref{eq:Psi-infty} into the first equation of the Lax pair \eqref{def: Lax pair}, and comparing the coefficients of $1/\zeta^2$, we find
\begin{equation}\label{eq:psi-1}
\frac{s}{4}[\sigma_3\Psi_1, \Psi_1]+\frac{s}{4}[\Psi_2, \sigma_3]-\beta[\Psi_1, \sigma_3]-\Psi_1=A_1-A_2.
 \end{equation}
 Hence, combining the definition \eqref{eq:H-psi}  with \eqref{eq:A-k} and   \eqref{eq:psi-1}, we obtain  \eqref{eq:H}. By \eqref{eq:H},   it is readily verified  that the Hamiltonian system  \eqref{eq:H-equation} is equivalent to the system of equations \eqref{eq:cpv}.
This completes the proof of the proposition.  \end{proof}

We then derive several differential identities for later use.
\begin{prop}\label{Pro: diff-identities-psi} Let $\Psi^{(1)}$ and $\Psi^{(2)}$ be the matrix functions
in the asymptotic behaviors of $\Psi(\zeta;s)$ near $\zeta=1$ and $\zeta=-1$, defined respectively in \eqref{eq: Psi-1} and  \eqref{eq: Psi-1 minus}, we have
\begin{equation}\label{eq:Psi-1-diff}
u_k(s)v_k(s)=\frac{(-1)^{k}i}{ \pi} e^{(-1)^{k+1}\pi i(\alpha-\beta)}\frac{d}{ds}   \left(\Psi^{(k)}_1(s) \right  )_{21}  \quad k=1,2.\end{equation}
 Moreover,  we have
 \begin{equation}\label{eq:y-def}y(s)=\frac{  \left (\Psi^{(0)}_0 (s)\right  )_{11}}{ \left (\Psi^{(0)}_0 (s)\right )_{21}}
 \end{equation} and
 \begin{equation}\label{eq:y-diff}
 s\frac{d}{ds}\ln y=u_1(s)\left(v_1(s)-1\right )^2+u_2(s)\left(v_2(s)-1\right )^2+2\beta.
 \end{equation}
 Denoting
\begin{equation}\label{eq:d}
d(s)=2\alpha \left(\Psi^{(0)}_0(s)\right)_{11}\left(\Psi^{(0)}_0(s)\right )_{21},\end{equation}
we have
\begin{equation}\label{eq:d-diff}
s\frac{d}{ds}\ln d(s)
=b_1(s)+b_2(s)=- u_1(s)\left(v_1^2(s)-1\right )-u_2(s)\left(v_2^2(s)-1\right ) +2\alpha.\end{equation}
\end{prop}
\begin{proof}
Substituting  \eqref{eq: Psi-origin},  \eqref{eq: Psi-1} and  \eqref{eq: Psi-1 minus} into the first equation of  \eqref{def: Lax pair}, we find
\begin{equation}\label{eq:Psi-A-0}
A_0(s)= \alpha\Psi^{(0)}_0(s)\sigma_3\left\{\Psi^{(0)}_0(s)\right\}^{-1},\end{equation}
\begin{equation}\label{eq:Psi-A-k}
A_k(s)= c_k \Psi^{(k)}_0(s)\sigma_{+}\left\{\Psi^{(k)}_0(s)\right\}^{-1},\end{equation}
where $\sigma_+=\begin{pmatrix}
                0&  1\\
             0 & 0     \end{pmatrix}$ and the constant $c_k=(-1)^k\frac{1}{2\pi i}e^{(-1)^{k+1}\pi i(\alpha-\beta)} $ for $k=1,2$.
Similarly, upon substitution of \eqref{eq: Psi-origin},  \eqref{eq: Psi-1} and  \eqref{eq: Psi-1 minus} into the second equation of  \eqref{def: Lax pair},
it follows that
\begin{equation}\label{eq:Psi-0-d}
\frac{d}{ds} \Psi^{(0)}_0(s)=B\Psi^{(0)}_0(s), \end{equation}
\begin{equation}\label{eq:Psi-k-d}
\frac{d}{ds} \Psi^{(k)}_0(s)=\left (\frac{(-1)^{k+1} }{4}\sigma_3+B\right ) \Psi^{(k)}_0(s), \end{equation}
and
\begin{equation}\label{eq:Psi-1-d}
\frac{d}{ds} \Psi^{(k)}_1(s)= \frac{1}{4} \left\{\Psi^{(k)}_0(s)\right\}^{-1}\sigma_3 \Psi^{(k)}_0(s) \end{equation}
 for $k=1,2$.
 We then get from \eqref{eq:Psi-A-k} and \eqref{eq:Psi-1-d} that
   \begin{equation}\label{eq:Psi-1-A}
\frac{d}{ds} \left (\Psi^{(k)}_1(s)\right )_{21}=\frac{1}{2c_k} \left(A_k(s)\right )_{11} \end{equation}
for $k=1,2$.
This, together with  \eqref{eq:A-k}, implies \eqref{eq:Psi-1-diff}.

Equation \eqref{eq:y-def} follows from a combination of \eqref{eq:A-0} and \eqref{eq:Psi-A-0}. Indeed, we have
\begin{equation*}
y(s)=\frac{\left(A_0(s)\right )_{11}+\alpha}{\left(A_0(s)\right )_{21} }=\frac{  \left (\Psi^{(0)}_0 (s)\right  )_{11}}{ \left (\Psi^{(0)}_0 (s)\right )_{21}}.
\end{equation*}Here use has been made of the fact that $\det\Psi^{(0)}_0 (s)=1$.
In view of \eqref{eq:b-k} and \eqref{eq:y-def}, and splitting
  \eqref{eq:Psi-0-d} to entries, we derive the logarithmic derivative \eqref{eq:d-diff}
for the gauge parameter $d(s)$. The equation \eqref{eq:y-diff} has already been
derived before; see \eqref{eq:y}.  This completes the proof.
\end{proof}

The following differential identities of the Hamiltonian with respect  to the parameters $\alpha$ and $\beta$ are crucial to the evaluation of the total integral of the Hamiltonian.
\begin{prop}
For the Hamiltonian $H=H(u_1,v_1,u_2,v_2,s;\alpha,\beta)$ it holds
\begin{equation}\label{eq:total diff}
H=\left(u_1\frac{dv_1}{ds}+u_2\frac{dv_2}{ds}-H\right)+\frac{d}{ds}\left(sH+\alpha \ln d(s)-\beta\ln y(s)-2(\alpha^2-\beta^2)\ln s\right),
\end{equation}
where $y(s)$ and $d(s)$ satisfy the properties given in Proposition \ref{Pro: diff-identities-psi}.
The Hamiltonian system implies the following useful differential formulas
\begin{equation}\label{eq:action-diff-a}
\frac{d}{d\alpha}\left(u_1\frac{dv_1}{ds}+u_2\frac{dv_2}{ds}-H\right)=\frac{d}{ds}\left(u_1\frac{d}{d\alpha}v_1+u_2\frac{d}{d\alpha}v_2-\ln d+2\alpha \ln s\right ),
\end{equation}
\begin{equation}\label{eq:action-diff-b}
\frac{d}{d\beta}\left(u_1\frac{dv_1}{ds}+u_2\frac{dv_2}{ds}-H\right)=\frac{d}{ds}\left (u_1\frac{d}{d\beta}v_1+u_2\frac{d}{d\beta}v_2+\ln y-2\beta \ln s\right ).
\end{equation}
\end{prop}
\begin{proof}
It follows from the Hamiltonian system \eqref{eq:H-equation}  and \eqref{eq:H}  that
\begin{equation}\label{eq:H-der}\frac{d}{ds}(sH(s))=-\frac{1}{2}\left(u_1(s)v_1(s)-u_2(s)v_2(s)\right).
\end{equation}
Using this equation, together with  \eqref{eq:H}, \eqref{eq:y-diff} and \eqref{eq:d-diff},  we readily verify \eqref{eq:total diff}.
The verification of \eqref{eq:action-diff-a} is straightforward, by substituting
\eqref {eq:H-equation},  \eqref{eq:H} and \eqref{eq:d-diff} into it.
While that of \eqref{eq:action-diff-b}   follows from \eqref {eq:H-equation},  \eqref{eq:H}  and \eqref{eq:y-diff}. Here we regard $u_k$ and $v_k$, $k=1,2$, as functions of $\alpha$ and $\beta$, and $\alpha$, $\beta$ and $s$ are independent.
 \end{proof}


 \subsection{Vanishing lemma and existence of solution to the model RH problem  }
We will prove the existence of solution to the model RH problem for $\Psi(\zeta;s)$ if the parameters  $\alpha>-1/2$, $\beta\in i \mathbb{R}$ and $s\in -i(0, +\infty)$.  We start with the vanishing lemma which shows that the homogeneous RH problem has only zero solution.
 \begin{lem}\label{pro: vanishing lemma}
 For $\alpha>-1/2$, $\beta\in i \mathbb{R}$ and $s\in -i(0, +\infty)$, we suppose that $\hat{\Psi}(\zeta)$ satisfies the same jump conditions \eqref{eq:Psi-jump} and the same  behaviors \eqref{eq: Psi-origin}, \eqref{eq: Psi-1} and \eqref{eq: Psi-1 minus} as $\Psi(\zeta)$, respectively as $\zeta$ tends to the origin and $\pm1$. Further assume the behavior of $\hat{\Psi}(\zeta)$ to be
  \begin{equation}\label{eq:hat-Psi-infty}
 \hat{\Psi}(\zeta)=
O\left(\frac{1}{\zeta}\right)\zeta^{-\beta\sigma_3} e^{\frac{1}{4} s\zeta \sigma_3},~~\mbox{as}~ \zeta\to \infty.
   \end{equation}
   Then, we have $\hat{\Psi}(\zeta)=0$ for $\zeta\in \mathbb{C}$.
 \end{lem}

 \begin{proof}
 To normalize the behavior as $\zeta\to\infty$, we introduce the first transformation
 \begin{equation}\label{def:hat-Psi-1}
  \hat{\Psi}^{(1)}(\zeta)=\hat{\Psi}(\zeta)e^{-\frac{1}{2}\pi i\beta\sigma_3}
\varphi(\zeta)^{\beta\sigma_3}e^{-\frac{1}{4} s\zeta \sigma_3},~~\zeta\not\in \Sigma_4\cup\Sigma_6\cup\Sigma_7,
   \end{equation}
   where the contours are illustrated in  Figure \ref{Model RH contour}, and
   \begin{equation}\label{def:varphi}
 \varphi(\zeta)=\zeta+\sqrt{\zeta^2-1}.
 \end{equation}
 The branch for $\sqrt{\zeta^2-1}$ is taken such that  $\arg (\zeta\pm 1)\in(-\pi, \pi) $ and the branch for $\zeta^{\beta}$ is  taken as in \eqref{eq:hat-Psi-2-infty}, such that $\arg \zeta\in(-\pi/2, 3\pi/2) $.
 Then, we have the normalized  behavior at infinity
  \begin{equation}\label{eq:hat-Psi-1-infty}
 \hat{\Psi}^{(1)}(\zeta)=
O\left(\frac{1}{\zeta}\right),
   \end{equation}
and the modified jump conditions
 \begin{equation}\label{eq:hat-Psi-1-jump}
 \hat{\Psi}^{(1)}_+(\zeta)=\hat{\Psi}^{(1)}_-(\zeta)
 \left\{
 \begin{array}{ll}
    \begin{pmatrix}
                                 1 & 0 \\
                               e^{-\pi i\alpha}\varphi(\zeta)^{2\beta}e^{-\frac{1}{2}s\zeta}& 1
                                 \end{pmatrix}, &  \zeta \in \Sigma_1, \\[.5cm]
    \begin{pmatrix}
                                1 &0\\
                              e^{\pi i(\alpha-2\beta)}\varphi(\zeta)^{2\beta}e^{-\frac{1}{2}s\zeta}&1
                                \end{pmatrix},  &  \zeta \in \Sigma_2, \\[.5cm]
    \begin{pmatrix}
                                 1 &-e^{-\pi i(\alpha-2\beta)} \varphi(\zeta)^{-2\beta}e^{\frac{1}{2}s\zeta}\\
                                 0 &1
                                 \end{pmatrix}, &   \zeta \in \Sigma_3, \\[.5cm]
   \begin{pmatrix}
                                 1 &-e^{\pi i\alpha} \varphi(\zeta)^{-2\beta}e^{\frac{1}{2}s\zeta}\\
                                 0 &1
                                 \end{pmatrix}, &   \zeta \in \Sigma_5, \\[.5cm]
                                 \begin{pmatrix}
                                0 &-e^{\pi i\alpha}e^{\frac{1}{2}s\zeta} \\
                                e^{-\pi i\alpha} e^{-\frac{1}{2}s\zeta}&0
    \end{pmatrix}, &   \zeta \in \Sigma_6, \\[.5cm]
                                 \begin{pmatrix}
                                0 &-e^{-\pi i\alpha} e^{\frac{1}{2}s\zeta}\\
                                e^{\pi i\alpha}e^{-\frac{1}{2}s\zeta} &0
                                 \end{pmatrix}, &   \zeta \in\Sigma_7.
 \end{array}  \right .  \end{equation}

 To deform the jump contours to the real axis, we take the second transformation
 \begin{equation}\label{def:hat-Psi-2}
 \hat{\Psi}^{(2)}(\zeta)=\hat{\Psi}^{(1)}(\zeta)
 \left\{
 \begin{array}{ll}
    \begin{pmatrix}
                                 1 & 0 \\
                               e^{-\pi i\alpha}\varphi(\zeta)^{2\beta}e^{-\frac{1}{2}s\zeta}& 1
                                 \end{pmatrix}, & \Im\zeta>0,~\zeta \in \Omega_5, \\[.5cm]
      \begin{pmatrix}
                                 1 &e^{\pi i\alpha} \varphi(\zeta)^{-2\beta}e^{\frac{1}{2}s\zeta}\\
                                 0 &1
                                 \end{pmatrix}, &  \Im\zeta<0,~\zeta \in \Omega_5, \\[.5cm]

    \begin{pmatrix}
                                1 &0\\
                              e^{\pi i(\alpha-2\beta)}\varphi(\zeta)^{2\beta}e^{-\frac{1}{2}s\zeta}&1
                                \end{pmatrix},  & \Im\zeta>0,~\zeta \in \Omega_2, \\[.5cm]
    \begin{pmatrix}
                                 1 &e^{-\pi i(\alpha-2\beta)} \varphi(\zeta)^{-2\beta}e^{\frac{1}{2}s\zeta}\\
                                 0 &1
                                 \end{pmatrix}, &   \Im\zeta<0,~\zeta \in \Omega_2, \\[.5cm]

                               I, &  \mbox{otherwise}.
 \end{array}  \right .  \end{equation}
Then, $\hat{\Psi}^{(2)}(\zeta)$ is analytic in $\mathbb{C}\setminus\mathbb{R}$, with jump condition   on the
real axis as follows:
 \begin{equation}\label{eq:hat-Psi-2-jump}
\hat{ \Psi}^{(2)}_+(\zeta)=\hat{\Psi}^{(2)}_-(\zeta)\hat{J}(\zeta),   ~~ \zeta\in\mathbb{R},  \end{equation}
where the orientation is from left to right, and the jump matrices are
\begin{equation*}
\hat{J}(\zeta)= \left\{
 \begin{array}{ll}
    \begin{pmatrix}
                                 0 & -e^{\pi i\alpha} \varphi(\zeta)^{-2\beta}e^{\frac{1}{2}s\zeta} \\
                               e^{-\pi i\alpha}\varphi(\zeta)^{2\beta}e^{-\frac{1}{2}s\zeta}& 1
                                 \end{pmatrix}, &  \zeta >1, \\[.5cm]
       \begin{pmatrix}
                                0 &-e^{\pi i\alpha}e^{\frac{1}{2}s\zeta}\\
                                e^{-\pi i\alpha}e^{-\frac{1}{2}s\zeta} &0
                                 \end{pmatrix}, &  0< \zeta <1, \\[.5cm]
                                 \begin{pmatrix}
                                0 &-e^{-\pi i\alpha}e^{\frac{1}{2}s\zeta} \\
                               e^{\pi i\alpha}e^{-\frac{1}{2}s\zeta}&0
                                 \end{pmatrix}, & -1<\zeta <0,\\[.5cm]
          \begin{pmatrix}
                                0 &-e^{-\pi i(\alpha-2\beta)} \varphi(\zeta)^{-2\beta}e^{\frac{1}{2}s\zeta}\\
                              e^{\pi i(\alpha-2\beta)}\varphi(\zeta)^{2\beta}e^{-\frac{1}{2}s\zeta}&1
                                \end{pmatrix},  & \zeta<-1. \\[.4cm]

 \end{array}  \right .
\end{equation*}
Moreover, $\hat{\Psi}^{(2)}(\zeta)$ has at most logarithmic singularities at $\pm 1$ and fulfills  the asymptotic estimates
 \begin{equation}\label{eq:hat-Psi-2-0}
 \hat{\Psi}^{(2)}(\zeta)=O(1)\zeta^{\alpha\sigma_3} \quad\mbox{as}~~\zeta\in \Omega_1,~\zeta\to 0,
\end{equation}
and
 \begin{equation}\label{eq:hat-Psi-2-infty}
 \hat{\Psi}^{(2)}(\zeta)=
O\left(\frac{1}{\zeta}\right) \quad\mbox{as}~~\zeta\to\infty.
   \end{equation}

 Next, we consider
   \begin{equation}\label{eq:Q}
Q(\zeta)= \hat{\Psi}^{(2)}(\zeta)\left( \hat{\Psi}^{(2)}(\bar{\zeta})\right)^*, \zeta\in \mathbb{C}\setminus\mathbb{R},
   \end{equation}
where $\left( \hat{\Psi}^{(2)}(\zeta)\right)^*$ denotes the Hermitian conjugate of $\hat{\Psi}^{(2)}(\zeta)$.
We see that $Q(\zeta)$ is analytic for $\zeta\in \mathbb{C}\setminus\mathbb{R}$. For $\alpha>-1/2$, we have
\begin{equation}\label{eq:integral Q}
\int_{\mathbb{R}}Q_{+}(x)dx= \int_{\mathbb{R}} \hat{\Psi}^{(2)}_{+}(x)\left(\hat{J}^{-1}(x)\right)^*\left( \hat{\Psi}^{(2)}_+(x)\right)^* dx=0,\end{equation}
where the jump $\hat{J}(x)$ is  introduced in \eqref{eq:hat-Psi-2-jump}.
 Adding to \eqref{eq:integral Q} its Hermitian conjugate, for purely imaginary $s$ and $\beta$, we have
 \begin{equation}\label{eq:integral Q-Re part}
\int_{-\infty}^{-1} \hat{\Psi}^{(2)}_{+}(x) \begin{pmatrix}
                                1 &0 \\
                               0&0
                                 \end{pmatrix}
                                 \left( \hat{\Psi}^{(2)}_+(x)\right)^* dx+\int^{\infty}_{1} \hat{\Psi}^{(2)}_{+}(x) \begin{pmatrix}
                                1 &0 \\
                               0&0
                                 \end{pmatrix}
                                 \left( \hat{\Psi}^{(2)}_+(x)\right)^* dx=0. \end{equation}
  Thus, the first column of $\hat{\Psi}^{(2)}_{+}(x)$  vanishes for $x>1$ and $x<-1$.  From \eqref{eq:hat-Psi-2-jump} one sees that  $\hat{\Psi}^{(2)}(\zeta)$ can
  be analytically extended  from the upper-half   complex plane to the lower-half, acrossing $(1,\infty)$. Thus,  the first column of $\hat{\Psi}^{(2)}(\zeta)$ vanishes in the upper-half $\zeta$-plane. This and the jump condition \eqref{eq:hat-Psi-2-jump} then imply the vanishing of the second column of $\hat{\Psi}^{(2)}(\zeta)$ in the lower-half   plane.

To show the vanishing of the remaining column of   $\hat{\Psi}^{(2)}(\zeta)$ in the whole plane, we consider two
scalar functions defined as
\begin{equation}\label{def: g-k}
g_k(\zeta) =\left\{\begin{array}{ll}
-\left(\hat{\Psi}^{(2)}(\zeta)\right)_{k2}e^{-\frac{s\zeta}{4}}\varphi(\zeta)^{2\beta}e^{-\pi i\alpha}, & \Im \zeta>0,\\
\left(\hat{\Psi}^{(2)}(\zeta)\right)_{k1}e^{\frac{s\zeta}{4}},& \Im \zeta<0,
\end{array}\right.\end{equation}
for $k=1,2$.
Then,  it follows from \eqref{eq:hat-Psi-2-jump} and \eqref{eq:hat-Psi-2-infty}  that $g_k(\zeta)$ is analytic and bounded
for $\zeta\not\in(-\infty,1]$ and $g_k(\zeta)=O(e^{-|s\zeta|/4})$ for purely imaginary $\zeta$
 and $is>0$. Applying Carlson's theorem,
we have $g_k(\zeta)=0$ for $k=1,2$. This then implies that the second column of   $\hat{\Psi}^{(2)}(\zeta)$ vanishes in the upper-half $\zeta$-plane, and the first column of   $\hat{\Psi}^{(2)}(\zeta)$ also vanishes in the lower-half  plane. This completes the proof the lemma.
   \end{proof}

By a standard analysis \cite{dkmv1,fmz,fz,z}, the vanishing lemma implies the existence of  unique solution to the RH problem for $\Psi(\zeta;s)$ for
the parameters $\alpha>-1/2$, $\beta\in i \mathbb{R}$ and $s\in -i(0, +\infty)$.

\begin{prop}\label{pro: existence of solution}
 For $\alpha>-1/2$, $\beta\in i \mathbb{R}$ and $s\in -i(0, +\infty)$, there exists unique solution to the RH problem for $\Psi(\zeta;s)$.
Then, $\Psi(\zeta;s)$ is a global function of $s$ on $-i(0,+\infty)$, that is, it is free of poles  for $s\in -i(0, +\infty)$. Particularly, the Hamiltonian $H(s)$  \eqref{eq:H-psi} is free of poles for $s\in -i(0, +\infty)$.
 \end{prop}


 \subsection{Special function solutions}\label{subsec:special function-sol}
Now we   show that the coupled Painlev\'e V system \eqref{int:cpv} admits special function solutions for some particular parameters.
The special solutions are useful in our derivation of the total integral of the Hamiltonian.

For $\beta=0$, the jump on $\Sigma_4$ vanishes; cf. \eqref{eq:Psi-jump}. It is readily verified  that $\sigma_1\Psi(-\zeta)\sigma_1$ also solves  the RH problem for $\Psi(\zeta)$, formulated in Section \ref{sec:model-RHP}.
 Then, it follows from the uniqueness of solution to the RH problem for $\Psi(\zeta)$ that
\begin{equation}\label{eq: symmetry}
\sigma_1\Psi(-\zeta)\sigma_1=\Psi(\zeta),
\end{equation}where $\sigma_1$ is the Pauli matrix; cf. \eqref{Pauli-matrix}.
Substituting   \eqref{eq: symmetry} into \eqref{eq: Psi-origin},  we have the constraint
$\Psi^{(0)}_0(s)=\sigma_1 \Psi^{(0)}_0(s)\sigma_3$. Noting that $\det\Psi^{(0)}_0(s)=1$, we can write
\begin{equation}\label{eq: Psi-0}
\Psi^{(0)}_0(s)=\frac{ I-i\sigma_2 }{\sqrt{2}}l(s)^{\sigma_3}
\end{equation}
for a certain scalar  function $l(s)$.
We define
\begin{equation}\label{def: P}
P(\zeta)=e^{-\frac{1}{4}\pi i\sigma_3}\zeta^{-\frac{1}{4}\sigma_3}\frac{(I+i\sigma_2)}{\sqrt{2}}\Psi(\sqrt{\zeta})e^{\frac{1}{4}\pi i\sigma_3}\left\{ \begin{array}{ll}
                      I, & \hbox{$0<\arg \zeta<\pi$,} \\
                    \left(\begin{array}{cc}
                                                             0 & -1 \\
                                                         1 & 0
                                                           \end{array}\right), & \hbox{$-\pi<\arg \zeta<0$.}
                    \end{array}
\right.
\end{equation}
Then $P(\zeta)$  satisfies the following model RH problem.

\begin{figure}[ht]
 \begin{center}
 \includegraphics[width=8 cm]{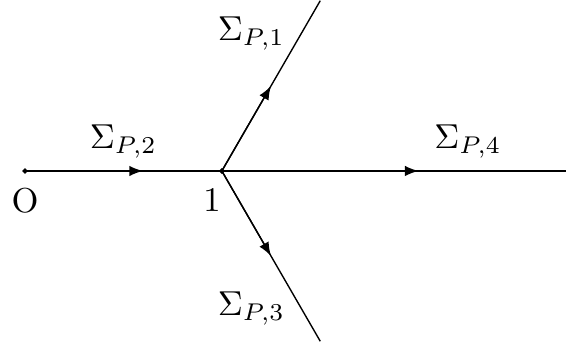} \end{center}
  \caption{\small{Contours  for the RH problem of $P(\zeta)$ }}
 \label{figure: Bessel model RHP}
 \end{figure}

\begin{description}
  \item (a)   $P(\zeta)$ is analytic in
  $\mathbb{C}\backslash \left\{\bigcup_{k=1}^4\Sigma_{P,k}\right\}$, where
  $ \Sigma_{P,1}=1+e^{\frac{ \pi i}{3}}\mathbb{R}^{+}$, $ \Sigma_{P,2}=(0,1)$,
  $ \Sigma_{P,3}=1+e^{-\frac{\pi i }{3}}\mathbb{R}^{+}$ and  $ \Sigma_{P,4}=(1,+\infty)$, as illustrated in Figure \ref{figure: Bessel model RHP}.
  \item (b)   $P(\zeta)$  satisfies the jump condition
  \begin{equation}\label{eq:P-jump}
P_+(\zeta)=P_-(\zeta) J_{P,k}(\zeta), \quad  \zeta\in \Sigma_{P,k}, \quad k=1,2,3,4,
  \end{equation}
where
\begin{equation*}
 J_{P,1}=\left(\begin{array}{cc}
                                                             1 & 0 \\
                                                         e^{-\pi i\left(\alpha-\frac 1 2\right)} & 1
                                                           \end{array}\right), ~ J_{P,2}=e^{\pi i\left(\frac 1 2-\alpha\right )\sigma_3},~ J_{P,3}=\left(\begin{array}{cc}
                                                             1 & 0 \\
                                                            e^{\pi i\left(\alpha-\frac 1 2\right)} & 1
                                                           \end{array}\right),~ J_{P,4}=i\sigma_2.
\end{equation*}

  \item (c)  $P(\zeta)$ has the asymptotic behavior as $\zeta\to\infty$
  \begin{equation}\label{eq: P-infinity-1}
P(\zeta)=\zeta^{-\frac{1}{4}\sigma_3}\frac{I-i\sigma_1}{\sqrt{2}}
   \left (I+O\left (\frac 1{\sqrt{\zeta}}\right )\right)e^{\frac{1}{4}s\sqrt{\zeta}\sigma_3}, \quad 0<\arg \zeta<\pi , \end{equation}
and
 \begin{align}\label{eq: P-infinity-2}
P(\zeta)&=\zeta^{-\frac{1}{4}\sigma_3}\frac{I-i\sigma_1}{\sqrt{2}}
   \left (I+O\left (\frac 1{\sqrt{\zeta}}\right )\right)e^{\frac{1}{4}s\sqrt{\zeta}\sigma_3}(-i\sigma_2) \nonumber\\
   &=\zeta^{-\frac{1}{4}\sigma_3}e^{-\frac{\pi i}{2}\sigma_3}\frac{I-i\sigma_1}{\sqrt{2}}
   \left (I+O\left (\frac 1{\sqrt{\zeta}}\right )\right)e^{-\frac{1}{4}s\sqrt{\zeta}\sigma_3},
   \quad -\pi<\arg \zeta<0. \end{align}

  \item (d)  $P(\zeta)$ satisfies the asymptotic behavior  near the origin
  \begin{equation}\label{eq: P-zero}
P(\zeta)=(I+c(s)\sigma_+)\left(I+O(\sqrt\zeta)\right)l(s)^{\sigma_3}
\zeta^{(\frac{\alpha}{2}-\frac{1}{4})\sigma_3},\end{equation}where $c(s)$ is a scalar function depending only on $s$, $l(s)$ is introduced in \eqref{eq: Psi-0}, and the branch of the power function is chosen such that $\arg\zeta\in (0, 2\pi)$.
  \item (e)  $P(\zeta)$ satisfies the asymptotic behavior  as $\zeta\to 1$
  \begin{equation}\label{eq: P-1}
P(\zeta)=O(\ln(\zeta-1)).\end{equation}
\end{description}

The RH problem is equivalent to the model RH problem considered by the authors in \cite{xz-2015}. We can write \begin{equation}\label{eq: P-Psi-xz}
P(\zeta;s)=ie^{-\frac{\pi i}4\sigma_3}2^{-\frac 1 2\sigma_3}\sigma_3\sigma_1\Psi_0^{XZ}\left (\frac{1-\zeta}4,is\right )\sigma_3,\end{equation}
where $\Psi_0^{XZ}(\zeta,s)$ is the solution to the model problem for $\Psi_0(\zeta,s)$, formulated  in \cite[Section\;1.2]{xz-2015}, with parameters $\gamma=-1/2$ and $\Theta=1/2-\alpha$  therein.

If we further assume that the parameter $\alpha= {1}/{2}$, the RH problem for $P(\zeta)$ can be solved explicitly in terms of the modified Bessel functions
\begin{equation}\label{eq: P-solution-special}
P(\zeta)=\left(I-\frac{I_0'(|s|/4)}{I_0(|s|/4)}\sigma_-\right)e^{-\frac{\pi i}4\sigma_3}\left(\frac{|s|}{4}\right)^{\frac 12 \sigma_3}\sigma_3\Phi_{B}\left(\frac{s^2}{16}(\zeta-1)\right)\sigma_3,\end{equation}
where $\Phi_B(\zeta)$ is defined in \eqref{eq:phi-B-solution}, and $\sigma_-=\begin{pmatrix}
                0& 0 \\
            1& 0     \end{pmatrix}$.  The  lower triangular matrix is chosen to ensure the fulfilment of the asymptotic behavior \eqref{eq: P-zero}.   From \eqref{eq: phi-B-infinity}, we have
the large-$\zeta$ behavior
\begin{equation}\label{eq: P-infinity-refined}
\begin{array}{rl}
 P(\zeta)= & \zeta^{-\frac{1}{4}\sigma_3}\frac{I-i\sigma_1}{\sqrt{2}}
   \left \{I+\frac {1
   }{\sqrt{\zeta}}\left(\frac{1}{2s}\left(
                     \begin{array}{cc}
                       -1 & 2i \\
                       2i & 1 \\
                     \end{array}
                   \right)-\frac{s}{8}\sigma_3-\frac{I_0'(|s|/4)}{2I_0(|s|/4)}\left(
                     \begin{array}{cc}
                       i & 1\\
                       1 &- i \\
                     \end{array}
                   \right)\right) +O\left (\frac 1{\zeta}\right ) \right\} \\
  &\times e^{\frac{1}{4}s\sqrt{\zeta}\sigma_3}
                                              \end{array}
\end{equation}
for $0<\arg \zeta<\pi$.
Moreover, it follows from \eqref{eq:Bessel-ODE}  that $P(\zeta)$ satisfies the differential equation
\begin{equation}\label{eq:P-ODE-1}
\frac{dP(\zeta)}{d\zeta}=\left(\left(\begin{array}{cc}
                                 0 & 0 \\
                                 \frac{s}{8i} & 0 \\
                                 \end{array}\right)+\frac{is}{8(\zeta-1)}
                             \left(  \begin{array}{cc}
                                 \frac{I_0'(|s|/4)}{I_0(|s|/4)} & 1\\
                             -\frac{I_0'(|s|/4)^2}{I_0(|s|/4)^2} &-\frac{I_0'(|s|/4)}{I_0(|s|/4)} \\
                                 \end{array}
                             \right)\right)P(\zeta).
\end{equation}
From the symmetry of $\Psi(\zeta)$  in \eqref{eq: symmetry} for $\beta=0$, we have
 \begin{equation}\label{eq: symmetry -beta-0}
u_1v_1=-u_2v_2, \quad u_2=-u_1v_1^2, \quad v_2=1/v_1, \quad y(s)=1. \end{equation}
Substituting \eqref{def: P} into \eqref{def: Lax pair}, we see that $P(\zeta)$ solves  the differential equation with coefficients involving  $u_1(s)$ and $v_1(s)$
\begin{equation}\label{eq:P-ODE-2}
\frac{dP(\zeta)}{d\zeta}=\left(\left(\begin{array}{cc}
                                 0 & 0 \\
                                 \frac{s}{8i} &0\\
                                 \end{array}\right)+\frac{u_1}{2(\zeta-1)}
                             \left(  \begin{array}{cc}
                                 1-v_1^2 &  -i(v_1+1)^2 \\
                       - i(v_1-1)^2  &  v_1^2-1\\
                                 \end{array}
                             \right)\right)P(\zeta).
\end{equation}

Now we are in a position   to derive the special function solutions to the  coupled Painlev\'e system of equations \eqref{int:cpv} for $\alpha=1/2$ and $\beta=0$.
The integral of the Hamiltonian associated with the special function solutions is explicitly calculated, and furnishes as  boundary condition in  our evaluation of the total integral of the Hamiltonian for general parameters later.

\begin{prop}\label{pro: special function solution}
For $\alpha=1/2$ and $\beta=0$, the coupled Painlev\'e system of equations \eqref{int:cpv} admits the following special function solutions
\begin{equation}\label{eq:u-special}u_1(s)=\frac{s}{16}\left(i+\frac{I_0'(|s|/4)}{I_0(|s|/4)}\right)^2, \quad u_2(s)=-\frac{s}{16}\left(i-\frac{I_0'(|s|/4)}{I_0(|s|/4)}\right)^2,
\end{equation}
\begin{equation}\label{eq:v-special}v_1(s)=\frac{i-\frac{I_0'(|s|/4)}{I_0(|s|/4)}}{i+\frac{I_0'(|s|/4)}{I_0(|s|/4)}}, \quad v_2(s)=\frac{i+\frac{I_0'(|s|/4)}{I_0(|s|/4)}}{i-\frac{I_0'(|s|/4)}{I_0(|s|/4)}},
\end{equation}
and the Hamiltonian associated with $u_k(s)$ and $v_k(s)$ takes the form
\begin{equation}\label{eq: H-special solution}
H\left(u_1,v_1,u_2,v_2,s;  {1}/{2},0\right)=\frac{s}{16}+\frac{i}{4}\frac{I_0'(|s|/4)}{I_0(|s|/4)},\end{equation}
where
\begin{equation*}
 I_0(s)=\sum_{n=0}^{\infty}\frac 1 {n!^2}\left(\frac{s^2 }{4}\right)^n
\end{equation*}
is the modified Bessel function of order zero.
\end{prop}
\begin{proof}
In this specific case when
  $\alpha=1/2$ and $\beta=0$,  the special function solutions  $u_1(s)$ and $v_1(s)$  in \eqref{eq:u-special} and \eqref{eq:v-special} are readily derived by comparing \eqref{eq:P-ODE-1} with \eqref{eq:P-ODE-2}. We further obtain   $u_2(s)$ and $v_2(s)$  in \eqref{eq:u-special} and \eqref{eq:v-special}
  by using the relation \eqref{eq: symmetry -beta-0}. The special Hamiltonian of the form \eqref{eq: H-special solution} follows from \eqref{eq:H}, \eqref{eq:Psi-infty}, \eqref{eq: P-infinity-1} and \eqref{eq: P-infinity-refined}.    This completes the proof of the proposition.
  \end{proof}


  \section{Asymptotics of the coupled Painlev\'e V system as $is\to 0^{+}$}\label{sec:small-s}
In this section, we  derive  the  asymptotic approximations  of the solutions to the  coupled Painlev\'e V system as $is\to 0^{+}$ by performing Deift-Zhou nonlinear steepest descent analysis  \cite{deift, dkmv2,dz} of  the RH problem  for $\Psi(\zeta;s)$.

\subsection{Outer parametrix}

Let
 \begin{equation}\label{eq:X}
X(\xi; s)=\left(\frac{|s|}{2}\right)^{-\beta\sigma_3}\Psi(2\xi /|s|; s).
  \end{equation}
  Then, $X(\xi)$ satisfies a re-scaled version of the RH problem for $\Psi(\zeta)$ in the complex $\xi$-plane with $\xi=|s|\zeta/2$. Note that,
  the jump contour $(-1,1)$ in the $\zeta$-plane is mapped onto  $(-|s|/2, |s|/2)$ in the $\xi$-plane. As $is \to 0^{+}$,  we consider the
  following approximate RH problem $\Phi(\xi)$ by ignoring the jump condition  for $X(\xi)$ along the shrinking line segment $(-|s|/2, |s|/2)$.

\subsubsection*{RH problem  for $\Phi$}
\begin{description}
  \item(a)  $\Phi(\xi)$ is analytic in
  $\mathbb{C}\setminus \left\{\cup^5_{j=1}\Sigma_{\Phi,j}\right\}$, where the oriented contours
         \begin{equation*}
     \Sigma_{\Phi,1}=e^{\frac{\pi   i}{4}}\mathbb{R}^{+}, \quad \Sigma_{\Phi,2}=e^{\frac{3\pi i}{4}}\mathbb{R}^{+}, \quad \Sigma_{\Phi,3}=e^{-\frac{3\pi i}{4}}\mathbb{R}^{+},\quad
  \Sigma_{\Phi,4}=e^{-\frac{ \pi i}{2}}\mathbb{R}^{+},\quad \Sigma_{\Phi,5}=e^{-\frac{\pi i  }{4}}\mathbb{R}^{+},\end{equation*}
as illustrated in Figure \ref{CFH-contour}.
  \begin{figure}[th]
 \begin{center}
   \includegraphics[width=6cm]{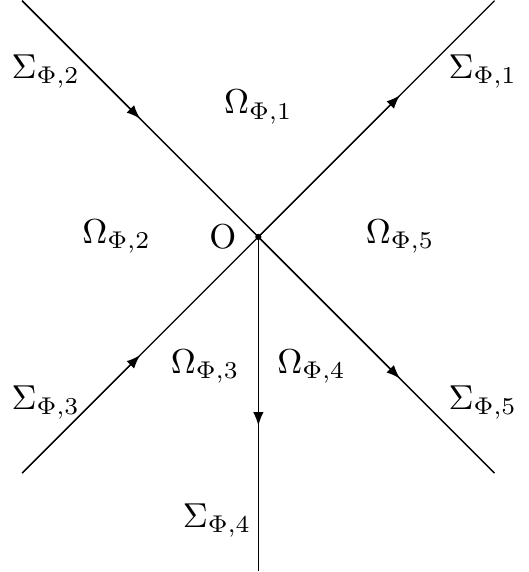} \end{center}
 \caption{\small{The jump contours and regions of the RH problem for $\Phi(\xi)$ }}
 \label{CFH-contour}
\end{figure}

  \item(b)  $\Phi(\xi)$ satisfies the jump conditions
 \begin{equation}\label{eq:Phi-jump}
 \Phi_+(\xi)=\Phi_-(\xi)
 \left\{
 \begin{array}{ll}
    \begin{pmatrix}
                                 1 & 0 \\
                               e^{-\pi i(\alpha-\beta)}& 1
                                 \end{pmatrix}, &  \xi \in \Sigma_{\Phi,1}, \\[.5cm]
    \begin{pmatrix}
                                1 &0\\
                              e^{\pi i(\alpha-\beta)}&1
                                \end{pmatrix},  &  \xi \in \Sigma_{\Phi,2}, \\[.5cm]
    \begin{pmatrix}
                                 1 &-e^{-\pi i(\alpha-\beta)} \\
                                 0 &1
                                 \end{pmatrix}, &   \xi \in \Sigma_{\Phi,3}, \\[.5cm]

     e^{2\pi i\beta\sigma_3} , &   \xi \in \Sigma_{\Phi,4}, \\ [.2cm]
                               \begin{pmatrix}
                                 1 &-e^{\pi i(\alpha-\beta)} \\
                                 0 &1
                                 \end{pmatrix}, &   \xi \in \Sigma_{\Phi,5}.

 \end{array}  \right .  \end{equation}

\item(c) As $\xi\to \infty$,  we have
  \begin{equation}\label{eq:Phi-infty}
 \Phi(\xi)=\left( I  +\frac{\Phi_1}{\xi}+
O\left(\frac{1}{\xi^2}\right) \right) \xi^{-\beta\sigma_3} e^{-\frac{i}{2}\xi \sigma_3},
   \end{equation}
 where $\Phi_1=  \begin{pmatrix}
                            - ( \alpha^2-\beta^2)i    &-e^{-\pi i\beta}  \frac{\Gamma(1+\alpha-\beta)}{i\Gamma(\alpha+\beta)}
\\
      e^{\pi i\beta}  \frac{\Gamma(1+\alpha+\beta)}{i\Gamma(\alpha-\beta)}                            &  (\alpha^2-\beta^2 )i
                                 \end{pmatrix}$ and
 $\arg \xi \in (-\pi/2, 3\pi/2)$.

 \item(d) The description of the asymptotic behavior  of $\Phi(\xi)$ as $\xi\to 0$  is divided into two cases, namely
 $ 2\alpha \not\in \mathbb{N}$ and   $ 2\alpha  \in \mathbb{N}$.

  In the first case, if $ 2\alpha \not\in \mathbb{N}$,  there exists a  function $\Phi^{(0)}(\xi)$, analytic at $\xi=0$,  such that
  \begin{equation}\label{eq: Phi-origin}
  \Phi(\xi)= \Phi^{(0)}(\xi)\xi^{\alpha \sigma_3} C_j,~~~\xi\to 0,~\xi\in \Omega_{\Phi, j}, ~j=1,2,3,4,5, \end{equation}
  where $\Omega_{\Phi, j}$ are the sectors illustrated in Figure \ref{CFH-contour}, the branch of $\xi^{\alpha \sigma_3}$ is chosen such that
  $\arg \xi \in (-\pi/2, 3\pi/2)$,
  the constant matrix
  $$C_1= \begin{pmatrix}
                                 1 &\frac{\sin(\pi(\alpha+\beta))}{\sin (2\pi \alpha)} \\
                                 0 &1
                                 \end{pmatrix}, \quad  2\alpha \not\in \mathbb{N},\quad \alpha>\frac 1 2, $$
and the other constant matrices are determined by $C_1$ and the jump conditions.

In the second case, if $ 2\alpha \in \mathbb{N}$,  there exists a  function $\widehat{\Phi}^{(0)}(\xi)$, analytic at $\xi=0$,  such that
  \begin{equation}\label{eq: Phi-origin-2}
  \Phi(\xi)= \widehat{\Phi}^{(0)}(\xi)\xi^{\alpha \sigma_3} \begin{pmatrix}
                                 1 &(-1)^{2\alpha}\frac{\sin(\pi(\alpha+\beta))}{\pi} \ln  \xi  \\
                                 0 &1
                                 \end{pmatrix} \hat{C}_j,~~~\xi\to 0,~\xi\in \Omega_{\Phi, j}, ~j=1,2,3,4,5, \end{equation}
  where $\Omega_{\Phi, j}$ are the sectors illustrated in Figure \ref{CFH-contour}, the branches of $\xi^{\alpha \sigma_3}$ and $\ln \xi $  are chosen such that
  $\arg \xi \in (-\pi/2, 3\pi/2)$,
  the constant matrix $\hat{C}_1$ is the identity matrix and the other constant matrices are determined by $\hat{C}_1$ and the jump conditions.
\end{description}

The connection matrices $C_j$ are specified so that the asymptotic behavior of $\Phi$ near the origin be in accordance with  the jump conditions, that is,
for $\xi\in \Sigma_{\Phi,4}$, we have the constraint
\begin{equation}\label{eq:Phi-jump-origin}
\Phi_{-}(\xi)^{-1}\Phi_{+}(\xi)=J_{3}J_{2} C_1^{-1}e^{-2\pi i\alpha\sigma_3}C_1J_{1}^{-1}J_{5}^{-1}=e^{2\pi i\beta\sigma_3},
\end{equation}where $J_{j}$ denote  the jumps on $\Sigma_{\Phi,j}$, given in \eqref{eq:Phi-jump}. These are the same as the jumps on $\Sigma_{j}$; see \eqref{eq:Psi-jump}. Similarly, we readily verify that the jump conditions are also fulfilled by \eqref{eq: Phi-origin-2} in the case $2\alpha\in \mathbb{N}$.

The solution to the above RH problem can be constructed explicitly in terms of the confluent hypergeometric function; see \cite{DIK1} and \cite{cik}.  The parametrix is equivalent to the one, namely $M(\xi)$, used in  \cite[Section\;4.2.1]{cik} after some elementary transfromation
 \begin{equation*}
 \Phi(\xi)=e^{\frac{\pi i \beta}{2}\sigma_3}\sigma_3M(e^{\frac{\pi i}{2}}\xi)\sigma_3, \quad \arg \xi\in (-\pi/2,3\pi/2).
 \end{equation*}
The analytic function $\Phi^{(0)}(\xi)$ brought in \eqref{eq: Phi-origin} now takes
\begin{equation}\label{eq: Phi-origin-0}{\small{
\Phi^{(0)}(\xi)= e^{-\frac{i\xi}{2}}  \begin{pmatrix}
                                 e^{-\frac{\pi i(\alpha+\beta)}{2}} \frac{\Gamma(1+\alpha-\beta)}{\Gamma(1+2\alpha)}
                                 \psi(\alpha+\beta,1+2\alpha,i\xi)
                                 &- e^{\frac{\pi i(\alpha-\beta)}{2}}\frac{\Gamma(2\alpha)}{\Gamma(\alpha+\beta)}\psi(-\alpha+\beta,1-2\alpha,i\xi) \\
                               e^{-\frac{\pi i(\alpha-\beta)}{2}}   \frac{\Gamma(1+\alpha+\beta)}{\Gamma(1+2\alpha)}\psi(1+\alpha+\beta,1+2\alpha,i\xi)
                                &e^{\frac{\pi i(\alpha+\beta)}{2}} \frac{\Gamma(2\alpha)}{\Gamma(\alpha-\beta)}\psi(1-\alpha+\beta,1-2\alpha,i\xi)
                                 \end{pmatrix},}}\end{equation}
where $\psi$ is the confluent hypergeometric function given in \eqref{def: phi-CHF}.
  Also,  the first column of $\hat{\Phi}^{(0)}(\xi)$ in \eqref{eq: Phi-origin-2} is the same as that of $\Phi^{(0)}(\xi)$; see \cite[Section\;4.2.1]{cik}.

It is straightforward to see that $\Phi(\xi)$ shares  the same jump conditions with $X(\xi)$ for $|\xi|>|s|/2$. For some small constant radius $\epsilon>0$,  we intend to  construct a local parametrix $L(\xi)$ in a neighborhood $U(0;\epsilon)$ of the origin,  which shares the same jump conditions as $X(\xi)$ and
  matches with $\Phi(\xi)$ on the boundary of  $U(0;\epsilon)$, namely, on $|\xi|=\epsilon$.

\subsection{Local parametrix}

      \subsubsection*{RH problem for $L$}
\begin{description}
  \item (a)  $L(\xi)$ is analytic in $U(0;\epsilon ) \setminus \Sigma^{X}$, where $ \Sigma^{X}$ denotes the jump contours for $X(\xi)$; see Figure
      \ref{Model RH contour} for a re-scaled version.

  \item (b) $L(\xi)$  satisfies the same jump conditions  on $U(0;\epsilon ) \cap \Sigma^{X}$  as $X(\xi)$.
   \item (c)  As $is\to 0^+$, we have the matching condition
     \begin{equation}\label{eq:matching condition-H}
L(\xi)=\left(I+O(s^{2\alpha+1})\right)\Phi(\xi), \end{equation}
for $|\xi|=\epsilon$ with  a constant $\epsilon>0$.

\end{description}

From the asymptotic behavior
\eqref{eq: Phi-origin} and \eqref{eq: Phi-origin-2} of $\Phi(\xi)$ near the origin, we may  seek a solution to the above RH problem of the form
\begin{equation}\label{eq:L}
L(\xi)= \Phi^{(0)}(\xi)  \begin{pmatrix}
                                 1 &|s|^{2\alpha}k(\xi/|s|) \\
                                 0 &1
                                 \end{pmatrix}\xi^{\alpha \sigma_3} C_j \quad \mbox{if}  \quad  2\alpha\not\in \mathbb{N},\quad \alpha>-\frac 12, \end{equation}
or
  \begin{equation}\label{eq:L-integer}
L(\xi)= \widehat{\Phi}^{(0)}(\xi)  \begin{pmatrix}
                                 1 &|s|^{2\alpha}k(\xi/|s|) \\
                                 0 &1
                                 \end{pmatrix}\xi^{\alpha \sigma_3} \begin{pmatrix}
                                 1 &\frac{(-1)^{2\alpha}\sin(\pi(\alpha+\beta))}{\pi} \ln \xi   \\
                                 0 &1
                                 \end{pmatrix}\hat{C}_j  \quad \mbox{if} \quad 2\alpha\in\mathbb{N}, \end{equation}
                                 for $ \xi/|s| \in \Omega_j$,  $j=1,2,3,4,5$; cf.\;Figure
                                 \ref{Model RH contour},
   where  $\Phi^{(0)}(\xi)$ and $C_j$ are given in \eqref{eq: Phi-origin-0} and \eqref{eq: Phi-origin},
   $ \widehat{\Phi}^{(0)}(\xi) $ and $\hat{C}_j$ are given in \eqref{eq: Phi-origin-2}.
   The branches of $\xi^{\alpha}$ and $\ln \xi $ are chosen so that $\arg \xi\in(-\pi/2, 3\pi/2)$. The scalar function $k(\xi)$, to be determined, is analytic  in $\mathbb{C}\setminus [-1/2,1/2]$.

   In view of  \eqref{eq:L} and \eqref{eq:Phi-jump-origin}, we see that $L(\xi)$ has the  same jumps as $X(\xi)$ on  $\Sigma^{X} \setminus [-|s|/2,|s|/2]$.
To  ensure the fulfilment by $L(\xi)$  of the jump condition  on  the interval $(-|s|/2,|s|/2)$ as $X(\xi)$, and the matching condition \eqref{eq:matching condition-H}, we look for the function $k(\xi)$, as  a solution of the following scalar RH problem for $\alpha>-1/2$.

   \subsubsection*{RH problem for $k$}
\begin{description}
  \item (a)  $k(\xi)$ is analytic in $\mathbb{C}\setminus [-1/2,1/2]$.

  \item (b) $k(\xi)$  satisfies the same jump relations
   \begin{equation}\label{eq:k-infty}
k_{+}(x)-k_{-}(x)=\left\{\begin{array}{ll}
                            -e^{\pi i(\alpha-\beta)}x^{2\alpha},  & x\in (0, 1/2), \\
                            -e^{\pi i(\alpha+\beta)}|x|^{2\alpha}, &  x\in (-1/2, 0).
                         \end{array}\right. \end{equation}
\item (c)  As $\xi\to \infty$,  $k(\xi)=O(1/\xi)$.
 \end{description}
 Applying the Sokhotski-Plemelj formula, we obtain an explicit representation
 \begin{equation}\label{eq:k-solution}
k(\xi)=-\frac{e^{\pi i(\alpha-\beta)}}{2\pi i}\int_0^{1/2} \frac{x^{2\alpha} dx}{x-\xi}
-\frac{e^{\pi i(\alpha+\beta)}}{2\pi i}\int_{-1/2}^0 \frac{ |x|^{2\alpha} dx}{x-\xi}\quad\mbox{for}\quad
\xi \in\mathbb{C}\setminus [-1/2,1/2].  \end{equation}
Having had $k(\xi)$ as in \eqref{eq:k-solution}, we readily verify that $L(\xi)$ in \eqref{eq:L}
shares the same jumps as   $X(\xi)$  on real intervals. Indeed,
for $\xi\in(0,|s|/2)$, we have
\begin{equation}\label{eq:jump-L-1}
L_{-}(\xi)^{-1}L_{+}(\xi)=J_5J_1(I-e^{\pi i(\alpha-\beta)}\sigma_+)=J_6, \end{equation}
where $J_k$ denotes the constant jump of $\Psi(\zeta)$ on the contour $\Sigma_k$; cf.\;\eqref{eq:Psi-jump} and \eqref{eq:Phi-jump}.
Now we turn to $\xi\in(-|s|/2,0)$.
First, with the branch  chosen such that $\arg \xi\in(-\pi/2,3\pi/2)$, we have
\begin{equation*}
 \xi_+^{\alpha}=\xi_{-}^{\alpha}=|\xi|^\alpha e^{\pi i\alpha} \quad \mbox{and} \quad \left(\ln \xi\right)_{+} =\left(\ln\xi \right )_{-}
\end{equation*}for $\xi$ on the negative real axis.
Therefore,  we have
\begin{equation}\label{eq:jump-L-2}
L_{-}(\xi)^{-1}L_{+}(\xi)=J_3J_2(I-e^{-\pi i(\alpha-\beta)}\sigma_+)=J_7. \end{equation}
Thus,  $L(\xi)$ has the same jump   as $X(\xi)$ for $\xi\in(-|s|/2,0)\cup(0, |s|/2)$, and hence for   $\xi\in U(0;\epsilon ) \cap \Sigma^{X}$.
Accordingly,  we  see that $X(\xi)L(\xi)^{-1}$  is analytic in $0< |\xi|<\epsilon$.

To show the  analyticity of  $X(\xi)L(\xi)^{-1}$ at $\xi=0$,  we study the behavior of $L(\xi)$ near the origin. From \eqref{eq:k-solution}, we have
for $2\alpha\not\in \mathbb{N}$
\begin{equation}\label{eq:k-asy-zero}
k(\xi)=k_r(\xi)+\left\{
\begin{array}{ll}
\displaystyle{  -\frac{\sin ( \pi(\alpha+\beta))}{\sin (2\pi \alpha)}\xi^{2\alpha}, }& \arg \xi \in(0,\pi), \\[.3cm]
\displaystyle{  \frac{\sin ( \pi(\alpha-\beta))}{\sin (2\pi \alpha)}\left(\xi e^{\pi i}\right)^{2\alpha},} & \arg \xi \in(-\pi, 0),
\end{array}\right .
\end{equation}
where $k_r(\xi)$ is an analytic function in   $|\xi|<1/2$.
Similarly, for $2\alpha\in\mathbb{N}$, we obtain from  \eqref{eq:k-solution} that
\begin{equation}\label{eq:k-asy-zero-N}
k(\xi)=
\hat{k}_r(\xi)+\left\{
\begin{array}{ll}
\displaystyle{\frac{(-1)^{2\alpha+1}\sin(\pi(\alpha+\beta))}{\pi}  \xi^{2\alpha}\ln \xi },& \arg \xi \in(0,\pi), \\[.3cm]
\displaystyle{  \frac{(-1)^{2\alpha+1}\sin(\pi(\alpha+\beta))}{\pi}  \xi^{2\alpha}\ln \xi }+e^{\pi i(\alpha-\beta)}\xi^{2\alpha}, & \arg \xi \in(-\pi, 0),
\end{array}\right .
\end{equation}
where $\hat{k}_r(\xi)$ is an analytic function in   $|\xi|<1/2$.

For $2\alpha\not\in \mathbb{N}$ and $\alpha>-1/2$, a combination of \eqref{eq: Psi-origin}, \eqref{eq:X}, \eqref{eq: Phi-origin}, \eqref{eq:L} and
\eqref{eq:k-asy-zero} shows that $X(\xi)L(\xi)^{-1}$ is  bounded in a neighborhood of   $ \xi=0$.  While for $2\alpha\in\mathbb{N}$, it follows from \eqref{eq: Psi-origin}, \eqref{eq:X},
\eqref{eq: Phi-origin-2}, \eqref{eq:L-integer} and
\eqref{eq:k-asy-zero-N}
that  $X(\xi)L(\xi)^{-1}$ is also bounded
in a neighborhood of   $ \xi=0$.
 Thus, the singularity at the   origin is  removable. Therefore, $X(\xi)L(\xi)^{-1}$ is  analytic for $ |\xi|<\epsilon$.

For later use, we  also derive from \eqref{eq:k-solution} that
\begin{equation}\label{eq:k-asy--1}
k(\xi)\sim  \frac{e^{\pi i(\alpha+\beta)}}{2\pi i}2^{-2\alpha}\ln\left(\xi+\frac{1}{2}\right)   \quad \mbox{as}\quad\xi\to -\frac 1 2,  \end{equation}
and
\begin{equation}\label{eq:k-asy-1}
k(\xi)\sim- \frac{e^{\pi i(\alpha-\beta)}}{2\pi i}2^{-2\alpha}\ln\left(\xi-\frac{1}{2}\right)   \quad \mbox{as}\quad \xi\to \frac 1 2.  \end{equation}
\subsection{Final transformation}
Next, we define
\begin{equation}\label{eq:Z}
  Z(\xi)=\left\{ \begin{array}{ll}
                      X(\xi)\Phi(\xi)^{-1}, & \hbox{$|\xi|>\epsilon$,} \\
                     X(\xi)L(\xi)^{-1}, & \hbox{$|\xi|<\epsilon$.}
                    \end{array}
\right.
 \end{equation}
As shown earlier in this section,   $Z(\xi)$ is analytic for $ |\xi|\neq\epsilon$.  From \eqref{eq: Phi-origin}, \eqref{eq:L}, \eqref{eq:k-solution}  and \eqref{eq:Z}, we have  the       estimate for the jump of $Z(\xi)$
\begin{equation}\label{eq:Z-jump}
J_Z(\xi)=I+c_{\alpha,\beta}\frac{|s|^{2\alpha+1}}{\xi}\Phi^{(0)}(\xi)\sigma_{+}\Phi^{(0)}(\xi)^{-1}+
O\left(s^{2\alpha+2}\right) \quad\mbox{as}\quad is\to 0^{+}, \end{equation}
where the constant $c_{\alpha,\beta}= \frac{e^{\pi i\alpha}\cos(\beta\pi) }{i\pi 2^{2\alpha+1}(2\alpha+1) }  $, and  the error term is uniform for $\xi$ on the   clockwise jump contour  $ |\xi|=\epsilon$. Therefore, the matching condition \eqref{eq:matching condition-H} is fulfilled. Hence, the RH problem for $Z(\xi)$ is a small-norm RH problem.
By a standard analysis  \cite{dkmv1}, we have the piecewise   approximation as $is\to 0^{+}$
\begin{equation}\label{eq:Z-solution}
  Z(\xi)=\left\{ \begin{array}{ll}
                    I+c_{\alpha,\beta}\frac{|s|^{2\alpha+1}}{\xi}\Phi^{(0)}(0)\sigma_{+}\left\{\Phi^{(0)}(0)\right\}^{-1}+
                    \varepsilon_Z(\xi), &  |\xi|>\epsilon , \\
                     I+c_{\alpha,\beta}\frac{|s|^{2\alpha+1}}{\xi}
                     \left[\Phi^{(0)}(0)\sigma_{+}\left\{\Phi^{(0)}(0)\right\}^{-1}-
                     \Phi^{(0)}(\xi)\sigma_{+}\left\{\Phi^{(0)}(\xi)\right\}^{-1}\right ]+\varepsilon_Z(\xi), & |\xi|<\epsilon  ,
                    \end{array}
\right.
 \end{equation}
 where the error term $\varepsilon_Z(\xi)= O\left(s^{2\alpha+2}\right)$ as  $is\to 0^{+}$, uniformly for all $\xi$ in the complex plane.


\subsection{Proof of Theorem \ref{thm-asy}: Asymptotics of $sH$ as $is\to 0^+$}\label{sec:small-s behavior}
Now we have the following asymptotic approximation of $X(\xi)=X(\xi;s)$, introduced in \eqref{eq:X}:
\begin{equation}\label{eq: X-solution}
 X(\xi)=\left\{ \begin{array}{ll}
                      Z(\xi)\Phi(\xi), & \hbox{$|\xi|>\epsilon$,} \\
                     Z(\xi)L(\xi), & \hbox{$|\xi|<\epsilon$;}
                    \end{array}
\right.
 \end{equation}cf. \eqref{eq:Z}.
Denoting
 \begin{equation}\label{eq: X-expand}
 X(\xi)=I+\frac{X_1(s)}{\xi}+O\left(\frac 1{\xi^2}\right) \quad\mbox{as}\quad \xi\to \infty,
 \end{equation}
 and taking  \eqref{eq:Psi-infty} and  \eqref{eq:X} into account, we have
 $$(\Psi_1)_{11}=2(X_1)_{11}/|s|.$$
 Using this together with \eqref{eq:H}, \eqref{eq:Z-solution}, \eqref{eq: X-solution} and \eqref{eq:Phi-infty}, we find
 \begin{align*}
 H&=-(\Psi_1)_{11}/2-\frac{(\alpha^2-\beta^2)}{s}\nonumber\\
& =-\frac{1}{is}\left((Z_1)_{11}+(\Phi_1)_{11}\right)-\frac{(\alpha^2-\beta^2)}{s}\nonumber\\
&=\frac{\Gamma(1+\alpha+\beta)\Gamma(1+\alpha-\beta)\cos(\pi \beta)}{i\pi 2^{2\alpha+1}(2\alpha+1)\Gamma(1+2\alpha)^2}|s|^{2\alpha}+O\left(s^{2\alpha+1}\right).
 \end{align*}This gives \eqref{thm: H-asy-small s}.

Next, we derive the small-$s$ behavior of $y(s)$, $d(s)$, and other functions. It follows from \eqref{eq:X},  \eqref{eq:Z-solution} and \eqref{eq: X-solution} that
 \begin{equation}\label{eq: X-in-appr}
\Psi(2\xi/|s|)=\left(\frac{|s|}{2}\right)^{\beta\sigma_3}X(\xi)=\left(\frac{|s|}{2}\right)^{\beta\sigma_3}
\left [I+O\left (s^{2\alpha+1}\right )\right ] L(\xi) \quad\mbox{as}\quad  is\to 0^+ \end{equation}
 for $|\xi|<\epsilon$.
Substituting  \eqref{eq:L} and  \eqref{eq:k-asy-zero} into \eqref{eq: X-in-appr}, and using the relations
\eqref{eq:y-def} and \eqref{eq:d},  we obtain
\begin{equation}\label{eq: y-asy-small s}
y(s)=\frac{\Gamma(1+\alpha-\beta)}{\Gamma(1+\alpha+\beta)}e^{-\pi i\beta}\left(\frac {|s|}2\right )^{2\beta}\left[1+O\left(s^{2\alpha+1}\right )\right], \quad is\to0^+,
\end{equation}
and
\begin{equation}\label{eq: d-asy-small s}
d(s)=2\alpha\frac{\Gamma(1+\alpha-\beta)\Gamma(1+\alpha+\beta)}{\Gamma(1+2\alpha)^2}e^{-\pi i\alpha}\left(\frac{|s|}2\right )^{2\alpha}\left[1+O\left(s^{2\alpha+1}\right)\right],  \quad is\to 0^+.
\end{equation}
Similarly, the relation \eqref{eq:Psi-A-k}, together with \eqref{eq:L},   \eqref{eq:k-asy--1},  \eqref{eq:k-asy-1} and  \eqref{eq: X-in-appr}, gives us
\begin{equation*}
u_1(s)=-\frac{\Gamma(1+\alpha+\beta)\Gamma(1+\alpha-\beta)}{2\pi i\Gamma(1+2\alpha)^2}e^{-\pi i\beta}\left(\frac{|s|}2\right )^{2\alpha}\left[1+O\left(s^{2\alpha+1}\right)+O(s)\right], \quad is\to0^+,
\end{equation*}
\begin{equation*}
v_1(s)=1+O\left(s^{2\alpha+1}\right)+O(s),  \quad is\to0^+,
\end{equation*}
\begin{equation*}
u_2(s)=\frac{\Gamma(1+\alpha+\beta)\Gamma(1+\alpha-\beta)}{2\pi i\Gamma(1+2\alpha)^2}e^{\pi i\beta}\left(\frac{|s|}2\right)^{2\alpha}\left[1+O\left(s^{2\alpha+1}\right)+O(s)\right], \quad is\to0^+,
\end{equation*}
and
\begin{equation*}
v_2(s)=1+O\left(s^{2\alpha+1}\right)+O(s),  \quad is\to0^+.
\end{equation*}
We mention that the error term $O(s)$ in the above formulas  has its roots back  in the approximation of $\Phi^{(0)}(\xi)$ in \eqref{eq: Phi-origin-0} as
$is\to 0^{+}$
$$\Phi^{(0)}(\pm |s|/2)=\Phi^{(0)}(0)+O(s).$$
These are formulas \eqref{thm: u-1-asy-small s}, \eqref{thm: v-1-asy-small s}, \eqref{thm: u-2-asy-small s} and \eqref{thm: v-2-asy-small s}, respectively. Thus, we have proved the small-$s$ part of Theorem \ref{thm-asy}.

For later use, comparing \eqref{eq: Psi-1} and \eqref{eq: Psi-1 minus} with \eqref{eq:L},   \eqref{eq:k-asy--1},  \eqref{eq:k-asy-1} and  \eqref{eq: X-in-appr},
we  also derive the error order estimates
\begin{equation}\label{eq:Psi-1-1}
\left(\Psi^{(1)}_1(s)\right)_{21}=O\left(s^{2\alpha+1}\right)+O(s),  \quad \left(\Psi^{(2)}_1(s)\right)_{21}=O\left(s^{2\alpha+1}\right)+O(s),  \quad is\to0^+.
\end{equation}
 \section{Asymptotics of the coupled Painlev\'e V system as $is\to +\infty$}\label{sec:large-s}
In this section, we  derive  the  asymptotics of the    coupled Painlev\'e V system as $is\to +\infty$ by using Deift-Zhou nonlinear steepest descent analysis of the  RH problem for $\Psi(\zeta;s)$.
\subsection{Normalization}
We introduce a normalization of $\Psi(\zeta;s)$ at infinity of the form
\begin{equation}\label{def: A}
 A(\zeta)=\Psi(\zeta;s)\exp(-sg(\zeta)\sigma_3),
 \end{equation}
 where
 \begin{equation}\label{def: g}
 g(\zeta)=\frac{1}{4}\sqrt{\zeta^2-1}
 \end{equation}is analytic in $\zeta\in \mathbb{C}\setminus[-1,1]$, such that
  $\arg(\zeta\pm 1)\in(-\pi,\pi)$.
\subsubsection*{ RH problem  for $A$}
\begin{description}
  \item(a)  $A(\zeta)$ is analytic in
  $\mathbb{C}\setminus \{\cup^7_{j=1}\Sigma_j\}$, where the contours  are shown  in Figure \ref{Model RH contour}.
 \item(b)  $A(\zeta)$ satisfies the jump conditions
  \begin{equation}\label{eq:A-jump}
 A_+(\zeta)=A_-(\zeta)
 \left\{
 \begin{array}{ll}
    \begin{pmatrix}
                                 1 & 0 \\
                               e^{-\pi i(\alpha-\beta)}e^{-2sg(\zeta)}& 1
                                 \end{pmatrix}, &  \zeta \in \Sigma_1, \\[.5cm]
    \begin{pmatrix}
                                1 &0\\
                              e^{\pi i(\alpha-\beta)}e^{-2sg(\zeta)}&1
                                \end{pmatrix},  &  \zeta \in \Sigma_2, \\[.5cm]
    \begin{pmatrix}
                                 1 &-e^{-\pi i(\alpha-\beta)} e^{2sg(\zeta)}\\
                                 0 &1
                                 \end{pmatrix}, &   \zeta \in \Sigma_3, \\[.5cm]

     e^{2\pi i\beta\sigma_3} , &   \zeta \in \Sigma_4, \\[.3cm]
                               \begin{pmatrix}
                                 1 &-e^{\pi i(\alpha-\beta)} e^{2sg(\zeta)}\\
                                 0 &1
                                 \end{pmatrix}, &   \zeta \in \Sigma_5, \\[.5cm]
                                 \begin{pmatrix}
                                0 &-e^{\pi i(\alpha-\beta)} \\
                                e^{-\pi i(\alpha-\beta)} &0
                                 \end{pmatrix}, &   \zeta \in \Sigma_6, \\[.5cm]
                                 \begin{pmatrix}
                                0 &-e^{-\pi i(\alpha-\beta)} \\
                                e^{\pi i(\alpha-\beta)} &0
                                 \end{pmatrix}, &   \zeta \in\Sigma_7.
 \end{array}  \right .  \end{equation}

\item(c) As $\zeta\to \infty$,  we have
  \begin{equation}\label{eq:A-infty}
 A(\zeta)=\left( I +
O\left(\frac{1}{\zeta}\right) \right) \zeta^{-\beta\sigma_3},
   \end{equation}
 where
 $\arg \zeta \in (-\pi/2, 3\pi/2)$.

 \item(d) The behaviors of $A(\zeta)$ near $\pm1$ and the origin are the same as $\Psi(\zeta;s)$; cf. \eqref{eq: Psi-origin}, \eqref{eq: Psi-1} and \eqref{eq: Psi-1 minus}.
  \end{description}

\subsection{Outer parametrix}

As $is\to+\infty$, the jump matrices \eqref{eq:A-jump} on the contours $\Sigma_k$, $k=1,2,3,5$ tend to the identity matrix exponentially fast.
Thus, we consider the following approximate RH problem with jumps on the remaining contours.

\subsubsection*{RH problem for $A^{(\infty)}$}
\begin{description}
  \item (a) $A^{(\infty)}(\zeta)$
 is analytic in
  $\mathbb{C}\setminus\left\{\Sigma_4\cup\Sigma_6\cup\Sigma_7\right\}$; see Figure \ref{Model RH contour} for the contours.
  \item(b)  $A^{(\infty)}(\zeta)$ satisfies the jump conditions
  \begin{equation}\label{eq:A-out-jump}
 A^{(\infty)}_+(\zeta)=A^{(\infty)}_-(\zeta)
 \left\{
 \begin{array}{ll}
         e^{2\pi i\beta\sigma_3} , &   \zeta \in \Sigma_4, \\[.2cm]
                                                               \begin{pmatrix}
                                0 &-e^{\pi i(\alpha-\beta)} \\
                                e^{-\pi i(\alpha-\beta)} &0
                                 \end{pmatrix}, &   \zeta \in \Sigma_6, \\[.5cm]
                                 \begin{pmatrix}
                                0 &-e^{-\pi i(\alpha-\beta)} \\
                                e^{\pi i(\alpha-\beta)} &0
                                 \end{pmatrix}, &   \zeta \in\Sigma_7.
 \end{array}  \right .  \end{equation}

\item(c) As $\zeta\to \infty$,  we have
  \begin{equation}\label{eq:A-out-infty}
 A^{(\infty)}(\zeta)=\left( I +
O\left(\frac{1}{\zeta}\right) \right) \zeta^{-\beta\sigma_3},
   \end{equation}
 where
 $\arg \zeta \in (-\pi/2, 3\pi/2)$.
  \end{description}

A solution can be constructed explicitly, that is,
\begin{equation}\label{eq:A-out}
 A^{(\infty)}(\zeta)=2^{\beta\sigma_3}e^{-\frac{\beta}{2}\pi i\sigma_3}\left( \frac{I-i\sigma_1 }{\sqrt{2}}\right)\left( \frac{\zeta-1 }{\zeta+1}\right)^{\frac{1}{4}\sigma_3}\left( \frac{I+i\sigma_1 }{\sqrt{2}}\right) \varphi(\zeta)^{-\beta\sigma_3}e^{\frac{\beta}{2}\pi i\sigma_3}D(\zeta)^{\sigma_3},
 \end{equation}
   where $\varphi(\zeta)=\zeta+\sqrt{\zeta^2-1}$ is defined as before  in \eqref{def:varphi}.
 The branches for $\sqrt{\zeta^2-1}$ and $\left( \frac{\zeta-1 }{\zeta+1}\right)^{ {1}/{4}}$ are taken such that  $\arg (\zeta\pm 1)\in(-\pi, \pi) $ and the branch for $\zeta^{\beta}$ is  taken such that $\arg \zeta\in(-\pi/2, 3\pi/2) $.
The function   $D(\zeta)$ is the Szeg\H{o} function which solves  the following scalar RH problem:
\begin{description}
  \item (a) $D(\zeta)$ is analytic in $\mathbb{C}\setminus [-1,1]$.
\item(b) $D(\zeta)$ satisfies the relation
\begin{equation}\label{eq:D-jump}
D_+(x)D_{-}(x)= \left\{  \begin{array}{ll}
 e^{\pi i\alpha},& -1<x<0,\\
 e^{-\pi i \alpha},& 0<x<1.
\end{array}  \right.
\end{equation}
\end{description}
Applying the Sokhotski-Plemelj formula, we have
\begin{equation}\label{eq:D-rep}
D(\zeta)=\exp\left\{\frac{\alpha\sqrt{\zeta^2-1}}{2i  }\left(-\int_{-1}^0 \frac{1}{\sqrt{1-x^2}}\frac{dx}{\zeta-x}
+\int_{0}^{1}  \frac{1}{\sqrt{1-x^2}}\frac{dx}{\zeta-x}\right)\right\}=\left(\frac{-i+\sqrt{\zeta^2-1}}{\zeta}\right)^{\alpha},\end{equation}
where the   branches for $\sqrt{\zeta^2-1}$ and $\zeta^{\alpha}$ are taken such that $\arg(\zeta\pm1), \arg \zeta \in  (-\pi,\pi) $; see \cite{ik}.
Straightforward calculation from \eqref{eq:D-rep} gives
\begin{equation}\label{eq:D-expan}
D(\zeta)=1-i\alpha/{\zeta}+O(1/\zeta^2),  \quad \zeta\to \infty,\end{equation}
\begin{equation}\label{eq:D-expan-zero}
D(\zeta)=e^{-i\alpha \pi/2} (\zeta/2)^{\alpha}\left(1+O(\zeta)\right ), \quad \zeta\to 0, \quad \Im \zeta>0,\end{equation}
and
\begin{equation}\label{eq:D-expan-1}
D(\zeta)=e^{\mp i\alpha \pi/2} \left(1+O(\sqrt{\zeta^2-1})\right), \quad \zeta\to \pm 1, \quad \Im \zeta>0.\end{equation}
Substituting \eqref{eq:D-expan-zero} into \eqref{eq:A-out}, we have
\begin{equation}\label{eq:A-out-near zero}
 A^{(\infty)}(\zeta)=2^{\beta\sigma_3}e^{-\frac{\beta}{2}\pi i\sigma_3}\left(\frac{  I-i\sigma_2}{\sqrt{2}}\right ) 2^{-\alpha\sigma_3} e^{-\frac{1}{2}\pi i\alpha\sigma_3} \left(I+O(\zeta)\right)\zeta^{\alpha\sigma_3}, \quad\zeta\to 0, \quad \Im \zeta>0,
 \end{equation}
where the Pauli matrix $\sigma_2$ is defined in \eqref{Pauli-matrix}. The behavior of $A^{(\infty)}(\zeta)$ as $\zeta\to 0$ from the lower-half plane can be obtained by combining \eqref{eq:A-out-jump} with \eqref{eq:A-out-near zero}.
Similarly, we get by substituting \eqref{eq:D-expan} into \eqref{eq:A-out}
\begin{equation}\label{eq:A-out-near infty }
 A^{(\infty)}(\zeta)=2^{\beta\sigma_3}e^{-\frac{\beta}{2}\pi i\sigma_3} \left(I+\frac{-i\alpha\sigma_3+\frac{1}{2}\sigma_2}{\zeta}+O\left(\frac 1{\zeta^2}\right)\right) 2^{-\beta\sigma_3} e^{\frac{1}{2}\pi i\beta\sigma_3}\zeta^{-\beta\sigma_3}, \quad\zeta\to \infty,
 \end{equation}
for $\arg \zeta\in(-\pi/2,3\pi/2)$.

\subsection{Local parametrices}

We need the local parametrices $A^{(+1)}(\zeta)$ and  $A^{(-1)}(\zeta)$ near $\zeta= 1$ and $\zeta=- 1$, respectively.
\begin{description}
  \item (a) $A^{(\pm1)}(\zeta)$ is analytic respectively in $U(\pm1)\setminus   \{\cup^7_{j=1}\Sigma_j\}$, where $\Sigma_j$ are illustrated in Figure \ref{Model RH contour}, and $U(\pm1)$ are small disks centered at $\pm1$, respectively.
\item(b) $A^{(\pm1)}(\zeta)$ satisfy the same jump conditions as $A(\zeta)$ in $U(\pm1)\cap   \{\cup^7_{j=1}\Sigma_j\}$.
\item(c) $A^{(\pm1)}(\zeta)$ satisfy the matching condition
\begin{equation}\label{eq:A-matching}
A^{(\pm1)}(\zeta)=(I+O(1/s))A^{(\infty)}(\zeta)  \end{equation}
for $\zeta$ on the boundaries of $U(\pm1) $, as $is\to +\infty$.
\end{description}

We focus on the construction of the local parametrix near $\zeta=-1$. The  local parametrix near $\zeta=1$ can be
constructed similarly. In   the neighborhood $U(-1)$, we seek a local parametrix of the   form
\begin{equation}\label{eq:A-local}
A^{(-1)}(\zeta)=E_{-1}(\zeta) \Phi_B\left(-|s|^2g^2(\zeta)\right )\left\{  \begin{array}{ll}
e^{\frac{1}{2}\pi i(\alpha-\beta)\sigma_3}e^{-sg(\zeta)\sigma_3},&  \Im  \zeta>0,\\
i\sigma_2 e^{\frac{1}{2}\pi i(\alpha-\beta)\sigma_3} e^{-sg(\zeta)\sigma_3},& \Im \zeta<0,\end{array} \right.
\end{equation}
where $E_{-1}(\zeta)$ is defined and analytic in   $U(-1)$. From the definition \eqref{def: g} it is readily seen that
$-g^2(\zeta)$ serves as a conformal mapping in $U(-1)$ and $U(+1)$. For later use, we   assign
\begin{equation}\label{eq:conformal mapping -g}
\sqrt{-g^2(\zeta)}= \left\{  \begin{array}{ll}
 -ig(\zeta),& \Im \zeta>0,\\
 ig(\zeta),& \Im \zeta<0,
\end{array} \right.
\end{equation}
with the $\zeta$-plane being cut   along $(-\infty, -1]\cup [1, +\infty)$, and the square root   being real positive for $-1<\zeta<1$; cf. \eqref{def: g}.
 The function  $\Phi_B(\zeta)$ in  \eqref{eq:A-local} solves the following model RH problem.

\begin{figure}[ht]
 \begin{center}
   \includegraphics[width=6 cm]{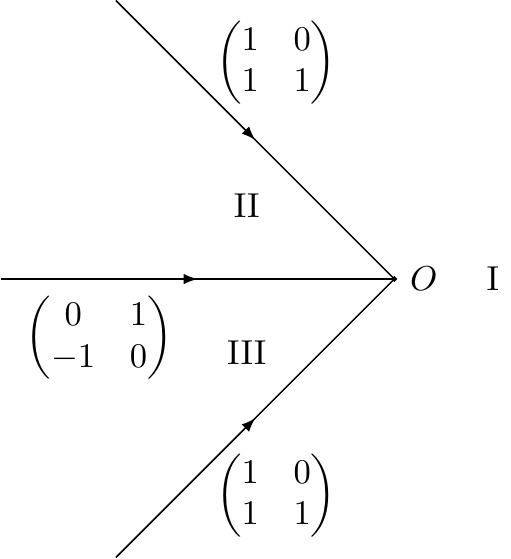} \end{center}
 \caption{\small{Contours $\Sigma_B$, sectors, and jumps  for the Bessel model RH problem  $\Phi_B$}}
 \label{figure: Bessel-model-B}
 \end{figure}

\begin{description}
  \item (a)   $\Phi_B(\zeta)$ is analytic in
  $\mathbb{C}\backslash \Sigma_B$, where the contour  $\Sigma_B$ is illustrated in Figure  \ref{figure: Bessel-model-B}.

  \item (b)   $\Phi_B(\zeta)$  satisfies the jump condition
  \begin{equation}\label{eq:phi-B-jump}
 \Phi_{B+}(\zeta)=\Phi_{B-}(\zeta) J_B(\zeta), \quad  \zeta\in \Sigma_B,
  \end{equation}
where the jump  $J_B(\zeta)$ is also indicated in Figure  \ref{figure: Bessel-model-B}.
  \item (c)  The asymptotic behavior of $\Phi_B(\zeta)$  at infinity is
  \begin{equation}\label{eq: phi-B-infinity}
\Phi_B(\zeta)=\zeta^{-\frac{1}{4}\sigma_3}\frac{I+i\sigma_1}{\sqrt{2}}
   \left (I+\frac {1
   }{8\sqrt{\zeta}} \begin{pmatrix}
                 -1 & -2i \\
                       -2i & 1
           \end{pmatrix}+O\left (\frac 1{\zeta}\right )\right)e^{\sqrt{\zeta}\sigma_3} \quad \textrm{as}~\zeta \rightarrow \infty .
 \end{equation}
\end{description}

The solution to the above RH problem can be constructed in terms of Bessel functions
\begin{equation}\label{eq:phi-B-solution}
\Phi_B(\zeta)=\pi^{\frac 12 \sigma_3} \left\{
 \begin{array}{ll}
   \left(
                               \begin{array}{cc}
                                 I_{0}(\sqrt{\zeta}) &\frac{i}{\pi} K_{0}(\sqrt{\zeta})  \\
                                 \pi i\sqrt{\zeta}    I'_{0}(\sqrt{\zeta})& -\sqrt{\zeta}    K'_{0}(\sqrt{\zeta}) \\
                                 \end{array}
                             \right),  &   \mbox{for} ~  \zeta \in \mathrm{I}, \\[0.5cm]
  \left(
                               \begin{array}{cc}
                                 I_{0}(\sqrt{\zeta}) &\frac{i}{\pi} K_{0}(\sqrt{\zeta})  \\
                                 \pi i\sqrt{\zeta}    I'_{0}(\sqrt{\zeta})& -\sqrt{\zeta}    K'_{0}(\sqrt{\zeta}) \\
                                 \end{array}
                             \right)  \begin{pmatrix}
                1 & 0 \\
                                 -1 & 1
           \end{pmatrix},    & \mbox{for}~  \zeta\in \mathrm{II},  \\[0.5cm]
   \left(
                               \begin{array}{cc}
                                 I_{0}(\sqrt{\zeta}) &\frac{i}{\pi} K_{0}(\sqrt{\zeta})  \\
                                 \pi i\sqrt{\zeta}    I'_{0}(\sqrt{\zeta})& -\sqrt{\zeta}    K'_{0}(\sqrt{\zeta}) \\
                                 \end{array}
                             \right)
                             \left(
                               \begin{array}{cc}
                                 1 & 0 \\
                                1 & 1 \\
                                 \end{array}
                             \right), & \mbox{for}~\zeta\in \mathrm{III},
 \end{array}
 \right .
\end{equation}
where $\arg \zeta\in (-\pi, \pi)$; see \cite{kmvv}.  For later use, we mention that $\Phi_B$ satisfies the differential equation
\begin{equation}\label{eq:Bessel-ODE}
\frac{d\Phi_B(\zeta)}{d\zeta}=\left(
                               \begin{array}{cc}
                                 0 & \frac{1}{2i\zeta} \\
                                 \frac{i}{2} & 0 \\
                                 \end{array}
                             \right)\Phi_B(\zeta).
\end{equation}

To satisfy the matching condition, for $\zeta\in U(-1)$, we define  $E_{-1}(\zeta)$ as
\begin{equation}\label{def:E}
E_{-1}(\zeta)=A^{(\infty)}(\zeta)
\left\{  \begin{array}{ll}
e^{-\frac{1}{2}\pi i(\alpha-\beta)\sigma_3}\frac{I-i\sigma_1}{\sqrt{2}} \left (-|s|^2g^2(\zeta)\right )^{\frac{1}{4}\sigma_3},&  \Im  \zeta>0,\\
 e^{-\frac{1}{2}\pi i(\alpha-\beta)\sigma_3}(-i\sigma_2)\frac{I-i\sigma_1}{\sqrt{2}} \left (-|s|^2g^2(\zeta)\right )^{\frac{1}{4}\sigma_3},& \Im  \zeta<0,\end{array}\right.
\end{equation}
where $g(\zeta)$ and $A^{(\infty)}(\zeta)$ are defined in \eqref{def: g} and \eqref{eq:A-out}, respectively.
From \eqref{eq:A-out-jump} and \eqref{eq:conformal mapping -g}, we see that $E_{-1}(\zeta)$ is analytic for $\zeta\in U(-1)$. Particularly, in view of \eqref{eq:A-out} and \eqref{eq:D-expan-1}, we have
\begin{equation}\label{eq:E-at -1}
E_{-1}(-1)=2^{\beta\sigma_3}e^{-\frac{1}{2}\beta\pi i\sigma_3}\frac{I-i\sigma_1}{\sqrt{2}}|s|^{\frac{1}{2}\sigma_3}2^{-\frac{1}{2}\sigma_3}e^{\frac{1}{4}\pi i\sigma_3}.\end{equation}

The other parametrix $A^{(+1)}(\zeta)$ at $\zeta=1$,   in a $\zeta$-neighborhood  $U(+1)$, can be constructed similarly, also in terms of the Bessel-type model RH problem $\Phi_B$. More precisely, we define
\begin{equation}\label{eq:A-local-1}
A^{(+1)}(\zeta)=E_1(\zeta)\sigma_3 \Phi_B(-|s|^2g^2(\zeta))\sigma_3\left\{  \begin{array}{ll}
e^{-\frac{1}{2}\pi i(\alpha-\beta)\sigma_3}e^{-sg(\zeta)\sigma_3},&  \Im \zeta>0\\
i\sigma_2 e^{-\frac{1}{2}\pi i(\alpha-\beta)\sigma_3} e^{-sg(\zeta)\sigma_3}, &\Im  \zeta<0,\end{array}  \right.
\end{equation}
where
\begin{equation}\label{def:E-1}
E_1(\zeta)=A^{(\infty)}(\zeta)
\left\{  \begin{array}{ll}
e^{\frac{1}{2}\pi i(\alpha-\beta)\sigma_3}\frac{I+i\sigma_1}{\sqrt{2}} (-|s|^2g^2(\zeta))^{\frac{1}{4}\sigma_3},&  \Im  \zeta>0,\\[.2cm]
 e^{\frac{1}{2}\pi i(\alpha-\beta)\sigma_3}(-i\sigma_2)\frac{I+i\sigma_1}{\sqrt{2}} (-|s|^2g^2(\zeta))^{\frac{1}{4}\sigma_3},& \Im  \zeta<0,\end{array}\right.
\end{equation}
and $A^{(\infty)}(\zeta)$ is given in \eqref{eq:A-out}.
Using \eqref{eq:A-out-jump}, \eqref{eq:A-out}, \eqref{eq:D-expan-1} and \eqref{eq:conformal mapping -g}, we see that $E_1(\zeta)$ is analytic for $\zeta\in U(+1)$. Moreover, we have
\begin{equation}\label{eq:E1-at 1}
E_1(1)=2^{\beta\sigma_3}e^{-\frac{1}{2}\beta\pi i\sigma_3}\frac{I+i\sigma_1}{\sqrt{2}}|s|^{\frac{1}{2}\sigma_3}2^{-\frac{1}{2}\sigma_3}e^{-\frac{1}{4}\pi i\sigma_3}.\end{equation}

\subsection{Final transformation}
Next, we define
\begin{equation}\label{eq:B-}
\hat B(\zeta)=\left\{ \begin{array}{ll}
                      A(\zeta)A^{(\infty)}(\zeta)^{-1}, & \hbox{$\zeta\setminus\{U(+1)\cup U(-1)\}$,} \\
                     A(\zeta)A^{(\pm 1)}(\zeta)^{-1}, & \hbox{$\zeta\in U(\pm 1)$.}
                    \end{array}
\right.
 \end{equation}
 For $\zeta\in \partial U(\pm1)$, clockwise-oriented, and $\Im \zeta>0$, we have jumps
\begin{equation}\label{eq:B-jump}
 J_{\hat B}(\zeta)=A^{(\pm1)}(\zeta)A^{(\infty)}(\zeta)^{-1} .\end{equation}
As $is\to+\infty$, by a standard analysis of small-norm RH problems \cite{dkmv1}, we have
\begin{equation}\label{eq:B-est}
{\hat B}(\zeta)=I+\frac{{\hat B}_1(\zeta)}{s}+O\left(\frac 1{s^2}\right ),\quad is\to +\infty, \end{equation}
where
\begin{equation}\label{eq:B-1}
{\hat B}_1(\zeta)=\frac{1}{2\pi i}\oint_{\partial U(-1)}\frac{J_{{\hat B},1}(x)dx}{x-\zeta}+\frac{1}{2\pi i}\oint_{\partial U(+1)}\frac{J_{{\hat B},1}(x)dx}{x-\zeta}
\end{equation}
for $\zeta\in\mathbb{C}\setminus\{\partial U(-1)\cup \partial U(+1)\}$, and the integration paths in \eqref{eq:B-1} are small clockwise circles encircling $-1$ and $1$, respectively.
Here, $J_{{\hat B},1}$ is the coefficient taken from
\begin{equation}\label{eq:B-jump-left-expand}
J_{\hat B}(\zeta)=I+\frac{J_{{\hat B},1}(\zeta)}{s}+O\left(\frac 1{s^2}\right ), ~~\zeta\in \partial U(-1)\cup \partial U(+1),~~\mbox{as}~is\to+\infty,   \end{equation}
where, for $\zeta\in\partial U(-1)$,
\begin{equation}\label{eq:J-1-left}
J_{{\hat B},1}(\zeta)=\frac{1}{2\sqrt{\zeta^2-1}}A^{(\infty)}(\zeta)e^{-\frac{1}{2}\pi i(\alpha-\beta)\sigma_3}
\begin{pmatrix}
-1 & -2i \\
                       -2i & 1
\end{pmatrix}
e^{\frac{1}{2}\pi i(\alpha-\beta)\sigma_3}A^{(\infty)}(\zeta)^{-1},\end{equation}
as follows from  \eqref{eq:A-local}, \eqref{def:E} and \eqref{eq:B-jump}, and a combination of \eqref{eq:A-local-1}, \eqref{def:E-1} and \eqref{eq:B-jump} gives
\begin{equation}\label{eq:J-1-right}
J_{{\hat B},1}(\zeta)=\frac{1}{2\sqrt{\zeta^2-1}}A^{(\infty)}(\zeta)e^{-\frac{1}{2}\pi i(\alpha-\beta)\sigma_3}\left(
                     \begin{array}{cc}
                       -1 & 2i \\
                       2i & 1 \\
                     \end{array}
                   \right)e^{\frac{1}{2}\pi i(\alpha-\beta)\sigma_3}A^{(\infty)}(\zeta)^{-1} \end{equation}
for $\zeta\in\partial U(+1)$,  with the principal branches taken such that $\arg(\zeta\pm1)\in(-\pi,\pi)$.
Putting in use \eqref{eq:A-out} and \eqref{eq:D-expan-1}, we obtain that $J_{{\hat B},1}(\zeta)$ has simple poles  $\zeta=\pm 1$ with the following leading term expansions
 \begin{equation}\label{eq:J-1-near -1}
J_{{\hat B},1}(\zeta)\sim \frac{1}{4(\zeta+1)}2^{\beta\sigma_3}e^{-\frac{\beta}{2}\pi i\sigma_3}\left(\sigma_3-i\sigma_1\right)2^{-\beta\sigma_3}e^{\frac{\beta}{2}\pi i\sigma_3} \quad \mbox{as}\quad \zeta\to-1,\end{equation}
and
\begin{equation}\label{eq:J-1-near 1}
J_{{\hat B},1}(\zeta)\sim \frac{1}{4(\zeta-1)}2^{\beta\sigma_3}e^{-\frac{\beta}{2}\pi i\sigma_3}\left(\sigma_3+i\sigma_1\right)2^{-\beta\sigma_3}e^{\frac{\beta}{2}\pi i\sigma_3} \quad\mbox{as} \quad \zeta\to 1.\end{equation}
Now the residue theorem applies. From \eqref{eq:J-1-near -1}  and  \eqref{eq:J-1-near 1}  we have
\begin{equation}\label{eq:J-1-res-left}
\frac{1}{2\pi i}\oint_{\partial U(-1)}\frac{J_{{\hat B},1}(x)dx}{x-\zeta}=\frac{1}{4}\left(\sigma_3-i2^{\beta\sigma_3}e^{-\frac{\beta}{2}\pi i\sigma_3}\sigma_12^{-\beta\sigma_3}e^{\frac{\beta}{2}\pi i\sigma_3}\right)\frac{1}{\zeta+1}
\end{equation}
for $\zeta\not\in U(-1)$, and
\begin{equation}\label{eq:J-1-res-right}
\frac{1}{2\pi i}\oint_{\partial U(+1)}\frac{J_{{\hat B},1}(x)dx}{x-\zeta}=\frac{1}{4}\left(\sigma_3+i2^{\beta\sigma_3}e^{-\frac{\beta}{2}\pi i\sigma_3}\sigma_12^{-\beta\sigma_3}e^{\frac{\beta}{2}\pi i\sigma_3}\right)\frac{1}{\zeta-1}
\end{equation}
for $\zeta\not\in U(+1)$.
Then, substituting \eqref{eq:J-1-res-left} and \eqref{eq:J-1-res-right} into \eqref{eq:B-1} gives the approximation as $\zeta\to \infty$
\begin{equation}\label{eq:B-1-infty}
\hat B_1(\zeta)= \frac{\sigma_3}{2\zeta}+O\left(\frac 1{\zeta^2}\right).
\end{equation}

\subsection{Proof of Theorem \ref{thm-asy}: Asymptotics of $sH$ as $is\to +\infty$}

In the outer region $|\zeta-1|>\delta$ and $|\zeta+1|>\delta$ for some small positive constant $\delta$, we get from the transformations \eqref{def: A} and  \eqref{eq:B-}
 \begin{equation}\label{eq: Psi-app-out}
 \Psi(\zeta)={\hat B}(\zeta)A^{(\infty)}(\zeta)\exp\left(\frac{s}{4}\sqrt{\zeta^2-1}\;\sigma_3\right).
 \end{equation}
 Expanding the last term as $\zeta\to\infty$, we obtain
 $$\exp\left(\frac{s}{4}\sqrt{\zeta^2-1}\;\sigma_3\right)=\exp\left(\frac{1}{4}s\zeta\sigma_3\right) \left(I-\frac{s}{8\zeta}\sigma_3+O\left(\frac 1{\zeta^2}\right)\right).$$
Substituting this and the expansions \eqref{eq:A-out-near infty } and \eqref{eq:B-1-infty} into \eqref{eq: Psi-app-out}, and recalling \eqref{eq:H-psi}, we
have
 \begin{equation*}
H(s)=-\frac {(\Psi_1)_{11}}2-\frac{(\alpha^2-\beta^2)}{s}=\frac{s}{16}+\frac{i\alpha}{2}-
\left(\alpha^2-\beta^2+\frac{1}{4}\right)\frac{1}{s}+O\left(\frac 1{s^2}\right)
\end{equation*} as $is\to+\infty$. This completes the proof of  \eqref{thm: H-asy-large s}.

From  \eqref{eq:y-def}, \eqref{eq:d}, \eqref{eq:A-out-near zero}, \eqref{eq:B-est}  and \eqref{eq: Psi-app-out}, we obtain
\begin{equation}\label{eq: y-asy-large s}
y(s)=2^{2\beta}e^{-\beta\pi i}\left(1+ O(  1/s)\right), \quad is\to+\infty,
\end{equation}
and
\begin{equation}\label{eq: d-asy-large s}
d(s)=2^{-2\alpha} e^{-\alpha\pi i}e^{is/2}\alpha (1+O(1/s)), \quad is\to+\infty.
\end{equation}
Now we derive the large-$s$ behavior of $u_k$ and $v_k$ introduced in the Hamiltonian \eqref{int:H}.
In the region $|\zeta+1|<\delta$, we get by tracing back the invertible transformations \eqref{def: A} and \eqref{eq:B-}
  \begin{equation}\label{eq: Psi-app-in}
 \Psi(\zeta)={\hat B}(\zeta)E_{-1}(\zeta) \Phi_B\left (-|s|^2g^2(\zeta)\right )
e^{\frac{1}{2}\pi i(\alpha-\beta)\sigma_3},\quad  \Im \zeta>0,
 \end{equation}
where $E_{-1}(\zeta)$ is independent of $s$ and analytic for $|\zeta+1|<\delta$,
${\hat B}(\zeta)$ is also analytic for $|\zeta+1|<\delta$.
Thus, we have
\begin{align}\label{eq: A-2-est}
 A_{2}(s)&={\hat B}(-1)E_{-1}(-1) \left(\lim_{\zeta\to 0}\zeta\Phi_B'(\zeta)\Phi_B^{-1}(\zeta)\right)E_{-1}(-1)^{-1}{\hat B}(-1)^{-1} \nonumber \\
 &=-\frac{i}{2}s{\hat B}(-1)E_{-1}(-1) \sigma_{+}E_{-1}(-1)^{-1}{\hat B}(-1)^{-1} \nonumber \\
 &=\frac{i}{8}s{\hat B}(-1)2^{\beta\sigma_3}e^{-\frac{1}{2}\beta\pi i\sigma_3}(i\sigma_3+\sigma_1)
 e^{\frac{1}{2}\beta\pi i\sigma_3}2^{-\beta\sigma_3}{\hat B}(-1)^{-1},
 \end{align}
where use has been made of the differential equation \eqref{eq:Bessel-ODE} and formula \eqref{eq:E-at -1}.
It follows from \eqref{eq:B-est} that
$$\hat B(-1)=I+O(1/s),\quad is\to+\infty.$$
Thus, from \eqref{eq:A-k}, \eqref{eq: A-2-est} and \eqref{eq: y-asy-large s}-\eqref{eq: d-asy-large s},  we obtain the asymptotic approximations of $u_2(s)$ and $v_2(s)$  as $is\to+\infty$, as stated in
\eqref{thm: u-2-asy-large s} and \eqref{thm: v-2-asy-large s}.

Using the approximation of $\Psi(\zeta;s)$ for $|\zeta-1|<\delta$, a similar argument justifies \eqref{thm: u-1-asy-large s} and \eqref{thm: v-1-asy-large s}, namely,   the asymptotic approximations of $u_1(s)$ and $v_1(s)$   as $is\to+\infty$.

Thus, we have obtained the  large-$s$ asymptotics in Theorem \ref{thm-asy}.
This, together with the  small-$s$ asymptotic formulas obtained in  Section \ref{sec:small-s behavior} and
the existence results given in Proposition \ref{pro: existence of solution},  completes the proof of Theorem \ref{thm-asy}.

It is worth noting that with   long but direct calculation   the later terms in the asymptotic expansions of $u_k(s)$ and $v_k(s)$ can also be determined.
For example, we have from \eqref{eq:B-est} that
$$\hat B(\pm1)=I+\frac{{\hat B}_1(\pm1)}{s}+O\left(\frac 1{s^2}\right), \quad is\to+\infty.$$
Thus, the asymptotics \eqref{thm: v-1-asy-large s} and \eqref{thm: v-2-asy-large s} can be refined
respectively to
\begin{equation}\label{eq: v-1-asy-large s}
v_1(s)=i+\frac{\gamma_1}{s}+O(1/s^2), \quad is\to+\infty,
\end{equation}
and
\begin{equation}\label{eq: v-2-asy-large s}
v_2(s)=-i+\frac{\gamma_2}{s}+O(1/s^2), \quad is\to+\infty,
\end{equation}
where the coefficients $\gamma_1$ and $\gamma_2$ can be determined by using ${\hat B}_1(\pm1)$.

\section{RH problem for orthogonal polynomials and the differential identity}\label{sec:d-identity}
We consider the orthogonal polynomials associated with the weight function \eqref{def: weight-circle}.
Note that the weight function \eqref{def: weight-circle} can be analytically  extended to the complex $z$-plane with a cut  $[0,\infty)$
 \begin{equation}\label{def: weight}
w(z)=(z-1)^{2\alpha}z^{\beta-\alpha}e^{-\pi i(\alpha+\beta)}, \end{equation}
where the branch cut of $(z-1)^{2\alpha}$ is taken along $[1,\infty)$
 so that $\arg (z-1)\in(0,2\pi)$ and the branch of $z^{\beta-\alpha}$ is chosen so that $\arg z\in (0,2\pi)$.
 Let $\pi_n(z;t)=z^n+\cdots$
  be the monic orthogonal polynomial  of degree $n$ with respect to
  the weight \eqref{def: weight-circle}, satisfying the orthogonality relation
 \begin{equation}\label{eq: oth}
\frac{1}{2\pi} \int_t^{2\pi-t} \pi_n(e^{i\theta};t)\overline{\pi_m(e^{i\theta};t)}w(e^{i\theta};t)d\theta=\chi_n^{-2}(t)\delta_{n,m}, \end{equation}
where $\chi_n(t)>0$ and $n,m=0,1,\cdots$. Denote the reverse polynomials  associated with $\pi_n(z)$ by
\begin{equation*}
 \pi_n^*(z)=z^n\overline{\pi_n(1/\bar{z})}=z^n\overline{\pi_n}(1/z)=z^n+\cdots,
\end{equation*}where $\overline{\pi_n}$ denotes the polynomial whose coefficients are complex conjugates of that of $\pi_n$.
Let $Y(z)$ be the  $2\times 2$ matrix-valued function
\begin{equation}\label{eq:Y-solution}
Y(z)=
\begin{pmatrix}
\pi_n(z) & \frac{1}{2\pi i } \int_{C_t}\frac{\pi_n(x) w(x)dx}{x^n(x-z)}\\
-\chi_{n-1}^2 \pi^*_{n-1}(z)&  -\frac{\chi_{n-1}^2}{2\pi i}  \int_{C_t}\frac{\pi^*_{n-1}(x) w(x)dx}{x^n(x-z)}
\end{pmatrix}.
\end{equation}
Then,  $Y(z)$ is the unique solution to the RH problem below.
\subsubsection*{RH problem for $Y$}
\begin{description}
  \item (a)  $Y(z)$ is analytic in $\mathbb{C} \setminus C_t$.

  \item (b) $Y(z)$  satisfies the jump condition
  \begin{equation}\label{eq:Y-jump}
  Y_{+}(z)=Y_{-}(z)
  \begin{pmatrix}
  1 & z^{-n}w(z;t ) \\
  0 & 1
  \end{pmatrix},\qquad z\in C_t,
  \end{equation}
where  the arc $C_t$ is defined in \eqref{def: arc}. The weight  function $w(z;t)$, as given  in \eqref{def: weight-circle},  is the restriction of $w(z)$ \eqref{def: weight} on the arc $C_t$.
  \item (c)  As $z\to \infty$, we have
  \begin{equation}\label{eq:Y-infinity}
  Y(z)=\left (I+O\left (\frac 1 {z}\right )\right)
 \begin{pmatrix}
 z^n & 0 \\
 0 & z^{-n}
 \end{pmatrix}.
 \end{equation}
 \item (d)   Denoting  $z_1=e^{it}$ and $z_2=e^{i(2\pi-t)}$, we have, respectively  for   $z_*=z_1$ and  $z_*=z_2$, that
  \begin{equation}\label{eq:Y-endpoint}
  Y(z)=\begin{pmatrix}
 O(1)& O\left(\ln(z-z_*)\right) \\
  O(1) & O\left(\ln(z-z_*)\right)
  \end{pmatrix}~~\mbox{as}~~z\to z_*.
 \end{equation}
\end{description}

We have the  differential identities with respect to $t$ in the following lemma, involving $Y(z)$ and the corresponding Toeplitz determinant. The lemma was first proved in \cite{dkv}.
\begin{lem}\label{lem:differential identity}
Let $D_n(t)$ be the Toeplitz determinant associated with the weight function \eqref{def: weight-circle}, then we have
\begin{equation}\label{eq: diff identity}
    \frac d {dt}\ln D_{n}(t)=-\frac{1}{2\pi }\sum_{j=1}^2z^{-n+1}_j w(z_j;t)\lim_{z\to z_j}\left(Y^{-1}(z)\frac{d}{dz} Y(z)\right)_{21},
\end{equation}
where $z_1= e^{it}$ and $z_2= e^{i(2\pi-t)}$.  \end{lem}
\begin{proof}
We start with the  expression of the Toeplitz determinant  in terms of the leading coefficients of the orthonormal polynomials
\begin{equation}\label{eq: Top-product}
D_n(t)=\prod_{k=0}^{n-1}\chi_k^{-2}(t).
\end{equation}
Taking logarithmic derivative on both side of \eqref{eq: Top-product} yields
\begin{equation}\label{eq: log-de}
\frac{d}{dt}\ln D_n(t)=-2\sum_{k=0}^{n-1}\chi_k^{-1}(t)\frac{d}{dt}\chi_k(t).
\end{equation}
The orthogonal relation \eqref{eq: oth}
serves as an integral representation  of $\chi_n(t)$. Taking derivatives on both sides, we have
 \begin{equation}\label{eq: log-leading coef}
2\chi_k^{-1}(t)\frac{d}{dt}\chi_k(t)=\frac{\chi_k^2(t)}{2\pi}\sum_{j=1}^2 \pi_k(z_j)\overline{\pi_k}(z^{-1}_j)w(z_j;t),
\end{equation}
where $z_1= e^{it}$, $z_2= e^{i(2\pi-t)}$ and $k=0,1,...,n-1$.
Substituting \eqref{eq: log-leading coef} into \eqref{eq: log-de} gives
\begin{equation}\label{eq: log-de-2}
\frac{d}{dt}\ln D_n(t)=-\frac{1}{2\pi}\sum_{j=1}^2\left(\sum_{k=0}^{n-1}\chi_k^2(t)\pi_k(z_j)\overline{\pi_k}(z^{-1}_j)w(z_j;t)\right).
\end{equation}
From
the Christoffel-Darboux identity \cite{Szego}
\begin{equation}\label{eq: CD-1}
\sum_{k=0}^{n-1}\chi_k^2(t)\pi_k(z)\overline{\pi_k}(z^{-1})=z\chi_n^2(t)\left[\overline{\pi}_{n}(z^{-1})\frac{d}{dz}\pi_n(z)-\pi_n(z) \frac{d}{dz}\overline{\pi}_{n}(z^{-1})\right]
-n\chi_n^2(t)\pi_n(z)\overline{\pi}_{n}(z^{-1}),
\end{equation}
and the recurrence relation \cite{Szego}
\begin{equation}\label{eq: recurrence }
z\overline{\pi_n}(1/z)=\frac{\chi_{n-1}(t)^2}{\chi_n(t)^2}\overline{\pi_{n-1}}(1/z)+
z^{1-n}\overline{\pi_n}(0)\pi_n(z),\end{equation}
we have
\begin{align}\label{eq: CD-2}
&\sum_{k=0}^{n-1}\chi_k^2(t)\pi_k(z)\overline{\pi_k}(z^{-1})\nonumber\\
&=-\chi_{n-1}^2(t)\left[\pi_n(z) \frac{d}{dz}\overline{\pi_{n-1}}(z^{-1})
-\overline{\pi_{n-1}}(z^{-1})\frac{d}{dz}\pi_n(z)+(n-1)z^{-1}\pi_n(z)\overline{\pi_{n-1}}(z^{-1})\right]\nonumber\\
&=-\chi_{n-1}^2(t)z^{1-n}\left[\pi_n(z) \frac{d}{dz}\pi_{n-1}^*(z^{-1})
-\pi_{n-1}^*(z^{-1})\frac{d}{dz}\pi_n(z)\right].
\end{align}
Then, the differential identity  \eqref {eq: diff identity} follows from \eqref{eq:Y-solution}, \eqref {eq: log-de-2} and \eqref{eq: CD-2}.

\end{proof}
\section{Asymptotics of the RH problem for $Y$}\label{sec:Y}
We carry out the Deift-Zhou nonlinear steepest descent  analysis of the RH problem for $Y(z)$ as $n\to\infty $ and $t\to 0^+$ in a way such that $nt$ is bounded.
\subsection{Normalization of the RH problem: $Y\to T$}
We introduce the transformation $Y\to T$ to normalize the large-$z$ behavior of $Y(z)$,
\begin{equation}\label{eq:Y-T}
  T(z)=\left\{ \begin{array}{ll}
                      Y(z)z^{-n\sigma_3}, & \hbox{$|z|>1$,} \\
                     Y(z), & \hbox{$|z|<1$,}
                    \end{array}
\right.
 \end{equation}
Then, it is readily seen that $T(z)$ satisfies the following RH problem.
 \subsubsection*{RH problem for $T$}
\begin{description}
  \item (a)  $T(z)$ is analytic in $\mathbb{C} \setminus C_t$.

  \item (b) $T(z)$  satisfies the jump condition
  \begin{equation}\label{eq:T-jump-1}
  T_+(z)=T_-(z)
  \begin{pmatrix}
  z^n & w(z) \\
  0 & z^{-n}
  \end{pmatrix},\qquad z\in C_t,
  \end{equation}
  and
  \begin{equation}\label{eq:T-jump-2}
  T_+(z)=T_-(z)
  \begin{pmatrix}
  z^n & 0 \\
  0 & z^{-n}
  \end{pmatrix},\qquad z\in C\setminus C_t,
  \end{equation}
  where  the weight function $w(z)$ is defined in \eqref{def: weight}, $C=\{z: |z|=1\} $ is the unit circle and the arc $C_t$ is defined in \eqref{def: arc}. Both the arcs in \eqref{eq:T-jump-1} and \eqref{eq:T-jump-2} are oriented counterclockwise.
       \item (c)  As $z\to \infty$, we have
  \begin{equation}\label{eq:T-infinity}
 T(z)=I+O\left (\frac 1 {z}\right ). \end{equation}
 \item (d)   As $z\to z_1$ and $z\to z_2$, where $z_1=e^{it}$ and $z_2=e^{i(2\pi-t)}$, $T(z)$ shares the same behavior as $Y(z)$; cf. \eqref{eq:Y-endpoint}.
\end{description}

\subsection{Opening of the lens: $T\to S$}
The diagonal entries of the jump matrix for $T$  in \eqref{eq:T-jump-1} are highly oscillating for $n$ large.
 To turn the oscillation to exponential decays on certain contours, we deform the contours $C_t$ and introduce
 the second transformation $T\to S$.
The transformation is  based on the following factorization of the  jump matrix  \eqref{eq:T-jump-1}
 \begin{equation}\label{eq:jump-fact}
   \begin{pmatrix}
  z^n & w(z ) \\
  0 & z^{-n}
  \end{pmatrix}=  \begin{pmatrix}
  1 & 0 \\
 z^{-n}w(z)^{-1 }& 1
  \end{pmatrix} \begin{pmatrix}
  0 & w(z) \\
  -w(z )^{-1} & 0  \end{pmatrix}
  \begin{pmatrix}
  1 & 0 \\
  z^{n}w(z)^{-1}&1
  \end{pmatrix}.
  \end{equation}
  Define the transformation
  \begin{equation}\label{eq:T-S}
  S(z)=\left\{ \begin{array}{lll}
                   T(z) \begin{pmatrix}
  1 & 0 \\
 z^{-n}w(z )^{-1 }& 1
  \end{pmatrix} , & \hbox{for $z\in  \Omega_E$,}\\[.5cm]
  T(z)\begin{pmatrix}
  1 & 0 \\
  -z^{n}w(z)^{-1}&1
  \end{pmatrix}  ,   &   \hbox{for $z\in  \Omega_I$,}\\
     T(z), & \hbox{otherwise, } \\
  \end{array}
\right.
 \end{equation}
 with the regions shown in Figure \ref{fig:open-lens}.
 Then,  $S(z)$ satisfies the following RH problem.
  \begin{figure}[ht]
 \begin{center}
   \includegraphics[height=8.5cm]{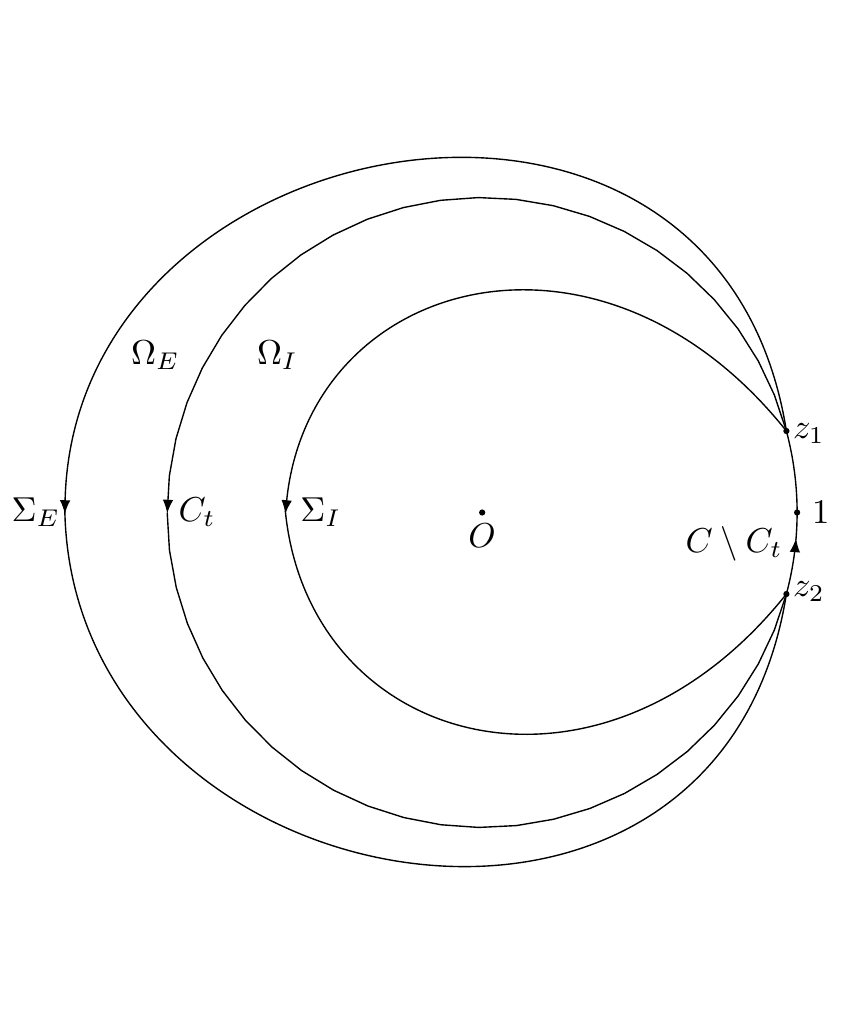}  \end{center}
  \caption{\small{ Opening of the lens, contours and regions for the  RH problem for $S(z)$.}}
 \label{fig:open-lens}
\end{figure}


 \subsubsection*{RH problem for $S$}
\begin{description}
  \item (a)  $S(z)$ is analytic in $\mathbb{C} \setminus \Sigma_E\cup C\cup \Sigma_I$, where the jump curves are illustrated in Figure  \ref{fig:open-lens}, and $C= \{z:|z|=1\}$ is the unit circle.

  \item (b) $S(z)$  satisfies the jump condition
  \begin{equation}\label{eq:S-jump}
  S_+(z)=S_-(z)\left\{ \begin{array}{lll}
                       \begin{pmatrix}
  1 & 0 \\
  z^{-n}w(z)^{-1}&1
  \end{pmatrix}  , & \hbox{for $z\in \Sigma_E$, } \\[.5cm]
                      \begin{pmatrix}
  0 & w(z) \\
-w(z)^{-1 }&0
  \end{pmatrix} , & \hbox{for $z\in C_t$,}\\[.5cm]
  \begin{pmatrix}
  1 & 0 \\
  z^{n}w(z)^{-1}&1
  \end{pmatrix}  ,   &   \hbox{for $z\in \Sigma_I$,}    \\[.5cm]
  \begin{pmatrix}
  z^{n} & 0 \\
  0&z^{-n}
  \end{pmatrix} ,    &   \hbox{for $z\in C\setminus C_t$.}
  \end{array}
\right.
    \end{equation}
  \item (c)  As $z\to \infty$, we have
  \begin{equation}\label{eq:S-infinity}
 S(z)=I+O\left (\frac 1 {z}\right ). \end{equation}
 \item (d)   At $z_1=e^{it}$ and $z_2=e^{i(2\pi-t)}$,   we have
  \begin{equation}\label{eq:S-endpoint}
  S(z)=O\left(\ln(z-z_1)\right)~~\mbox{as}~~ z\to z_1~~~\mbox{and}~~~ S(z)=O\left(\ln(z-z_2)\right )~~\mbox{as}~~ z\to z_2 .
 \end{equation}
\end{description}

\subsection{Global parametrix}
As $n\to \infty$, the jump matrices in \eqref{eq:S-jump} tend  to the identity matrix for $z$ bounded away from the unit circle  $C$.
The  arc $ C_t$ also tends to the unit circle in the double scaling case as $n\to \infty$.   Thus, we consider the approximate RH problem below, with jump occurs on the unit circle alone.

\subsubsection*{RH problem for $N$}
\begin{description}
  \item (a)  $N(z)$ is analytic in $\mathbb{C} \setminus   C$, where $C$ is the unit circle.

  \item (b) $N(z)$  satisfies the jump condition on the unit circle, counterclockwise oriented,
  \begin{equation}\label{eq:N-jump}
 N_+(z)=N_-(z)\                      \begin{pmatrix}
  0 & w(z) \\
-w(z)^{-1 }&0
  \end{pmatrix},   \end{equation}where the weight $w(z)$ is given in \eqref{def: weight}.
  \item (c)  As $z\to \infty$, we have
  \begin{equation}\label{eq:N-infinity}
 N(z)=I+O\left (\frac 1 {z}\right ). \end{equation}
 \end{description}

 The solution to the above RH problem is given explicitly by
 \begin{equation}\label{eq:N-solution}
  N(z)=\left\{ \begin{array}{ll}
                      \left(\frac{z-1}{z}\right)^{(\beta-\alpha)\sigma_3}, & \hbox{$|z|>1$,} \\
                   (z-1)^{(\alpha+\beta)\sigma_3}e^{-\pi i(\alpha+\beta)\sigma_3} \begin{pmatrix}
  0 & 1 \\
-1&0
  \end{pmatrix}, & \hbox{$|z|<1$,}
                    \end{array}
\right.
 \end{equation}
 where the branch of $z^{\alpha-\beta}$ is chosen so that $\arg z\in (0,2\pi)$ and both branch cuts of $(z-1)^{\alpha+\beta}$
 and $(z-1)^{\beta-\alpha}$  are taken along $[1,\infty)$, with $\arg(z-1)\in(0,2\pi)$.

 \subsection{Local parametrix}
The global parametrix $N(z)$ fails to   approximate   $S(z)$ near the endpoints $z_1$ and $z_2$, where the weight has jump discontinuities.  For $n$ large,
 both the endpoints   are close to $1$ and  belong to a neighborhood $U(1,r)=\left\{z: |z-1|<r\right\}$ of $z=1$, with $r$ small but fixed.  We intend to
 construct a local parametrix $P^{(1)}(z)$ in   $U(1,r)$, which  satisfies the same jump condition as $S(z)$ and
  matches with $N(z)$ on the boundary of  $U(1,r)$.

   \subsubsection*{RH problem for $P^{(1)}$}
\begin{description}
  \item (a)  $P^{(1)}(z)$ is analytic in $U(1,r) \setminus \left\{\Sigma_E\cup C\cup \Sigma_I\right\}$; see Figure \ref{fig:open-lens} for the contours.

  \item (b) $P^{(1)}(z)$  satisfies the same jump condition as $S(z)$ on  $U(1,r) \cap \left\{\Sigma_E\cup C\cup \Sigma_I\right\}$; cf. \eqref{eq:S-jump}.
   \item (c)  As $n\to \infty$, we have the matching condition
  \begin{equation}\label{eq:matching condition}
P^{(1)}(z)=\left(I+O\left (\frac 1 {n}\right )\right)N(z). \end{equation}
\end{description}

Let
\begin{equation}\label{eq:M}
  M(z)=\left\{ \begin{array}{ll}
                      z^{\frac{n}{2}\sigma_3}w(z)^{-\frac{1}{2}\sigma_3}, & \hbox{$|z|<1$, $z\in U(1,r)$,} \\
                     z^{\frac{n}{2}\sigma_3}w(z)^{\frac{1}{2}\sigma_3} \begin{pmatrix}
  0 & -1\\
1&0
  \end{pmatrix}, & \hbox{$|z|>1$, $z\in U(1,r)$,}
                    \end{array}
\right.
 \end{equation}where the principal branches are chosen for the power functions.
  We also define
  \begin{equation}\label{eq:mapping}
f(z)=f(z;t)=e^{-\frac{ \pi i}{2} } t^{-1}\ln z,\end{equation}
where the branch   of  $\ln z$  is chosen such that $\arg z\in (-\pi, \pi)$.
Note that $f(z)$ is a conformal mapping from  $U(1,r)$ to a neighborhood of the origin so long as $0<r<1$, with $f(z_1)= 1$ and  $f(z_2)= -1$.
  We seek a solution to the above RH problem of the following form:
\begin{equation}\label{eq:parametrix}
P^{(1)}(z)=E(z)\Psi(f(z); -2int )e^{\pm \frac{\pi i}{2}(\alpha-\beta)\sigma_3}M(z),  \quad \pm \Im z>0, \end{equation}
 where $\Psi(\zeta; s)$ is the solution to the model RH problem discussed in Section \ref{sec:model-RHP},  $E(z)$ is a certain matrix-valued analytic function defined in $U(1,r)$, to be determined.
It then follows from the jumps \eqref{eq:Psi-jump} for $\Psi$ and the definition of $M(z)$ in \eqref{eq:M}  that $P^{(1)}(z)$ fulfills the same jump condition  as $S(z)$ for $z \in U(1,r)\cap \left\{\Sigma_E\cup C\cup \Sigma_I\right\}$.
We then choose $E(z)$ so that  the matching condition  \eqref{eq:matching condition} holds.

\begin{prop}\label{pro:E}
Let
\begin{equation}\label{eq:E}
E(z)=e^{\pm \frac{\pi i(\alpha-\beta)} 2  \sigma_3} e^{-\frac{\pi i(\alpha+\beta)} 2  \sigma_3} z^{\frac{ \alpha-\beta}{2} \sigma_3}
\left(\frac{f(z)}{z-1}\right)^{-\beta\sigma_3}\begin{pmatrix}
  0 & 1\\
-1&0
  \end{pmatrix},~\pm \Im z>0,~ z\in U(1,r), \end{equation}
where $r\in (0, 1)$ is a small constant, the function $\zeta^{-\beta}$ takes  the principal branch, and the branch of $z^{(\alpha-\beta)/2}$ is chosen such that $\arg z\in (0, 2\pi)$.
Then, $E(z)$ is analytic for $z\in U(1,r)$ and  the matching condition  \eqref{eq:matching condition} is satisfied for $nt$ bounded.
\end{prop}
\begin{proof} It follows from the definition of $f(z)$ in \eqref{eq:mapping} that
$\left(\frac{f(z)}{z-1}\right)^{-\beta}=e^{\frac{ \pi i \beta}{2} } t^{\beta}+O(z-1)$ is analytic for $z\in U(1,r)$.
With the branch chosen, it is readily seen that $e^{\pm  {\pi i(\alpha-\beta)}/ 2   }  z^{ (\alpha-\beta)/{2}  }$
 is also analytic  for $z \in U(1,r)$ with $r<1$. Thus, $E(z)$ is  analytic     in $U(1,r)$.

Next, we show that $P^{(1)}(z)$ satisfies the  matching condition \eqref{eq:matching condition}. Since $nt$ is bounded, we see that $f(z)$ is large for $n$ large. Hence, for   $z\in \partial U(1,r)$, in view of \eqref{eq:parametrix} and   the large-$\zeta$ behavior \eqref{eq:Psi-infty} of $\Psi(\zeta; s)$, we obtain
\begin{equation}\label{eq:PN}
P^{(1)}(z)N(z)^{-1}=E(z)\left[I+O\left(  1 /n\right)\right ]f(z)^{-\beta\sigma_3}z^{-\frac{n}{2}\sigma_3}
e^{\pm \frac{\pi i(\alpha-\beta)}{2}\sigma_3}M(z)N(z)^{-1}, \quad \pm \Im z>0. \end{equation}
Recalling the definitions of $N(z)$ and $M(z)$, given  in \eqref{eq:N-solution} and \eqref{eq:M}, respectively, we have
\begin{equation}\label{eq:MN}
M(z)N(z)^{-1}=z^{\frac{n}{2}\sigma_3}e^{-\frac{\pi i}{2}(\alpha+\beta)\sigma_3}z^{\frac{1}{2}(\alpha-\beta)\sigma_3}(z-1)^{\beta\sigma_3}
\begin{pmatrix}
  0 & -1\\
1&0
  \end{pmatrix}, \end{equation}
where the branch for  $z^{ {n}/{2} }$  is chosen such that $\arg z\in (-\pi, \pi)$, the branch cut for $z^{ (\beta-\alpha)/2}$ is taken along the positive real axis so that $\arg z\in (0,2\pi)$ and
the branch cut for $(z-1)^{-\beta}$ is chosen along $[1,\infty )$ so that $\arg (z-1)\in (0,2\pi)$. Note that $|f(z)^{\beta}|=e^{-\frac{\pi i}{2} \beta}e^{i\beta \arg \ln z}$ is independent  of $t$, bounded and
bounded away from zero for pure imaginary parameter $\beta$ and  $z \in\partial U(1,r)$.
Substituting \eqref{eq:E} and \eqref{eq:MN}  into \eqref {eq:PN} gives us the matching condition \eqref{eq:matching condition}.
This completes the proof of the proposition.
\end{proof}

\subsection{Final transformation: $S \to R$}

Define
\begin{equation}\label{eq:R}
R(z)=\left\{ \begin{array}{ll}
                      S(z)N(z)^{-1}, & \hbox{$|z-1|>r$,} \\
                     S(z)P^{(1)}(z)^{-1}, & \hbox{$|z-1|<r$}.
                    \end{array}
\right.
 \end{equation}
 Then, $R(z)$ satisfies the RH problem.

\begin{figure}[ht]
 \begin{center}
  \includegraphics[height=8.5 cm]{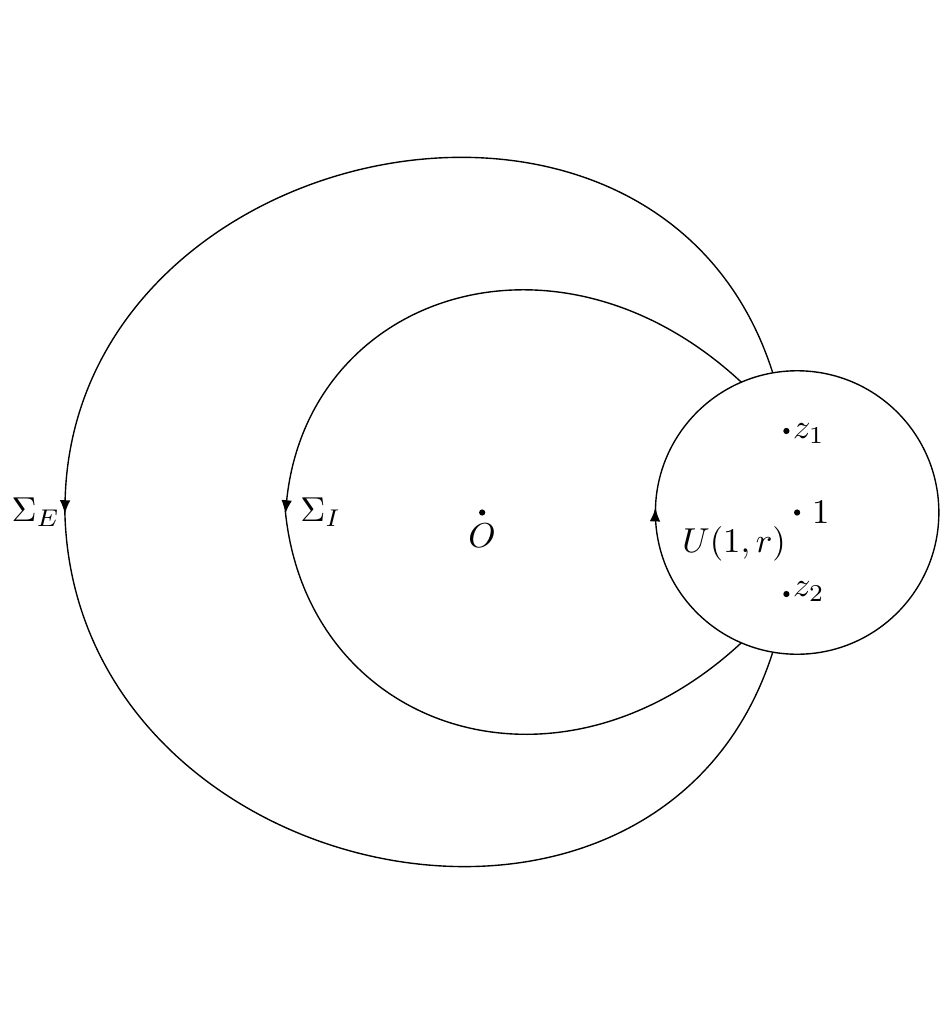}  \end{center}
  \caption{\small{Contours  for  the RH problem for  $R(z)$ }}
 \label{fig: contours for R}
\end{figure}

  \subsubsection*{RH problem for $R$}
\begin{description}
  \item (a)  $R(z)$ is analytic for  $C\setminus \Sigma_R$, where the remaining contour  $\Sigma_R$ is illustrated in Figure \ref{fig: contours for R}.
\item (b) $R(z)$  satisfies the   jump condition
  \begin{equation}\label{eq:R-jump}
R_{+}(z)=R_{-}(z)J_{R}(z), \quad z\in  \Sigma_R,\end{equation}
with
\begin{equation}\label{eq:J-R}
J_{R}(z)=\left\{ \begin{array}{ll}
                   P^{(1)}(z)N(z)^{-1}, & \hbox{$|z-1|=r$, clockwise,} \\
                     N(z)J_E(z)N(z)^{-1}, & \hbox{$z\in \Sigma_E$, $|z-1|>r$,}\\
                     N(z)J_I(z)N(z)^{-1}, & \hbox{$z\in \Sigma_I$, $|z-1|>r$,}
                    \end{array}
\right.
 \end{equation}   where $J_E(z)$ and $J_I(z)$ denote the jumps of $S(z)$ on $\Sigma_E$ and $\Sigma_I$, respectively; see \eqref{eq:S-jump} for the jump matrices and Figure \ref{fig: contours for R} for the contours.
\item (c)  As $z\to \infty$, we have
  \begin{equation}\label{eq:R-infty}
R(z)=I+O(1/z). \end{equation}
\end{description}

It is readily seen that $J_{R}(z)-I$ is uniformly exponentially small for $z$ on the portions of
 $\Sigma_E$ and $\Sigma_I$, described in \eqref{eq:J-R}, as $n\to\infty$.
 Therefore, for large $n$ and bounded $nt$, we have from the matching condition \eqref{eq:matching condition} that
 \begin{equation}\label{eq:R-jump-estimate}
J_{R}(z)=I+O(1/n),
\end{equation}
 uniformly for $z\in \Sigma_{R}$.
As before, after a standard analysis of the small-norm RH problems \cite{dkmv1}, a combination of  \eqref{eq:R-jump} and \eqref{eq:R-jump-estimate} gives
 \begin{equation}\label{eq:R-approx}
R(z)=I+O(1/n),  \quad \frac{d}{dz}R(z)=O(1/n),
\end{equation}
uniformly in the whole complex   $z$-plane.
\subsection{Proof of Theorem \ref{thm-gap-probability}}

As a consequence of the nonlinear steepest descent analysis of the RH problem for $Y$ performed in Section \ref{sec:Y},  and applying the differential identity \eqref{eq: diff identity}, we derive  the asymptotic approximation of the Toeplitz determinant as stated in Theorem \ref{thm-gap-probability}.

Tracing back the series of transformations $Y\to T\to S\to R$, we find that
 \begin{equation}\label{eq:Y-rep}
Y(z)=R(z) E(z)\Psi\left (f(z); -2int \right)e^{\pm \frac{\pi i(\alpha-\beta)}{2}\sigma_3}z^{\frac{n}{2}\sigma_3}w(z)^{-\frac{1}{2}\sigma_3},~~\pm\Im z>0,
\end{equation}
for  $|z|<1$ and $|z-1|<r$.
From Proposition \ref{pro:E}, we see that $E(z)$ is analytic and bounded for $|z-1|<r$.  We also have the analyticity of
$R(z)$ in $|z-1|<r$ and  the approximation \eqref{eq:R-approx}.
Thus, we have
 \begin{equation}\label{eq:Y-derivative}
z^{-n}w(z)\left(Y(z)^{-1}\frac{d}{dz} Y(z)\right)_{21}=f'(z)e^{\pm \pi i(\alpha-\beta)}
\left\{\Psi^{-1}(f(z); -2int )\Psi_{\zeta}(f(z); -2int )\right\}_{21}+O\left(\frac 1 n\right),
\end{equation}respectively for $\Im \pm z>0$,
with  $|z|<1$  and $|z-1|<r$.
The definition of $f(z)$ implies that $f'(z)=\frac 1{itz}$; cf. \eqref{eq:mapping}.
Substituting \eqref{eq:Y-derivative} into the differential identity \eqref{eq: diff identity}, we obtain
\begin{align}\label{eq: D-log-der}
    \frac d {dt}\ln D_{n}(t)=&-\frac{1}{2\pi i t}e^{ \pi i(\alpha-\beta)}
    \lim_{\zeta\to 1}\left(\Psi_{+}(\zeta; -2int )^{-1}\frac{d}{d\zeta}\Psi_{+}(\zeta; -2int )\right)_{21}\nonumber  \\
    &-\frac{1}{2\pi i t} e^{ -\pi i(\alpha-\beta)}
    \lim_{\zeta\to -1}\left(\Psi_{+}(\zeta; -2int )^{-1}\frac{d}{d\zeta}\Psi_{+}(\zeta; -2int )\right)_{21}+O\left(\frac 1 n\right )\nonumber  \\
=& -\frac{1}{2\pi i t} \left(e^{ \pi i(\alpha-\beta)}\Psi^{(1)}_1(-2int)+e^{ -\pi i(\alpha-\beta)}\Psi_1^{(2)}(-2int)\right)_{21}+O\left(\frac 1 n\right ),
\end{align}
where $\Psi^{(1)}_1(s)$ and $\Psi^{(2)}_1(s)$ are defined in \eqref{eq: Psi-1} and  \eqref{eq: Psi-1 minus}, respectively.
Comparing \eqref{eq:Psi-1-diff} with \eqref{eq:H-der}, we have
\begin{equation}\label{eq:Psi-1-H}
-\frac{1}{\pi i } \left(e^{ \pi i(\alpha-\beta)}\Psi^{(1)}_1(s)+e^{ -\pi i(\alpha-\beta)}\Psi_1^{(2)}(s)\right)_{21}=2sH(s)+c,\end{equation}
where $H(s)$ is the Hamiltonian for the coupled Painlev\'e V system and  $c$ is a constant independent of the variable $s$. Using the asymptotics as $is\to0^+$ in  \eqref{thm: H-asy-small s} and  \eqref{eq:Psi-1-1}, we have the constant $c=0$.
Substituting \eqref{eq:Psi-1-H} into  \eqref{eq: D-log-der}, we then obtain by taking integration
\begin{equation*}
\frac{D_n(t)}{D_n(0)}=\exp\left(\int_0^{-2int}H(\tau)d\tau+O(1/n)\right),\end{equation*}
where  the error bound is uniform for $nt$ bounded. Thus we obtain \eqref{eq:D-n-approx}, and complete the proof of Theorem \ref{thm-gap-probability}.

\section{Proof of Theorem \ref{thm: total integral}   }\label{sec:large-gap}

Integrating  the large-$s$ asymptotic approximation of $H(s;\alpha,\beta)$ in Theorem \ref{thm-asy}, we have as $is\to +\infty$
\begin{equation}\label{def: C-1}
\int_0^{s}H(\tau;\alpha,\beta)d\tau=\frac{s^2}{32}+\frac{i\alpha }{2}s-\left(\alpha^2-\beta^2+\frac{1}{4}\right)\ln|s|+C_1(\alpha,\beta)+O\left(\frac 1 s\right), \end{equation}
where the path of integration is  a line segment   of the negative imaginary axis.
Since $H(s)$ is free of poles for $s\in -i(0, +\infty)$ with at most weak singularity as $is\to 0^+$, as indicated in the asymptotic behavior   \eqref{thm: H-asy-small s},  the integral  \eqref{def: C-1}  is well-defined.

To determine the constant $C_1(\alpha,\beta)$,  which contains the global information of the Hamiltonian, the strategy is the following.
First, we express the integral alternatively
by using the differential identity \eqref{eq:total diff}.
The evaluation of the constant $C_1(\alpha,\beta)$ amounts to a calculation of   the constant term $C_2(\alpha,\beta)$ in the asymptotic expansion of the action integral. Next, it is much more convenient  to evaluate    $C_2(\alpha,\beta)$.
Actually,  by  using the differential identities  \eqref{eq:action-diff-a} and \eqref{eq:action-diff-b}, we see that each derivative  of the action integral, with respect to the parameters $\alpha$ or $\beta$, is a simple combination of the  functions $u_k(s)$ and $v_k(s)$, not involving    any integrals.  Thus, the derivatives of the constant term $C_2(\alpha,\beta)$ with respect to the parameters $\alpha$ and $\beta$  are then readily determined.  Finally, it is known that for special parameters $\alpha=1/2$ and $\beta=0$, the Hamiltonian admits a special function solution \eqref{eq: H-special solution} and its integral is explicitly calculated in \eqref{initial-D}.
This special case actually serves as an initial condition. The constant terms are thus determined, in terms of the Barnes $G$-function.

To integrate on both sides of the differential   identity \eqref{eq:total diff}, we first consider the following limit
 as $is\to 0^+$,
\begin{align}\label{def: anti-d-zero}
&\lim_{is\to 0^+}\left(sH(s;\alpha,\beta)+\alpha \ln d(s)-\beta\ln y(s)-2(\alpha^2-\beta^2)\ln s   \right)  \nonumber\\
&= \alpha\ln\left(\frac{\Gamma(1+\alpha+\beta)\Gamma(1+\alpha-\beta)}{\Gamma(1+2\alpha)^2}\right)
-\beta\ln\left(\frac{\Gamma(1+\alpha-\beta)}{\Gamma(1+\alpha+\beta)}\right)   \nonumber\\
&\quad+\alpha\ln(2\alpha)-2(\alpha^2-\beta^2)\ln2,  \end{align}
by using \eqref{eq: y-asy-small s} and \eqref{eq: d-asy-small s},
and  the asymptotics of  $H(s;\alpha,\beta)$ given in Theorem \ref{thm-asy}.
 Recalling \eqref{eq:total diff}, the convergence of the integral  \eqref{def: C-1} and the    boundedness near $0^+$ of the function in  \eqref{def: anti-d-zero} imply that  the action differential defined by
 \begin{equation*}
 u_1(s)\frac{dv_1(s)}{ds}+u_2(s)\frac{dv_2(s)}{ds}-H(s;\alpha,\beta)
 \end{equation*}
 is pole-free and integrable along $-i(0,\delta_1]$  for some small $\delta_1>0$.
 From the large-$s$ asymptotics of $H(s;\alpha,\beta)$, $u_k(s)$, $v_k(s)$ in Theorem \ref{thm-asy}, the action differential is also free of poles on $-i[M_1,+\infty)$ for   $M_1>0$ sufficiently large. Moreover, the action differential  has at most finite number of poles on $-i(\delta_1,M_1)$. Integrating  both sides of   \eqref{eq:total diff} on the interval $[0,s]$, and using the approximation     \eqref{def: anti-d-zero} as an initial value, we obtain
 \begin{align}\label{eq: H and action integral}
\int_0^{s}H(\tau;\alpha,\beta)d\tau=& \int_0^{s}\left[u_1(\tau)\frac{dv_1(\tau)}{d\tau}+u_2(\tau)\frac{dv_2(\tau)}{d\tau}-H(\tau;\alpha,\beta)\right ]d\tau\nonumber\\
&+sH(s;\alpha,\beta)+\alpha \ln d(s)-\beta\ln y(s)-2(\alpha^2-\beta^2)\ln s-C_0, \end{align}
where  $C_0$ is the limit in \eqref{def: anti-d-zero}, the path of integration of the first integral is an interval of the negative imaginary axis as described  in \eqref{def: C-1}. The integration contour of the second integral consists of the intervals $(0,-i\delta_1]$, $[-iM_1,s]$ and a curve $\Gamma$ connecting $\delta_1$ and $M_1$ in the right half complex plane $ \Re \tau>0$ to avoid the poles of the action differential and the Hamiltonian, and such that the region inside the closed curve $-i(\delta_1,M_1)\cup \Gamma^{-1}$ contains no poles of the action differential  and the Hamiltonian.
 From the large-$s$ asymptotic formulas \eqref{thm: u-1-asy-large s}, \eqref{thm: u-2-asy-large s}, \eqref{thm: H-asy-large s}, \eqref{eq: v-1-asy-large s} and  \eqref{eq: v-2-asy-large s},  we obtain
\begin{align}\label{def: C-2}
&\int_0^{s}\left[u_1(\tau)\frac{dv_1(\tau)}{d\tau}+u_2(\tau)\frac{dv_2(\tau)}{d\tau}-H(\tau;\alpha,\beta)\right ]d\tau\nonumber\\
&=-\frac{s^2}{32}-\frac{i\alpha s}{2}
+\left(\alpha^2-\beta^2+\frac{1}{4}-\frac{i(\gamma_1+\gamma_2)}{8}\right)\ln|s|+C_2(\alpha,\beta)+
O\left(\frac 1 s\right)  \end{align}
and
\begin{align}\label{def: anti-d-infty}
&sH(s;\alpha,\beta)+\alpha \ln d(s)-\beta\ln y(s)-2(\alpha^2-\beta^2)\ln s\nonumber\\
&= \frac{s^2}{16}+i \alpha s
-2(\alpha^2-\beta^2)\ln |s|-2(\alpha^2+\beta^2)\ln 2
-\left(\alpha^2-\beta^2+\frac{1}{4}\right)+\alpha\ln \alpha+O\left(\frac 1 s\right)
\end{align}
as $is\to+\infty$,  where the logarithms take the principal branches, the constant  $C_2(\alpha,\beta)$  depends  only on the parameters $\alpha$ and $\beta$, and the constants $\gamma_1$ and  $\gamma_2$ are the coefficients in the asymptotic expansions \eqref{eq: v-1-asy-large s} and \eqref{eq: v-2-asy-large s}, respectively.
In deriving  \eqref{def: anti-d-infty}, use has also been made of the large-$s$ asymptotic formulas \eqref{eq: y-asy-large s} and \eqref{eq: d-asy-large s}.
Substituting \eqref{def: C-1}, \eqref{def: C-2} and  \eqref{def: anti-d-infty} into \eqref{eq: H and action integral},  we obtain
\begin{align}\label{def: C-k relation}
C_1(\alpha,\beta)=&C_2(\alpha,\beta)-\left(\alpha^2-\beta^2+\frac{1}{4}\right )-\alpha\ln 2-4\beta^2\ln2\nonumber\\
&-\alpha\ln\left(\frac{\Gamma(1+\alpha+\beta)\Gamma(1+\alpha-\beta)}{\Gamma(1+2\alpha)^2}\right)+\beta\ln\left(\frac{\Gamma(1+\alpha-\beta)}{\Gamma(1+\alpha+\beta)}\right),
\end{align}
and the relation
\begin{equation}\label{eq: gamma-1+2}
\gamma_1+\gamma_2=-4i.\end{equation}

Next, we evaluate the constant $C_2(\alpha,\beta)$ by using the differential identities \eqref{eq:action-diff-a} and \eqref{eq:action-diff-b}.
From the relation \eqref{eq: gamma-1+2}  and
the large-$s$ asymptotics of $d(s)$, $u_k(s)$ and $v_k(s)$, given in \eqref{thm: u-1-asy-large s}, \eqref{thm: u-2-asy-large s}, \eqref{eq: d-asy-large s}, \eqref{eq: v-1-asy-large s} and \eqref{eq: v-2-asy-large s}, we have
\begin{equation}\label{def: anti-d-infty-2}
u_1(s) \frac{dv_1(s)}{d\alpha} +u_2 (s)\frac{dv_2(s)}{d\alpha} -\ln d(s) +2\alpha\ln s=-\frac{i}{2}s+2\alpha\ln|s|
+2\alpha\ln2
-\ln \alpha+O(1/s),
\end{equation}
and
\begin{equation}\label{def: anti-d-infty-3}
u_1(s) \frac{dv_1(s)}{d\beta} +u_2(s) \frac{dv_2(s)}{d\beta} +\ln y(s) -2\beta\ln s=-2\beta\ln|s|+2\beta\ln2+O(1/s),  \end{equation}
as $is\to +\infty$.
On the other hand, from the asymptotics of  $d(s)$, $y(s)$, $u_k(s)$ and $v_k(s)$ as $is\to0^+$, given in \eqref{eq: y-asy-small s}, \eqref{eq: d-asy-small s} and Theorem \ref{thm-asy}, we have
\begin{align}\label{def: anti-d-zero-2}
  &\lim_{is\to0^+}\left(u_1(s)\frac{dv_1(s)}{d\alpha}+u_2(s)\frac{dv_2(s)}{d\alpha}-\ln d(s)+2\alpha\ln s\right )\nonumber \\
   &  =2\alpha\ln 2-\ln \frac{2\alpha \Gamma(1+\alpha+\beta)\Gamma(1+\alpha-\beta)}{\Gamma(1+2\alpha)^2},
   \end{align}
and
\begin{equation}\label{def: anti-d-zero-3}
\lim_{is\to0^+}\left(u_1(s)\frac{dv_1(s)}{d\beta}+u_2(s)\frac{dv_2(s)}{d\beta}+\ln y(s)-2\beta\ln s\right)=
-2\beta\ln 2+\ln \frac{ \Gamma(1+\alpha-\beta)}{\Gamma(1+\alpha+\beta)}
. \end{equation}
Integrating   both sides of   \eqref{eq:action-diff-a} and \eqref{eq:action-diff-b} on the interval $[0,s]$ (the integration contour deformed if necessary, similar to  \eqref{eq: H and action integral}, to keep away from   the poles of the integrand), and using the approximations \eqref{def: C-2} and \eqref{def: anti-d-infty-2}-\eqref{def: anti-d-zero-3}, we obtain
the first-order differential equations
\begin{equation}\label{def: C-2-dff-a}
\frac{d}{d\alpha}C_2(\alpha,\beta)=\ln2+\ln\left(\frac{\Gamma(1+\alpha+\beta)\Gamma(1+\alpha-\beta)}{\Gamma(1+2\alpha)^2}\right),
\end{equation}
and
\begin{equation}\label{def: C-2-dff-b}
\frac{d}{d\beta}C_2(\alpha,\beta)=4\beta\ln2-\ln\left(\frac{\Gamma(1+\alpha-\beta)}{\Gamma(1+\alpha+\beta)}\right).\end{equation}
It follows from  \eqref{def: C-2-dff-a}  and  \eqref{def: C-2-dff-b} that
\begin{equation}\label{eq: C-2-0, alpha}
C_2(\alpha,\beta)=C_2(0,\beta)+\alpha\ln2+\int_{\beta}^{\alpha+\beta}\ln(\Gamma(1+\tau))d\tau+\int_{-\beta}^{\alpha-\beta}\ln(\Gamma(1+\tau))d\tau-\int_{0}^{2\alpha}\ln(\Gamma(1+\tau))d\tau,
\end{equation}
\begin{equation}\label{eq: C-2-0, beta}
C_2(0,\beta)=C_2(0,0)+2\beta^2\ln2+\int_0^{\beta}\ln(\Gamma(1+\tau))d\tau+\int_0^{-\beta}\ln(\Gamma(1+\tau))d\tau.
\end{equation}
Recalling the integral expression of   the Barnes $G$-function
\begin{equation}\label{def: Barnes-G}\ln G(1+z)=\frac{z}{2}\ln (2\pi)-\frac{z(z+1)}{2}+z\ln \Gamma(z+1)-\int_0^z \ln \Gamma(1+t)dt, \quad \Re z>-1;
\end{equation}cf., e.g., \cite[(2.37)]{DIK1},
we obtain from \eqref{def: C-k relation} and \eqref{eq: C-2-0, alpha} that
\begin{align}\label{eq: C-1-C-2-beta}
C_1(\alpha,\beta)=&C_2(0,\beta)-\ln\left(\frac{G(1+\alpha+\beta)G(1+\alpha-\beta)}{G(1+2\alpha)}\right)\nonumber\\
&-4\beta^2\ln2-\frac{1}{4}-\int_0^{\beta}\ln(\Gamma(1+\tau))d\tau-\int_0^{-\beta}\ln(\Gamma(1+\tau))d\tau.
\end{align}
This equation, together with \eqref{eq: C-2-0, beta}, implies that
\begin{equation}\label{eq: C-1-C-2-est}
C_1(\alpha,\beta)=C_2(0,0)-\ln\left(\frac{G(1+\alpha+\beta)G(1+\alpha-\beta)}{G(1+2\alpha)}\right)
-2\beta^2\ln2-\frac{1}{4}.
\end{equation}
In particular, we have
\begin{equation}\label{eq: C-1-special}
C_1(1/2,0)=C_2(0,0)-\ln\left(\frac{G(3/2)^2}{G(2)}\right)-\frac{1}{4}=C_2(0,0)-\ln \left(\pi G(1/2)^2\right)-\frac{1}{4}
\end{equation}
by using the recurrence relation
$$G(z+1)=\Gamma(z)G(z), \quad G(1)=1.$$
Comparing  \eqref{eq: C-1-C-2-est} and \eqref{eq: C-1-special}, we obtain
\begin{equation}\label{eq: C-1-expr}
C_1(\alpha,\beta)=C_1(1/2,0)+\ln\left(\frac{\pi G( {1}/{2})^2G(1+2\alpha)}{G(1+\alpha+\beta)G(1+\alpha-\beta)}\right)-2\beta^2\ln2.
\end{equation}
Substituting \eqref{eq: C-1-expr} into \eqref{def: C-1}, we obtain the large-$s$ asymptotic approximation of the integral of  the Hamiltonian $H(s)$
\begin{align}\label{eq: total integral of H}
\int_0^{s}H(\tau;\alpha,\beta)d\tau=&\frac{s^2}{32}+\frac{i\alpha s}{2}-\left(\alpha^2-\beta^2+\frac{1}{4}\right)\ln|s|+C_1(1/2,0)\nonumber\\
&+\ln\left(\frac{\pi G( {1}/{2})^2G(1+2\alpha)}{G(1+\alpha+\beta)G(1+\alpha-\beta)}\right)-2\beta^2\ln2+O\left(\frac 1 s\right), \end{align}
as $is\to+\infty$.

Finally,   we evaluate $C_1(1/2,0)$ by using the special function solution to the Hamiltonian for $\alpha=1/2$ and $\beta=0$.
From \eqref{initial-D}, we have
\begin{equation} \label{eq:total integral of H-special}
\int_0^{s}H(\tau;1/2,0)d\tau=\frac{s^2}{32}+  \frac{is}{4}-\frac{ \ln|s|}{2}
-\frac{ \ln(\pi/2)}{2}+O\left(\frac 1 s\right),\quad  is\to +\infty.
\end{equation}Here we have used the asymptotic formula $I_0(x)=\frac{e^x}{\sqrt{2\pi x}}\left\{1+O(1/x)\right\}$ as $x\to+\infty$; see \cite[(9.7.1)]{as}.
Comparing \eqref{eq: total integral of H} and \eqref{eq:total integral of H-special} yields
$$C_1( {1}/{2},0)=\frac{1}{2}\ln 2-\frac{1}{2}\ln \pi.$$
Thus, we obtain
\begin{align}\label{eq: total integral of H-done}
\int_0^{s}H(\tau;\alpha,\beta)d\tau=&\frac{s^2}{32}+\frac{i\alpha s}{2}-\left(\alpha^2-\beta^2+\frac{1}{4}\right)\ln(|s|/4)\nonumber\\
&+\ln\left(\frac{\sqrt{\pi}\; G( {1}/{2})^2\;G(1+2\alpha)}{2^{2\alpha^2}G(1+\alpha+\beta)\;G(1+\alpha-\beta)}\right)+O\left(\frac 1 s\right), \end{align}
as $is\to+\infty$. The equation  can be written equivalently in the form of \eqref{eq: total integral-H-thm}. This completes the proof of Theorem \ref{thm: total integral}.

\section*{Acknowledgements}

The work of Shuai-Xia Xu was supported in part by the National Natural Science Foundation of China under grant numbers 11571376 and 11971492. Yu-Qiu Zhao was supported in part by the National Natural Science Foundation of China under grant numbers 11571375 and 11971489.


\end{document}